\documentclass[11pt,a4paper,reqno]{amsart}
\pdfoutput=1
\usepackage{times}
\usepackage[margin=1.0in,tmargin=0.9in,bmargin=0.80in]{geometry}
\usepackage[dvipsnames]{xcolor}
\usepackage[english]{babel} 

\definecolor{bred}{rgb}{0.8,0,0}

\usepackage{amsmath,amssymb,graphicx,epsfig}
\usepackage{mathtools}
\usepackage{enumerate,amsthm,epstopdf}
\usepackage{enumitem}
\usepackage{cite}
\setlist[enumerate]{label={\upshape(\roman*)}}
\usepackage[flushleft]{threeparttable}
\usepackage{multirow}
\usepackage{longtable}
\usepackage[utf8]{inputenc} 
\usepackage[T1]{fontenc}    
\usepackage{bbm}
\usepackage{xargs}
\usepackage{latexsym}
\usepackage{wasysym}
\usepackage{esvect}
\usepackage{float}
\usepackage{hyperref}       
\usepackage{url}            
\usepackage{booktabs}       
\usepackage{amsfonts}       
\usepackage{nicefrac}       
\usepackage{microtype}      
\usepackage{cleveref}
\usepackage{dsfont}
\usepackage{ragged2e}
\usepackage[linesnumbered,ruled,vlined]{algorithm2e}
\usepackage{footmisc}
\usepackage{subcaption}
\usepackage[section]{placeins}
\usepackage[font=small,labelfont=bf]{caption}
\usepackage[textsize=footnotesize]{todonotes}
\hypersetup{colorlinks,linkcolor={blue},citecolor={bred},urlcolor={blue}}

\SetKwInput{KwInput}{Input}                
\SetKwInput{KwOutput}{Output}              

\DeclareMathOperator*{\trace}{tr}
\DeclareMathOperator*{\sym}{sym}
\DeclareMathOperator*{\diag}{diag}
\DeclareMathOperator*{\rand}{rand}

\newcommand{\Pol}{\textup{Pol}}

\DeclareMathOperator*{\linspan}{span}
\DeclareMathOperator{\im}{Im}
\DeclareMathOperator{\re}{Re}
\DeclareMathOperator{\vech}{vec}

\newtheorem{theorem}{Theorem}[section]
\newtheorem{proposition}[theorem]{Proposition}
\newtheorem{lemma}[theorem]{Lemma}
\newtheorem{corollary}[theorem]{Corollary}
\newtheorem{remark}[theorem]{Remark}
\newtheorem{definition}[theorem]{Definition}
\newtheorem{example}[theorem]{Example}
\newtheorem{assumption}[theorem]{Assumption}

\makeatletter
\newcommand{\algorithmstyle}[1]{\renewcommand{\algocf@style}{#1}}
\makeatother

\makeatletter
\newcommand{\labeltext}[3][]{%
	\@bsphack%
	\csname phantomsection\endcsname
	\def\tst{#1}%
	\def\labelmarkup{\emph}
	\def\refmarkup{}%
	\ifx\tst\empty\def\@currentlabel{\refmarkup{#2}}{\label{#3}}%
	\else\def\@currentlabel{\refmarkup{#1}}{\label{#3}}\fi%
	\@esphack%
	\labelmarkup{#2}
}
\makeatother

\usepackage{scalerel,stackengine}
\stackMath
\newcommand\reallywidehat[1]{%
	\savestack{\tmpbox}{\stretchto{%
			\scaleto{%
				\scalerel*[\widthof{\ensuremath{#1}}]{\kern.1pt\mathchar"0362\kern.1pt}%
				{\rule{0ex}{\textheight}}
			}{\textheight}%
		}{2.4ex}}%
	\stackon[-6.9pt]{#1}{\tmpbox}%
}
\newcommand\reallywidecheck[1]{%
	\savestack{\tmpbox}{\stretchto{%
			\scaleto{%
				\scalerel*[\widthof{\ensuremath{#1}}]{\kern-.6pt\bigwedge\kern-.6pt}%
				{\rule[-\textheight/2]{1ex}{\textheight}}
			}{\textheight}%
		}{0.5ex}}%
	\stackon[1pt]{#1}{\scalebox{-1}{\tmpbox}}%
}

\allowdisplaybreaks

\begin{document}

\title[]{Chaotic hedging with iterated integrals and neural networks}

\author[]{Ariel Neufeld}

\address{Nanyang Technological University, Division of Mathematical Sciences, 21 Nanyang Link, Singapore}
\email{ariel.neufeld@ntu.edu.sg}

\author[]{Philipp Schmocker}

\address{ETH Zurich, Department of Mathematics, R\"amistrasse 101, Zurich, Switzerland}
\email{philipp.schmocker@math.ethz.ch}

\date{\today}
\thanks{Financial support by the Nanyang Assistant Professorship Grant (NAP Grant) \emph{Machine Learning based Algorithms in Finance and Insurance} is gratefully acknowledged.}

\begin{abstract}
	In this paper, we derive an $L^p$-chaos expansion based on iterated Stratonovich integrals with respect to a given exponentially integrable continuous semimartingale. By omitting the orthogonality of the expansion, we show that every $p$-integrable functional, $p \in [1,\infty)$, can be approximated by a finite sum of iterated Stratonovich integrals. Using (possibly random) neural networks as integrands, we therefere obtain universal approximation results for $p$-integrable financial derivatives in the $L^p$-sense. Moreover, we can approximately solve the $L^p$-hedging problem (coinciding for $p = 2$ with the quadratic hedging problem), where the approximating hedging strategy can be computed in closed form within short runtime.
\end{abstract}

\maketitle

\vspace{-0.5cm}

\section{Introduction}

We address the problem of pricing and hedging a financial derivative in a market with finite time horizon $T > 0$ that consists of $d \in \mathbb{N}$ risky assets whose price processes are modelled by a $d$-dimensional continuous semimartingale $X := (X_t)_{t \in [0,T]}$. In order to derive our universal approximation results, we assume that $X$ is exponentially integrable and that both its finite variation part and its quadratic variation admit finite moments of all orders (see Assumption~\ref{AssExpInt}). This includes in particular affine and some polynomial diffusions as well as some stochastic volatility models (see \cite{duffie00,cuchiero11,filipovic16} and also Section~\ref{SecExpIntSemimg}+\ref{SecNumEx}).

In this setting, we aim to learn the payoff and the hedging strategy of a given financial derivative $G$, which is assumed to be a $p$-integrable functional$^\text{\footref{Footnote1}}$ of the underlying process $X$, where $p \in [1,\infty)$. In the special case of $p = 2$ and $X$ being a Brownian motion, the Wiener-Ito chaos decomposition (see \cite{wiener38,cameron47}) ensures that every square-integrable functional $G$ of the Brownian motion can be represented as infinite sum of orthogonal multiple integrals (see \cite[Theorem~4.2]{ito51}), which therefore yields a natural approximation of $G$ by truncating the infinite series. In turn, these multiple integrals can be rewritten as iterated Ito integrals with respect to that Brownian motion (see, e.g., \cite[p.~10]{dinunno08}).

In our paper, we omit the orthogonality of the chaos expansion and consider iterated Stratonovich integrals with respect to the given exponentially integrable continuous semimartingale $X$. By expressing monomials of $X$ as iterated Stratonovich integrals and using that polynomials of $X$ are $L^p$-dense due to the exponential integrability assumption on $X$, we obtain an $L^p$-chaos expansion with respect to $X$ (see Theorem~\ref{ThmChaos}). This extends classical chaos expansions in the literature (proven for the compensated Poisson process~\cite{ito56}, Az\'ema martingales~\cite{emery89}, some L\'evy processes via Teugel martingales~\cite{nualart01}, and other specific martingales in~\cite{jamshidian05,ditella16CRP}) from orthogonal $L^2$-decompositions using iterated \emph{Ito integrals} to $L^p$-denseness of iterated \emph{Stratonovich integrals}. Nevertheless, despite working with Stratonovich integrals, we can use the $L^p$-chaos expansion to approximately solve the $L^p$-hedging problem (see Theorem~\ref{ThmLpHedging}), where the wealth process is formulated using classical Ito integration, as standard in the financial context. In the case of $p = 2$, this coincides with the quadratic hedging problem in \cite{schweizer99,pham00}.

The machine learning application consists of using (possibly random) neural networks to learn either the deterministic integrands of the iterated Stratonovich integrals in the $L^p$-chaos expansion (see Theorem~\ref{ThmUAT}) or the deterministic integrands appearing in the $L^p$-hedging problem (see Theorem~\ref{ThmLpHedgingNN}). By the universal approximation property of neural networks (see \cite{cybenko89,hornik91,pinkus99,neufeld24}), these integrands can be approximated arbitrarily well in a suitable function space, providing a tractable framework for constructing approximate financial derivatives and hedging strategies via machine learning techniques. Previously, \cite{bgtw19} applied neural networks to learn optimal hedging strategies in a discrete-time setting, followed by other successful neural networks applications in mathematical finance (see, e.g., \cite{han17,ruf19,sirignano19,ckt20,eckstein20,neufeld21,sester22,schmocker22}).

In addition, we extend the universal approximation result to random neural networks with randomly initialized weights and biases (see Theorem~\ref{ThmRandUAT}), which are inspired by the works on extreme learning machines and random feature learning (see, e.g., \cite{huang06,rahimi07,grigoryeva18,gonon20,gonon21}). In this case, only the linear readouts needs to be trained, which can be efficiently performed, e.g., by the least squares method.

Furthermore, the \emph{signature} has been successfully applied in \cite{perez19,lyons20,cartea22,primavera22,cuchiero23,bayer25} to learn path-dependent functionals. Originating from rough path theory (see \cite{lyons98,friz10,friz20}), the signature satisfies a similar universality property than neural networks, but on the corresponding path space. Since the (Stratonovich) signature of a semimartingale is equal to iterated Stratonovich integrals with constant integrand equal to one, our results can be interpreted as $L^p$-universality of the signature.

\vspace{-0.1cm}

\subsection{Overview of the main results}

Our first main result is a chaos expansion of $L^p(\Omega,\mathcal{F}_T,\mathbb{P})$ into iterated Stratonovich integrals (see Theorem~\ref{ThmChaos}): Given $p \in [1,\infty)$ and\footnote{\label{Footnote1}Here, $G \in L^p(\Omega,\mathcal{F}_T,\mathbb{P})$ is $\mathcal{F}_T$-measurable, where $\mathcal{F}_T$ is defined as the $\mathbb{P}$-completion of $\sigma(\lbrace X_t: t \in [0,T] \rbrace)$.} $G \in L^p(\Omega,\mathcal{F}_T,\mathbb{P})$, there exists for every $\varepsilon > 0$ some $N \in \mathbb{N}$ and deterministic functions $g_n: [0,T]^n \rightarrow (\mathbb{R}^d)^{\otimes n}$, $n = 0,...,N$, such that
\vspace{-0.25cm}
\begin{equation}
	\label{EqIntroChaos}
	\left\Vert G - \sum_{n=0}^N J^\circ_n(g_n) \right\Vert_{L^p(\mathbb{P})} < \varepsilon,
	\vspace{-0.02cm}
\end{equation}
where $J^\circ_n(g_n)_T$ denotes the $n$-fold iterated Stratonovich integral of $g_n$ (see Definition~\ref{DefItInt}). This allows us to approximate any given financial derivative $G \in L^p(\Omega,\mathcal{F}_T,\mathbb{P})$ by finitely many iterated Stratonovich integrals. In order to prove~\eqref{EqIntroChaos}, we express monomials of $X$ as iterated Stratonovich integrals and use that polynomials of $X$ are dense in $L^p(\Omega,\mathcal{F}_T,\mathbb{P})$ due to the exponential integrability assumption on $X$.

Moreover, by using the $L^p$-chaos expansion in \eqref{EqIntroChaos} together with a denseness result for predictable processes, we obtain our second main result, which approximately solves the $L^p$-hedging problem (see Theorem~\ref{ThmLpHedging}): Given $p \in [1,\infty)$ and $G \in L^p(\Omega,\mathcal{F}_T,\mathbb{P})$, there exists for every $\varepsilon > 0$ some $g_0 \in \mathbb{R}$, $N,m_1,...,m_N \in \mathbb{N}$, and functions $g_{n,j,0},g_{n,j,1}: [0,T] \rightarrow \mathbb{R}^d$, $n = 1,...,N$ and $j = 1,...,m_n$, such that
\begin{equation}
	\vspace{-0.05cm}
	\label{EqIntroLpHedging}
	\left\Vert G - g_0 - \int_0^T \left( \vartheta^{g_{1:N}}_t \right)^\top dX_t \right\Vert_{L^p(\mathbb{P})} \leq \varepsilon + \inf_{(c,\theta) \in \mathbb{R} \times \Theta^p(X)} \left\Vert G - c - \int_0^T \theta_t^\top dX_t \right\Vert_{L^p(\mathbb{P})}
	\vspace{-0.05cm}
\end{equation}
with $\vartheta^{g_{1:N}}_t := \sum_{n=1}^N \sum_{j=1}^{m_n} \frac{W(g_{n,j,0})_t^{n-1}}{(n-1)!} g_{n,j,1}(t)$, where $W(g_{n,j,0})_t := \int_0^t g_{n,j,0}(s) dX_s$ and $\Theta^p(X)$ is the space of predictable $\mathbb{R}^d$-valued processes satisfying some integrability conditions (see \eqref{EqDefThetapNorm}). For $p = 2$, the minimization of the $L^p$-hedging error $\big\Vert G - c - \int_0^T \theta_t^\top dX_t \big\Vert_{L^p(\mathbb{P})}$ in \eqref{EqIntroLpHedging} coincides with the quadratic hedging problem in \cite{schweizer99,pham00}. In addition, note that \eqref{EqIntroLpHedging}, unlike \eqref{EqIntroChaos}, uses classical Ito integration.

Furthermore, by using the universal approximation property of (possibly random) neural networks, we obtain the machine learning results of this paper, which replace the deterministic functions in \eqref{EqIntroChaos}+\eqref{EqIntroLpHedging} with (possibly random) neural networks (see Theorem~\ref{ThmUAT}+\ref{ThmRandUAT} for \eqref{EqIntroChaos} and Theorem~\ref{ThmLpHedgingNN}+\ref{ThmLpHedgingRN} for \eqref{EqIntroLpHedging}). In particular, for random neural networks, only the linear readout needs to be trained, which means that \eqref{EqIntroChaos}+\eqref{EqIntroLpHedging} can be solved efficiently by interpreting them as linear regression tasks (see also Algorithm~\ref{Alg}).

\vspace{-0.1cm}

\subsection{Outline}

In Section~\ref{SecExpIntSemimg}, we introduce the main setting including exponentially integrable continuous semimartingale. In Section~\ref{SecStochIntItInt}, we define iterated Stratonovich integrals. In Section~\ref{SecChaosHedg}, we show an $L^p$-chaos expansion using Stratonovich integrals and approximately solve the $L^p$-hedging problem. In Section~\ref{SecUATs}, we derive universal approximation results using (possibly random) neural networks, which are applied in Section~\ref{SecNumEx} to numerically solve the $L^2$-hedging problem. Finally, Section~\ref{SecProofs} contains all proofs.

\vspace{-0.1cm}

\subsection{Notation}

As usual, $\mathbb{N} = \lbrace 1,2,...\rbrace$ and $\mathbb{N}_0 = \mathbb{N} \cup \lbrace 0 \rbrace$ denote the sets of natural numbers, while $\mathbb{R}$ and $\mathbb{C}$ represent the real and complex numbers (with complex unit $\mathbf{i} := \sqrt{-1}$), where $s \wedge t := \min(s,t)$, for $s,t \in \mathbb{R}$. For $d \in \mathbb{N}$, we denote by $\mathbb{R}^d$ the Euclidean space equipped with $\Vert x \Vert = \big( \sum_{i=1}^d x_i^2 \big)^{1/2}$, where $e_i \in \mathbb{R}^d$ is the $i$-th unit vector. For $d,l \in \mathbb{N}$, we denote by $\mathbb{R}^{d \times l}$ the vector space of matrices $A = (a_{i,j})_{i=1,...,d,\,j=1,...,l} \in \mathbb{R}^{d \times l}$ equipped with the Frobenius norm $\Vert A \Vert_F = \big( \sum_{i=1}^d \sum_{j=1}^l a_{i,j}^2 \big)^{1/2}$. If $d=l$, we denote by $\mathbb{S}^d_+ \subset \mathbb{R}^{d \times d}$ the cone of symmetric non-negative definite matrices, denote by $I_d \in \mathbb{S}^d_+$ the identity matrix, and define the trace $\trace(A) := \sum_{i=1}^d a_{i,i}$ for $A = (a_{i,j})_{i,j=1,...,d} \in \mathbb{R}^{d \times d}$.

Furthermore, for $k,d \in \mathbb{N}$ and $E \subseteq \mathbb{R}^k$, we denote by $C(E;\mathbb{R}^d)$ the Banach space of continuous functions $f: E \rightarrow \mathbb{R}^d$ equipped with the supremum norm $\Vert f \Vert_\infty := \sup_{t \in E} \Vert f(t) \Vert$, where we abbreviate $C(E) := C(E;\mathbb{R})$. In addition, for $p \in [1,\infty)$ and a measure space $(S,\Sigma,\mu)$, we denote by $L^p(S,\Sigma,\mu;\mathbb{R}^d)$ the Banach space of (equivalence classes of) $\Sigma/\mathcal{B}(\mathbb{R}^d)$-measurable functions $f: S \rightarrow \mathbb{R}^d$ with finite norm $\Vert f \Vert_{L^p(\mu)} := \big( \int_S \Vert f(t) \Vert^p \mu(dt) \big)^{1/p}$, where we abbreviate $L^p(S,\Sigma,\mu) := L^p(S,\Sigma,\mu;\mathbb{R})$. Moreover, if $S$ is a topological space, we denote by $\mathcal{B}(S)$ the Borel $\sigma$-algebra of $S$.

\section{Exponentially integrable semimartingales}
\label{SecExpIntSemimg}

Throughout this paper, we fix a finite time horizon $T > 0$ and a probability space $(\Omega,\mathcal{A},\mathbb{P})$. For $d \in \mathbb{N}$, the financial market is modelled by a $d$-dimensional stochastic process $X := \big( X^1_t,...,X^d_t \big)_{t \in [0,T]}^\top$ that is assumed to be a continuous semimartingale with respect to the usual $\mathbb{P}$-augmented filtration\footnote{\label{FootnoteAugmFilt}The usual $\mathbb{P}$-augmented filtration generated by $X$ is defined as the smallest filtration $\mathbb{F} := (\mathcal{F}_t)_{t \in [0,T]}$ such that $\mathbb{F}$ is complete with respect to $(\Omega,\mathcal{A},\mathbb{P})$ and right-continuous on $[0,T)$, and $X$ is $\mathbb{F}$-adapted (see \cite[p.~45]{revuz99}).} generated by $X$, i.e., there exists a continuous $\mathbb{F}$-adapted process $A := (A_t)_{t \in [0,T]}$ of finite variation and a continuous local martingale $M := (M_t)_{t \in [0,T]}$ such that for every $t \in [0,T]$ it holds that $X_t = X_0 + A_t + M_t$.

Hereby, we denote by $\langle X \rangle := (\langle X^i, X^j \rangle_t)_{t \in [0,T], \, i,j = 1,...,d}$ the $\mathbb{S}^d_+$-valued quadratic variation of $X$ satisfying $\langle X \rangle = \langle M \rangle$, whereas $\int_0^T \Vert dA_t \Vert$ represents the total variation of $A$ over $[0,T]$.

\begin{assumption}
	\label{AssExpInt}
	Let $X$ be a continuous semimartingale with decomposition $X_t = X_0 + A_t + M_t$, $t \in [0,T]$, into a $\mathbb{P}$-a.s.~constant initial value $X_0 \in \mathbb{R}^d$, a continuous $\mathbb{F}$-adapted process $A = (A_t)_{t \in [0,T]}$ of finite variation with $A_0 = 0$, and a continuous local martingale $M = (M_t)_{t \in [0,T]}$ with $M_0 = 0$. Moreover, assume the following integrability conditions:
	\begin{enumerate}
		\item\label{AssExpInt1} There exists $\varepsilon > 0$ such that for every $t \in [0,T]$ it holds that $\mathbb{E}\left[ \exp\left( \varepsilon \Vert X_t \Vert \right) \right] < \infty$.
		\item\label{AssExpInt2} For every $\gamma \in [1,\infty)$ it holds that $\mathbb{E}\big[ \big( \int_0^T \Vert dA_t \Vert \big)^\gamma \big] < \infty$.
		\item\label{AssExpInt3} For every $\gamma \in [1,\infty)$ it holds that $\mathbb{E}\big[ \Vert \langle M \rangle_T \Vert_F^\gamma \big] < \infty$.
	\end{enumerate}
\end{assumption}

\begin{remark}
	Condition~\ref{AssExpInt1} is necessary for polynomials of $X$ to be dense in $L^p(\Omega,\mathcal{F}_T,\mathbb{P})$ (see Proposition~\ref{PropPolDense}), whereas \ref{AssExpInt2}+\ref{AssExpInt3} ensure the inequality in Lemma~\ref{LemmaLpNorm}~\ref{LemmaLpNorm2}. Note that $\mathbb{E}\big[ \exp\big( \varepsilon \int_0^T \Vert dA_t \Vert \big) \big] < \infty$ and $\mathbb{E}\big[ \exp\big( \varepsilon \Vert \langle M \rangle_T \Vert_F \big) \big] < \infty$ for some $\varepsilon > 0$ are sufficient conditions for \ref{AssExpInt2} and \ref{AssExpInt3}, respectively.
\end{remark}

Let us give some examples of continuous semimartingales satisfying Assumption~\ref{AssExpInt}, which includes affine diffusions and some polynomial diffusions (see \cite{duffie00,cuchiero11,filipovic16}). To this end, we denote by $\Pol_n(\mathbb{R}^d)$ the set of polynomial functions $\mathbb{R}^d \ni x := (x_1,...,x_d) \mapsto \sum_{\alpha \in \mathbb{N}^d_0, \, \vert\alpha\vert \leq n} c_\alpha x_1^{\alpha_1} \cdots x_d^{\alpha_d} \in \mathbb{R}$, with some coefficients $c_\alpha \in \mathbb{R}$, where $\vert \alpha \vert = \alpha_1 + ... + \alpha_n$ for $\alpha := (\alpha_1,...,\alpha_d) \in \mathbb{N}^d_0$.

\begin{example}[Polynomial diffusions]
	\label{ExPolynDiff}
	For a subset $E \subseteq \mathbb{R}^d$ and some functions $a := (a_{i,j})_{i,j=1,...,d}: E \rightarrow \mathbb{S}^d_+$ and $b := (b_1,...,b_d)^\top: E \rightarrow \mathbb{R}^d$ with $a_{i,j} \in \Pol_2(\mathbb{R}^d)$ and $b_i \in \Pol_1(\mathbb{R}^d)$ for all $i,j = 1,...,d$, let here $X = (X_t)_{t \in [0,T]}$ be an $E$-valued strong solution of the stochastic differential equation (SDE)
	\begin{equation}
		\label{EqExPolynDiff0a}
		dX_t = b(X_t) dt + \sigma(X_t) dB_t, \quad\quad t \in [0,T],
	\end{equation}
	with initial value $X_0 \in E$, where $\sigma: E \rightarrow \mathbb{R}^{d \times d}$ satisfies $\sigma(x) \sigma(x)^\top = a(x)$ for all $x \in E$, and where $B := (B_t)_{t \in [0,T]}$ is a $d$-dimensional Brownian motion. Then, $X$ is a continuous semimartingale (with finite variation part $t \mapsto A_t := \int_0^t b(X_s) ds$ and local martingale part $t \mapsto M_t := \int_0^t \sigma(X_s) dB_s$), called a polynomial diffusion in the sense of \cite[Definition~2.1]{filipovic16}. Moreover, by assuming that the diffusion coefficient $a: E \rightarrow \mathbb{S}^d_+$ satisfies a linear growth condition, i.e.~there exists $C_1 > 0$ such that for every $x \in E$ it holds that
	\begin{equation}
		\label{EqExPolynDiff0b}
		\Vert a(x) \Vert_F \leq C_1 \left( 1 + \Vert x \Vert \right),
	\end{equation}
	then one can show that $X$ satisfies Assumption~\ref{AssExpInt}. The proof can be found in Section~\ref{SecProofsExpIntSemimg}.
	
	For $d = 1$, we observe that $X = (X_t)_{t \in [0,T]}$ is an $E$-valued strong solution of the SDE
	\begin{equation*}
		dX_t = (\beta_0 + \beta_1 X_t) dt + \sqrt{\alpha_0 + \alpha_1 X_t + \alpha_2 X_t^2} dB_t, \quad\quad t \in [0,T],
	\end{equation*}
	with initial value $X_0 \in E$, where $\beta_0, \beta_1, \alpha_0, \alpha_1, \alpha_2 \in \mathbb{R}$ are such that $\sqrt{\alpha_0 + \alpha_1 x + \alpha_2 x^2} \geq 0$ for all $x \in E$. Hence, if the drift vanishes (i.e.~$\beta_0 = \beta_1 = 0$), there are in particular three relevant cases:
	\begin{subequations}
		\begin{align}
		\label{EqPolynDiff1}
		& E = \mathbb{R} & & \text{with } \alpha_0 > 0 \text{, } \alpha_1 = 0 \text{ and } \alpha_2 \geq 0, & & \text{e.g., Brownian motion (BM),} \\
		\label{EqPolynDiff2}
		& E = [0,\infty) & & \text{with } \alpha_0 = 0 \text{ and } \alpha_1, \alpha_2 \geq 0, & & \text{e.g., geometric BM (GBM)}, \\
		\label{EqPolynDiff3}
		& E = [0,1] & & \text{with } \alpha_0 = 0 \text{ and } \alpha_1 = -\alpha_2, & & \text{e.g., Jacobi process.}
		\end{align}
	\end{subequations}
	If $\alpha_2 = 0$ or $E$ is compact (e.g., \eqref{EqPolynDiff3}), then the diffusion coefficient $E \ni x \mapsto a(x) := \alpha_0 + \alpha_1 x + \alpha_2 x^2 \in [0,\infty)$ satisfies the linear growth condition in \eqref{EqExPolynDiff0b} and $X$ thus satisfies Assumption~\ref{AssExpInt}.
\end{example}

Moreover, we can also consider financial market models $X$ with stochastic volatility such that $X$ satisfies Assumption~\ref{AssExpInt} (see also the numerical experiments in Section~\ref{SecNumEx}).

\section{Stochastic integral and iterated Stratonovich integrals}
\label{SecStochIntItInt}

In this section, we introduce the stochastic integral of a deterministic function with respect to the given continuous semimartingale $X := (X_t)_{t \in [0,T]}$ having semimartingale decomposition $X_t = X_0 + A_t + M_t$, $t \in [0,T]$. Subsequently, we introduce iterated Stratonovich integrals that are defined by iteration.

\vspace{-0.1cm}

\subsection{Stochastic integral}
\label{SecStochInt}

In order to define the stochastic integral of a deterministic function, we fix some $p \in [1,\infty)$. Then, we denote by $L^p(X)$ the vector space of (equivalence classes of) c\`agl\`ad (i.e., left continuous on $(0,T]$ with right limits on $[0,T)$) functions $g: [0,T] \rightarrow \mathbb{R}^d$ such that
\vspace{-0.05cm}
\begin{equation}
	\label{EqDefLpNorm}
	\Vert g \Vert_{L^p(X)} := \mathbb{E}\left[ \left( \int_0^T \left\vert g(t)^\top dA_t \right\vert \right)^p \right]^\frac{1}{p} + \mathbb{E}\left[ \left( \int_0^T g(t)^\top d\langle M \rangle_t g(t) \right)^\frac{p}{2} \right]^\frac{1}{p} < \infty.
	\vspace{-0.05cm}
\end{equation}

\begin{remark}
	\label{RemBoundedMbl}
	Since for every $g \in L^p(X)$ the function $t \mapsto g(T-t)$ is c\`adl\`ag (right continuous on $[0,T)$ with left limits on $(0,T]$), we conclude from \cite[p.~122]{billingsley99} that $g: [0,T] \rightarrow \mathbb{R}^d$ is $\mathcal{B}([0,T])/\mathcal{B}(\mathbb{R}^d)$-measurable and bounded, i.e.~$\Vert g \Vert_\infty := \sup_{t \in [0,T]} \Vert g(t) \Vert < \infty$.
\end{remark}

We show that $L^p(X)$ is a normed vector space under the norm $\Vert \cdot \Vert_{L^p(X)}$ and obtain an estimate for $L^p(X)$-functions. The proofs of the results in this section can be found in Section~\ref{SecProofsStochInt}.

\begin{lemma}
	\label{LemmaLpNorm}
	Let $X$ be a continuous semimartingale and let $p \in [1,\infty)$. Then, the following holds true:
	\begin{enumerate}
		\item\label{LemmaLpNorm1} $(L^p(X),\Vert \cdot \Vert_{L^p(X)})$ is a normed vector space.
		\item\label{LemmaLpNorm2} If Assumption~\ref{AssExpInt} additionally holds, then there exists a constant $C_{p,T,X} > 0$ such that for every $g \in L^p(X)$ it holds that $\Vert g \Vert_{L^p(X)} \leq C_{p,T,X} \Vert g \Vert_\infty < \infty$.
	\end{enumerate}
\end{lemma}

\vspace{-0.05cm}

Now, we are able to introduce the stochastic integral as operator, which is defined on this $L^p(X)$-space.

\begin{definition}
	Let $X$ be a continuous semimartingale and let $p \in [1,\infty)$. Then, for every $t \in [0,T]$, we introduce the \emph{stochastic integral (up to time $t$)} as operator $W(\cdot)_t: L^p(X) \rightarrow L^p(\Omega,\mathcal{F}_t,\mathbb{P})$ defined by
	\vspace{-0.05cm}
	\begin{equation*}
		L^p(X) \ni g \quad \mapsto \quad W(g)_t := \int_0^t g(s)^\top dX_s \in L^p(\Omega,\mathcal{F}_t,\mathbb{P}).
		\vspace{-0.05cm}
	\end{equation*}
\end{definition}

\begin{remark}
	\label{RemStochInt}
	For every fixed $g \in L^p(X)$ and $t \in [0,T]$, we first observe that $W(g)_t = W(\mathds{1}_{[0,t]} g)_T$. Moreover, we show in Lemma~\ref{LemmaBDG} below that the operator $W(\cdot)_t: L^p(X) \rightarrow L^p(\Omega,\mathcal{F}_t,\mathbb{P})$ is well-defined and bounded. In addition, $W(g)_t = \int_0^t g(s)^\top dX_s$ is well-defined as stochastic integral since the integrand $s \mapsto g(s)$ is left-continuous and $\mathbb{F}$-adapted, thus $\mathbb{F}$-predictable and locally bounded.
\end{remark}

Moreover, we extend the BDG inequality to continuous semimartingales by using the norm in \eqref{EqDefLpNorm}. 

\begin{lemma}
	\label{LemmaBDG}
	Let $X$ be a continuous semimartingale and let $p \in [1,\infty)$. Then, there exists a constant $C_{1,p} > 0$ such that for every $g \in L^p(X)$ it holds that
	\vspace{-0.1cm}
	\begin{equation*}
		\mathbb{E}\left[ \sup_{t \in [0,T]} \left\vert W(g)_t \right\vert^p \right]^\frac{1}{p} \leq C_{1,p} \Vert g \Vert_{L^p(X)}.
		\vspace{-0.05cm}
	\end{equation*}
	Hereby, $C_{1,p} := \max(1, C_p) > 0$ consists of the constant $C_p > 0$ appearing in the classic upper Burkholder-Davis-Gundy (BDG) inequality with exponent $p \in [1,\infty)$.
\end{lemma}

\vspace{-0.1cm}

\subsection{Iterated Stratonovich integrals}
\label{SecItInt}

Let us first briefly revisit the notion of the tensor product (see \cite[Chapter~1]{ryan02}). For any $n \in \mathbb{N}_0$, we define $(\mathbb{R}^d)^{\otimes n} := \mathbb{R}^d \otimes \cdots \otimes \mathbb{R}^d$, with convention $(\mathbb{R}^d)^{\otimes 0} \cong \mathbb{R}$. Moreover, for any $p \in [1,\infty)$, we define $L^{np}(X)^{\otimes n} := L^{np}(X) \otimes \cdots \otimes L^{np}(X)$, with convention $L^{np}(X)^{\otimes n} := \mathbb{R}$ for $n = 0$. Then, every tensor $g \in L^{np}(X)^{\otimes n}$ can be seen as linear functional acting on multilinear forms $L^{np}(X) \times \cdots \times L^{np}(X) \rightarrow \mathbb{R}$ (see \cite[Section~1.1]{ryan02}), which can be expressed as
\vspace{-0.05cm}
\begin{equation*}
	g = \sum_{j=1}^m g_{j,1} \otimes \cdots \otimes g_{j,n} \,\, \in \,\, L^{np}(X)^{\otimes n} := L^{np}(X) \otimes \cdots \otimes L^{np}(X)
	\vspace{-0.05cm}
\end{equation*}
for some $m \in \mathbb{N}$ and $g_{j,1},...,g_{j,n} \in L^{np}(X)$, $j = 1,...,m$, where the representation might not be unique (see \cite[Section~1.1]{ryan02}). On the vector space $L^{np}(X)^{\otimes n}$, we use the projective tensor norm defined by
\vspace{-0.05cm}
\begin{equation*}
	\Vert g \Vert_{L^{np}(X)^{\otimes n}} = \inf\left\lbrace \sum_{j=1}^m \prod_{k=1}^n \Vert g_{j,k} \Vert_{L^{np}(X)}: \,\,
	\begin{matrix*}[l]
		g = \sum_{j=1}^m g_{j,1} \otimes \cdots \otimes g_{j,n} \\
		m \in \mathbb{N}, \, g_{j,1},...,g_{j,n} \in L^{np}(X)
	\end{matrix*}
	\right\rbrace
	\vspace{-0.05cm}
\end{equation*}
and show some properties of $\Vert \cdot \Vert_{L^{np}(X)^{\otimes n}}$. The proofs of this section are given in Section~\ref{SecProofsItInt}.

\begin{lemma}
	\label{LemmaLpXnNorm}
	For every $n \in \mathbb{N}$ and $p \in [1,\infty)$, $(L^{np}(X)^{\otimes n},\Vert \cdot \Vert_{L^{np}(X)^{\otimes n}})$ is a normed vector space. Moreover, for every $g_1,...,g_n \in L^{np}(X)$, it holds that $\left\Vert g_1 \otimes \cdots \otimes g_n \right\Vert_{L^{np}(X)^{\otimes n}} = \prod_{k=1}^n \Vert g_k \Vert_{L^{np}(X)}$.
\end{lemma}

Note that for $n = 1$, we recover $L^p(X)$ from Section~\ref{SecStochInt}, i.e.~$L^p(X)^{\otimes 1} \cong L^p(X)$. Now, we can introduce iterated Stratonovich integrals with respect to $X$ and show that they are well-defined.

\begin{definition}
	\label{DefItInt}
	Let $X$ be a continuous semimartingale, let $n \in \mathbb{N}$, and $p \in [1,\infty)$. Then, for every $t \in [0,T]$, we introduce the \emph{$n$-fold iterated Stratonovich integral} as operator $J^\circ_n(\cdot)_t: L^{np}(X)^{\otimes n} \rightarrow L^p(\Omega,\mathcal{F}_t,\mathbb{P})$ defined by $J^\circ_0(g)_t = g$, for $g \in \mathbb{R}$, and for $n \geq 1$ by\footnote{For a sufficiently integrable $\mathbb{F}$-predictable process $\theta := (\theta_t)_{t \in [0,T]}$ and an $\mathbb{F}$-adapted process $Y := (Y_t)_{t \in [0,T]}$, the Stratonovich integral $\int_0^t \theta_s \!\circ\! dY_s$ is defined as the $L^2(\mathbb{P})$-limit of $\sum_{k=1}^n \frac{1}{2} (\theta_{t_k} \!+\! \theta_{t_{k+1}}) (Y_{t_{k+1}} \!-\! Y_{t_k})$ over partitions of the form $0 \leq t_0 < t_1 < ... < t_n \leq T$, as the mesh $\max_{k=1,...,n} \vert t_k - t_{k-1} \vert$ goes to zero (see, e.g., \cite[Chapter~V.5]{protter05}).}
	\vspace{-0.05cm}
	\begin{equation*}
		J^\circ_n(g)_t := \sum_{j=1}^m \int_0^t \Bigg( \int_0^{t_{n-1}} \cdots \Bigg( \int_0^{t_2} \circ \underbrace{dW(g_{j,1})_{t_1}}_{=g_{j,1}(t_1)^\top dX_{t_1}} \Bigg) \cdots \circ \hspace{-0.18cm} \underbrace{dW(g_{j,n-1})_{t_{n-1}}}_{=g_{j,n-1}(t_{n-1})^\top dX_{t_{n-1}}} \Bigg) \circ \underbrace{dW(g_{j,n})_{t_n}}_{=g_{j,n}(t_n)^\top dX_{t_n}},
		\vspace{-0.1cm}
	\end{equation*}
	for $g = \sum_{j=1}^m g_{j,1} \otimes \cdots \otimes g_{j,n} \in L^{np}(X)^{\otimes n}$.
\end{definition}

\begin{remark}
	\label{RemItInt}
	While we prove in Lemma~\ref{LemmaItIntLinearBDG} below that $J^\circ_n(\cdot): L^{np}(X)^{\otimes n} \rightarrow L^p(\Omega,\mathcal{F}_t,\mathbb{P})$ is well-defined, linear, and bounded, we first show by induction on $n \in \mathbb{N}$ that $J^\circ_n(g_{j,1} \otimes \cdots \otimes g_{j,n})_t$ is well-defined as iterated Stratonovich integrals, for all $g_{j,1},...,g_{j,n} \in L^{np}(X)$ and $t \in [0,T]$. Indeed, for $n = 1$, we observe that $J^\circ_1(g_1)_t = \int_0^t \circ dW(g_1)_s = W(g_1)_t$ is by Remark~\ref{RemStochInt} well-defined as stochastic integral. Now, if $J^\circ_{n-1}(g_{j,1} \otimes \cdots \otimes g_{j,n-1})_s$ is well-defined, for all $s \in [0,t]$ and some $n \in \mathbb{N} \cap [2,\infty)$, then
	\vspace{-0.05cm}
	\begin{equation}
		\label{EqRemItInt1}
		\begin{aligned}
			& J^\circ_n(g_{j,1} \otimes \cdots \otimes g_{j,n})_t = \int_0^t J^\circ_{n-1}(g_{j,1} \otimes \cdots \otimes g_{j,n-1})_s \circ dW(g_{j,n})_s \\
			& = \int_0^t J^\circ_{n-1}(g_{j,1} \otimes \cdots \otimes g_{j,n-1})_s dW(g_{j,n})_s + \frac{1}{2} \langle J^\circ_{n-1}(g_{j,1} \otimes \cdots \otimes g_{j,n-1}), W(g_{j,n}) \rangle_t \\
			& = \int_0^t J^\circ_{n-1}(g_{j,1} \otimes \cdots \otimes g_{j,n-1})_s g_{j,n}(s)^\top dX_s + \frac{1}{2} \int_0^t J^\circ_{n-2}(g_{j,1} \otimes \cdots \otimes g_{j,n-2})_s g_{j,n-1}(s)^\top d\langle X \rangle_s g_{j,n}(s)
			\vspace{-0.05cm}
		\end{aligned}
	\end{equation}
	is also well-defined as stochastic integral (resp.~Lebesgue-Stieltjes integral), since the integrand $s \mapsto J^\circ_{n-1}(g_{j,1} \otimes \cdots \otimes g_{j,n-1})_s g_n(s)$ is left-continuous and $\mathbb{F}$-adapted, thus $\mathbb{F}$-predictable and locally bounded.
\end{remark}

\begin{lemma}
	\label{LemmaItIntLinearBDG}
	Let $X$ be a continuous semimartingale, let $n \in \mathbb{N}$, $p \in [1,\infty)$, and $t \in [0,T]$. Then, the following holds true:
	\begin{enumerate}
		\item\label{LemmaItIntLinearBDG1} The operator $J^\circ_n(\cdot)_t: L^{np}(X)^{\otimes n} \rightarrow L^p(\Omega,\mathcal{F}_t,\mathbb{P})$ is linear.
		\item\label{LemmaItIntLinearBDG2} $J^\circ_n(g)_t$ does (up to $\mathbb{P}$-null sets) not depend on the representation of $g \in L^{np}(X)^{\otimes n}$, i.e.~if $g \in L^{np}(X)^{\otimes n}$ has representations $g^{(1)} = \sum_{j=1}^{m_1} g^{(1)}_{j,1} \otimes \cdots \otimes g^{(1)}_{j,n} \in L^{np}(X)^{\otimes n}$ and $g^{(2)} = \sum_{j=1}^{m_2} g^{(2)}_{j,1} \otimes \cdots \otimes g^{(2)}_{j,n} \in L^{np}(X)^{\otimes n}$, then $J^\circ_n\big(g^{(1)}\big)_t = J^\circ_n\big(g^{(2)}\big)_t$, $\mathbb{P}$-a.s.
		\item\label{LemmaItIntLinearBDG3} There exists a constant $C_{n,p} > 0$ such that for every $g \in L^{np}(X)^{\otimes n}$, it holds that
		\vspace{-0.1cm}
		\begin{equation*}
			\mathbb{E}\left[ \sup_{t \in [0,T]} \left\vert J^\circ_n(g)_t \right\vert^p \right]^\frac{1}{p} \leq C_{n,p} \Vert g \Vert_{L^{np}(X)^{\otimes n}},
			\vspace{-0.1cm}
		\end{equation*}
		with $C_{n,p} := \frac{3}{2} \prod_{k=1}^n \max\big( 1, C_\frac{np}{k} \big) > 0$, where $C_r > 0$ denotes the constant of the classic upper Burkholder-Davis-Gundy inequality with exponent $r \in [1,\infty)$.
	\end{enumerate} 
\end{lemma}

Note that \ref{LemmaItIntLinearBDG1}+\ref{LemmaItIntLinearBDG2} follow from standard tensor product results, whereas \ref{LemmaItIntLinearBDG3} is obtained by repeated application of the BDG-type inequality in Lemma~\ref{LemmaBDG} after adding the Stratonovich correction as in \eqref{EqRemItInt1}.

Moreover, we compute iterated Stratonovich integrals of symmetric tensors. To this end, we consider the vector subspace of symmetric tensors defined by
\vspace{-0.03cm}
\begin{equation}
	\label{EqDefLpSym}
	L^{np}_{\sym}(X)^{\otimes n} := \linspan\left\lbrace g \in L^{np}(X)^{\otimes n}: \sym(g) = g \right\rbrace \subseteq L^{np}(X)^{\otimes n},
	\vspace{-0.03cm}
\end{equation}
where $\sym(g) := \frac{1}{n!} \sum_{\sigma \in \mathcal{S}_n} \sum_{j=1}^m g_{j,\sigma(1)} \otimes \cdots \otimes g_{j,\sigma(n)}$ for $g = \sum_{j=1}^m g_{j,1} \otimes \cdots \otimes g_{j,n} \in L^{np}(X)^{\otimes n}$, with $\mathcal{S}_n$ denoting the set of permutations $\sigma: \lbrace 1,...,n \rbrace \rightarrow \lbrace 1,...,n \rbrace$.

\begin{proposition}
	\label{PropMon}
	Let $X$ be a continuous semimartingale, $n \in \mathbb{N}_0$, and $p \in [1,\infty)$. Then, for every $g =$ $\sum_{j=1}^m g_{j,1} \otimes \cdots \otimes g_{j,n} \in L^{np}_{\sym}(X)^{\otimes n}$ and $t \in [0,T]$, we have $J^\circ_n(g)_t = \frac{1}{n!} \sum_{j=1}^m \prod_{k=1}^n W(g_{j,k})_t$, $\mathbb{P}$-a.s.
\end{proposition}

\section{Chaos expansion and $L^p$-hedging}
\label{SecChaosHedg}

In this section, we show a chaos expansion of $L^p(\Omega,\mathcal{F}_T,\mathbb{P})$ into iterated Stratonovich integrals with respect to a given continuous semimartingale $X := \big(X^1_t,...,X^d_t\big)_{t \in [0,T]}$ satisfying Assumption~\ref{AssExpInt}. 

As a consequence, every financial derivative $G \in L^p(\Omega,\mathcal{F}_T,\mathbb{P})$ can be approximated by finitely many iterated Stratonovich integrals. By combining this representation with a denseness result for $\mathbb{F}$-predictable processes, we further show that we can approximately solve the $L^p$-hedging problem.

\subsection{Chaos expansion with iterated Stratonovich integrals}
\label{SecChaos}

First, we prove that polynomials of $X$ are dense in $L^p(\Omega,\mathcal{F}_T,\mathbb{P})$. To this end, we use the exponential integrability of $X$ in Assumption~\ref{AssExpInt}, which is similar to \cite[Section~5]{jamshidian05}. The proofs of the results in this section can be found in Section~\ref{SecProofsChaos}.

\begin{proposition}
	\label{PropPolDense}
	Let $X$ be a continuous semimartingale satisfying Assumption~\ref{AssExpInt} and $p \in [1,\infty)$. Then,
	\vspace{-0.25cm}
	\begin{equation}
		\label{EqPropPolDense1}
		\Pol(X) := \linspan\left\lbrace \left( X^{i_1}_{t_1} \right)^{k_1} \cdots \left( X^{i_m}_{t_m} \right)^{k_m}: \,\,
		\begin{matrix}
			m \in \mathbb{N}, \, (i_1,...,i_m) \in \lbrace 1,...,d \rbrace^m, \\
			(k_1,...,k_m) \in \mathbb{N}_0^m, \, (t_1,...,t_m) \in [0,T]^m
		\end{matrix}
		\right\rbrace
	\end{equation}
	is dense in $L^p(\Omega,\mathcal{F}_T,\mathbb{P})$.
\end{proposition}

For the following chaos expansion, we now apply Proposition~\ref{PropPolDense} to show that the direct sum\footnote{The direct sum of $(\mathcal{Z}_n)_{n \in \mathbb{N}_0} \subseteq L^p(\Omega,\mathcal{F}_T,\mathbb{P})$ is defined as $\bigoplus_{n=0}^\infty \mathcal{Z}_n := \big\lbrace \sum_{n \in N} Z_n: N \subseteq \mathbb{N}_0 \text{ finite}, \, Z_n \in \mathcal{Z}_n \big\rbrace$.} $\bigoplus_{n \in \mathbb{N}_0} \big\lbrace J^\circ_n(g_n)_T: g_n \in L^{np}_{\diag}(X)^{\otimes n} \big\rbrace$ is dense in $L^p(\Omega,\mathcal{F}_T,\mathbb{P})$, where we denote by
\vspace{-0.1cm}
\begin{equation*}
	L^{np}_{\diag}(X)^{\otimes n} := \linspan\left\lbrace g_0^{\otimes n} := g_0 \otimes \cdots \otimes g_0: g_0 \in L^{np}(X) \right\rbrace \subseteq L^{np}_{\sym}(X)^{\otimes n} \subseteq L^{np}(X)^{\otimes n}
	\vspace{-0.05cm}
\end{equation*}
the subspace of diagonal tensors. The idea of the proof is the following. For simple functions of the form $[0,T] \ni t \mapsto g_\lambda(t) := \sum_{l=1}^m \lambda_l \mathds{1}_{[0,t_l]}(t) e_{i_l} \in \mathbb{R}^d$, with $m \in \mathbb{N}$, $(i_1,...,i_m) \in \lbrace 1,...,d \rbrace^m$, $(k_1,...,k_m) \in \mathbb{N}_0^m$, and $(t_1,...,t_m) \in [0,T]^m$, it follows from Proposition~\ref{PropMon} that
\vspace{-0.15cm}
\begin{equation*}
	J^\circ_n\left(g_\lambda^{\otimes n}\right) = \frac{W(g_\lambda)_T^n}{n!} = \frac{\left( \sum_{l=1}^m \lambda_l \big( X^{i_l}_{t_l} - X^{i_l}_0 \big) \right)^n}{n!}.
	\vspace{-0.05cm}
\end{equation*}
Hence, for every continuous linear functional $l \in L^p(\Omega,\mathcal{F}_T,\mathbb{P})^*$ vanishing on $\bigoplus_{n \in \mathbb{N}_0} \big\lbrace J^\circ_n(g_n)_T: g_n \in L^{np}_{\diag}(X)^{\otimes n} \big\rbrace$, we apply a dominated convergence argument to obtain that $l$ also vanishes on $e^{{W(g_\lambda)_T}} = \sum_{n=1}^\infty \frac{W(g_\lambda)_T^n}{n!}$. Since the partial derivatives of $e^{{W(g_\lambda)_T}}$ with respect to $\lambda$ yield the polynomials of $X$, we thus conclude that $l$ also vanishes on $\Pol(X)$. Finally, by using that $\Pol(X)$ is by Proposition~\ref{PropPolDense} dense in $L^p(\Omega,\mathcal{F}_T,\mathbb{P})$ together with the Hahn-Banach theorem, we obtain the chaos expansion.

\begin{theorem}[Chaos expansion with iterated Stratonovich integrals]
	\label{ThmChaos}
	Let $X$ be a continuous semimartingale satisfying Assumption~\ref{AssExpInt} and let $p \in [1,\infty)$. Then, the following holds true:
	\begin{enumerate}
		\item\label{ThmChaos1} For $n,m_n \in \mathbb{N}$ and $g_n \!:=\! \sum_{j=1}^{m_n} g_{n,j}^{\otimes n} \in L^{np}_{\diag}(X)^{\otimes n}$ with $g_{n,j} \!\in\! L^{np}(X)$, $j \!=\! 1,...,m_n$, we have
		\vspace{-0.05cm}
		\begin{equation*}
			J^\circ_n(g_n)_T = \int_0^T \left( \vartheta^{g_n}_t \right)^\top dX_t + \int_0^T \trace\left( \eta^{g_n}_t d\langle X \rangle_t \right), \quad\quad \mathbb{P}\text{-a.s.},
			\vspace{-0.1cm}
		\end{equation*}
		where the $\mathbb{R}^d$-valued process $\vartheta^{g_n} := \big( \vartheta^{g_n}_t \big)_{t \in [0,T]}$ and the $\mathbb{R}^{d \times d}$-valued process $\eta^{g_n} := \big( \eta^{g_n}_t \big)_{t \in [0,T]}$ are for every $t \in [0,T]$ given by
		\vspace{-0.05cm}
		\begin{equation*}
			\quad\quad\quad \vartheta^{g_n}_t := \sum_{j=1}^{m_n} \frac{W(g_{n,j})_t^{n-1}}{(n-1)!} g_{n,j}(t), \quad\quad \text{and} \quad\quad \eta^{g_n}_t := \frac{1}{2} \sum_{j=1}^{m_n} \frac{W(g_{n,j})_t^{n-1}}{(n-1)!} g_{n,j}(t) g_{n,j}(t)^\top.
			\vspace{-0.1cm}
		\end{equation*}
		\item\label{ThmChaos2} The direct sum $\bigoplus_{n \in \mathbb{N}_0} \big\lbrace J^\circ_n(g_n)_T: g_n \in L^{np}_{\diag}(X)^{\otimes n} \big\rbrace$ is dense in $L^p(\Omega,\mathcal{F}_T,\mathbb{P})$, i.e.~for every $G \in L^p(\Omega,\mathcal{F}_T,\mathbb{P})$ and $\varepsilon > 0$ there exist $N \in \mathbb{N}$ and $g_n \in L^{np}_{\diag}(X)^{\otimes n}$, $n = 0,...,N$, such that
		\vspace{-0.15cm}
		\begin{equation*}
			\left\Vert G - \sum_{n=0}^N J^\circ_n(g_n)_T \right\Vert_{L^p(\mathbb{P})} < \varepsilon.
			\vspace{-0.1cm}
		\end{equation*}
	\end{enumerate}
\end{theorem}

\begin{remark}
	Theorem~\ref{ThmChaos} shows that iterated Stratonovich integrals are dense in $L^p(\Omega,\mathcal{F}_T,\mathbb{P})$ instead of an orthogonal decompositions of iterated Ito integrals in $L^2(\Omega,\mathcal{F}_T,\mathbb{P})$ proven in \cite{wiener38,ito51,emery89,nualart01,ditella16CRP}.
\end{remark}

\subsection{$L^p$-hedging}
\label{SecLpHedging}

In this section, we consider the following $L^p$-hedging problem, where $p \in [1,\infty)$. For a given financial derivative $G \in L^p(\Omega,\mathcal{F}_T,\mathbb{P})$, we aim to find an optimal initial endowment $c \in \mathbb{R}$ and an optimal trading strategy $\theta := (\theta_t)_{t \in [0,T]}$ such that the $L^p$-hedging error $\Vert G - c - \int_0^T \theta_t^\top dX_t \big\Vert_{L^p(\mathbb{P})}$ \hfill is

\pagebreak

\noindent minimized. For $p = 2$, this problem coincides with the quadratic hedging approaches in \cite{schweizer99,pham00}, where sufficient conditions can be found for the existence of an optimal pair $c \in \mathbb{R}$ and $\theta := (\theta_t)_{t \in [0,T]}$.

We denote by $\Theta^p(X)$ the space of $\mathbb{F}$-predictable $\mathbb{R}^d$-valued processes $\theta := (\theta_t)_{t \in [0,T]}$ with
\vspace{-0.13cm}
\begin{equation}
	\label{EqDefThetapNorm}
	\Vert \theta \Vert_{\Theta^p(X)} := \sum_{i=1}^d \mathbb{E}\left[ \left( \int_0^T \left\vert \theta^i_t dA^i_t \right\vert \right)^p \right]^\frac{1}{p} + \sum_{i=1}^d \mathbb{E}\left[ \left( \int_0^T \left( \theta^i_t \right)^2 d\langle M^i \rangle_t \right)^\frac{p}{2} \right]^\frac{1}{p} < \infty.
	\vspace{-0.05cm}
\end{equation}
Note that the Kunita-Watanabe inequality and the Cauchy-Schwarz inequality ensure that for every $\theta \in \Theta^p(X)$ it holds that $\mathbb{E}\big[ \big( \int_0^T \big\vert \theta_t^\top dA_t \big\vert \big)^p \big]^{1/p} + \mathbb{E}\big[ \big( \int_0^T \theta_t^\top d\langle M \rangle_t \theta_t \big)^{p/2} \big]^{1/p} \leq \Vert \theta \Vert_{\Theta^p(X)}$.

Then, by combining the chaos expansion (Theorem~\ref{ThmChaos}) with a denseness result for $\mathbb{F}$-predictable processes, we are able to approximately solve the $L^p$-hedging problem. The proof is given in Section~\ref{SecProofsLpHedging}.

\begin{theorem}[$L^p$-hedging]
	\label{ThmLpHedging}
	Let $X$ be a continuous semimartingale satisfying Assumption~\ref{AssExpInt} and let $p \in [1,\infty)$. Then, for every $G \in L^p(\Omega,\mathcal{F}_T,\mathbb{P})$ and $\varepsilon > 0$ there exist some $g_0 \in \mathbb{R}$ and $N,m_n \in \mathbb{N}$ as well as $g_{n,j,0},g_{n,j,1} \in L^{np}(X)$, $n = 1,...,N$ and $j = 1,...,m_n$, such that
	\vspace{-0.13cm}
	\begin{equation*}
		\left\Vert G - g_0 - \int_0^T \left( \vartheta^{g_{1:N}}_t \right)^\top dX_t \right\Vert_{L^p(\mathbb{P})} \leq \varepsilon + \inf_{(c,\theta) \in \mathbb{R} \times \Theta^p(X)} \left\Vert G - c - \int_0^T \theta_t^\top dX_t \right\Vert_{L^p(\mathbb{P})},
		\vspace{-0.1cm}
	\end{equation*} 
	where the $\mathbb{R}^d$-valued process $\vartheta^{g_{1:N}} := \big( \vartheta^{g_{1:N}}_t \big)_{t \in [0,T]}$ is for every $t \in [0,T]$ given by
	\vspace{-0.15cm}
	\begin{equation*}
		\vartheta^{g_{1:N}}_t := \sum_{n=1}^N \sum_{j=1}^{m_n} \frac{W(g_{n,j,0})_t^{n-1}}{(n-1)!} g_{n,j,1}(t).
		\vspace{-0.1cm}
	\end{equation*}
\end{theorem}

In the following, we use (possibly random) neural networks to learn the integrands $g_n \in L^{np}(X)^{\otimes n}$, $n = 0,...,N$, in the chaos expansion (Theorem~\ref{ThmChaos}) and the functions $g_{n,j,0},g_{n,j,1} \in L^{np}(X)$, $n = 1,...,N$ and $j = 1,...,m_n$ in the approximation of the $L^p$-hedging problem (Theorem~\ref{ThmLpHedging}).

\section{Universal approximation of financial derivatives}
\label{SecUATs}

In this section, we approximate any given financial derivative $G \in L^p(\Omega,\mathcal{F}_T,\mathbb{P})$ in two steps: first by using the chaos expansion (Theorem~\ref{ThmChaos}) and then by replacing the deterministic integrands by (possibly random) neural networks (Theorem~\ref{ThmUAT}+\ref{ThmRandUAT}). We also refer to Algorithm~\ref{Alg} for the procedure.

\subsection{Fully trained neural networks}
\label{SecUAT}

Inspired by the functionality of a human brain, neural networks were introduced in \cite{mcculloch43}. Here, we generalize them to our tensor-valued setting.

\begin{definition}
	For $p \in [1,\infty)$ and $n,m \in \mathbb{N}$, a diagonal tensor $\varphi = \sum_{j=1}^m y_j \varphi_j^{\otimes n} \in L^{np}_{\diag}(X)^{\otimes n}$ is called a \emph{(tensor-valued) neural network} if for every $j = 1,...,m$ it holds that
	\vspace{-0.1cm}
	\begin{equation}
		\label{EqDefNN1}
		[0,T] \ni t \quad \mapsto \quad \varphi_j(t) = \rho(a_j t + b_j) \in \mathbb{R}^d,
		\vspace{-0.1cm}
	\end{equation}
	for some \emph{activation function} $\rho \in C(\mathbb{R})$ that is applied componentwise in \eqref{EqDefNN1}. Hereby, $a_1,...,a_m \in \mathbb{R}^d$ are \emph{weights}, $b_1,...,b_m \in \mathbb{R}^d$ are \emph{biases}, and $y_1,...,y_m \in \mathbb{R}$ are \emph{linear readouts}. Moreover, we define $\mathcal{NN}^\rho_{d,n}$ as the set of all tensor-valued neural networks, whereas for $n = 0$ we set $\mathcal{NN}^\rho_{d,0} := \mathbb{R}$.
\end{definition}

\begin{remark}
	Since $\rho \in C(\mathbb{R})$ is continuous, Lemma~\ref{LemmaLpNorm}~\ref{LemmaLpNorm2} ensures that $\varphi_j \in L^p(X)$ for all $j = 1,...,m$. Therefore, every tensor-valued neural network $\varphi \in \mathcal{NN}^\rho_{d,n}$ is a diagonal tensor in $L^{np}_{\diag}(X)^{\otimes n}$.
\end{remark}

Now, we use the universal approximation property of neural networks (see, e.g., \cite{cybenko89,hornik91,pinkus99,neufeld24}) to approximate any diagonal tensor in $L^{np}_{\text{diag}}(X)^{\otimes n}$ by a tensor-valued neural network. To this end, we assume that the activation function $\rho \in C(\mathbb{R})$ is non-polynomial, i.e.~$\rho \in C(\mathbb{R})$ is almost everywhere algebraically not equal to a polynomial (see also \cite{leshno93}). The proof can be found in Section~\ref{SecProofsUAT}.

\begin{proposition}
	\label{PropUAT}
	Let $X$ be a continuous semimartingale satisfying Assumption~\ref{AssExpInt}, let $n \in \mathbb{N}_0$, $p \in [1,\infty)$, and let $\rho \in C(\mathbb{R})$ be non-polynomial. Then, $\mathcal{NN}^\rho_{d,n}$ is dense in $L^{np}_{\diag}(X)^{\otimes n}$, i.e.~for every $g \in L^{np}_{\diag}(X)^{\otimes n}$ and $\varepsilon > 0$ there exists some $\varphi \in \mathcal{NN}^\rho_{d,n}$ such that $\Vert g - \varphi \Vert_{L^{np}(X)^{\otimes n}} < \varepsilon$.
\end{proposition}

Hence, by using the chaos expansion (Theorem~\ref{ThmChaos}) together with Proposition~\ref{PropUAT}, we obtain the following universal approximation theorem for financial derivatives. The proof is given in Section~\ref{SecProofsUAT}.

\begin{theorem}[Universal approximation]
	\label{ThmUAT}
	Let $X$ be a continuous semimartingale satisfying Assumption~\ref{AssExpInt}, let $p \in [1,\infty)$, and let $\rho \in C(\mathbb{R})$ be non-polynomial. Then, $\bigoplus_{n \in \mathbb{N}_0} \big\lbrace J^\circ_n(\varphi_n)_T: \varphi_n \in \mathcal{NN}^\rho_{d,n} \big\rbrace$ is dense in $L^p(\Omega,\mathcal{F}_T,\mathbb{P})$, i.e.~for every $G \in L^p(\Omega,\mathcal{F}_T,\mathbb{P})$ and $\varepsilon > 0$ there exists some $N \in \mathbb{N}$ and $\varphi_n \in \mathcal{NN}^\rho_{d,n}$, $n = 0,...,N$, such that $\big\Vert G - \sum_{n=0}^N J^\circ_n(\varphi_n)_T \big\Vert_{L^p(\mathbb{P})} < \varepsilon$.
\end{theorem}

\subsection{Random neural networks}
\label{SecRandUAT}

Inspired by the works on extreme learning machines and random feature learning (see \cite{huang06,rahimi07,grigoryeva18,gonon20,gonon21,schmocker23}), we now consider random neural networks, instead of fully trained neural networks, where the weights and biases are now randomly initialized. In this case, only the linear readout needs to be trained, which can be performed efficiently, e.g., by the least squares method.

To this end, we impose the following condition on the activation function $\rho \in C(\mathbb{R})$ and the random initialization of the weights and biases $\widetilde{a}_1,\widetilde{b}_1: \widetilde{\Omega} \rightarrow \mathbb{R}^d$ defined on another probability space $(\widetilde{\Omega},\widetilde{\mathcal{A}},\widetilde{\mathbb{P}})$.

\begin{assumption}
	\label{AssCondCDF}
	Let $\widetilde{a}_1,\widetilde{b}_1: \widetilde{\Omega} \rightarrow \mathbb{R}^d$ be two random vectors defined on a (possibly different) probability space $(\widetilde{\Omega},\widetilde{\mathcal{A}},\widetilde{\mathbb{P}})$ such that for every $(a,b) \in \mathbb{R}^d \times \mathbb{R}^d$ and $r > 0$ it holds that $\widetilde{\mathbb{P}}[\lbrace \widetilde{\omega} \in \widetilde{\Omega}: \Vert (\widetilde{a}_1(\widetilde{\omega}),\widetilde{b}_1(\widetilde{\omega})) - (a,b) \Vert < r \rbrace] > 0$. Moreover, let $\rho \in C(\mathbb{R})$ such that for every $q \in [1,\infty)$ we have
	\begin{equation*}
		\widetilde{\mathbb{E}}\left[ \left\Vert \rho\left( \widetilde{a}_1 \cdot + \widetilde{b}_1 \right) \right\Vert_\infty^q \right] := \widetilde{\mathbb{E}}\left[ \sup_{t \in [0,T]} \left\Vert \rho\left( \widetilde{a}_1 t + \widetilde{b}_1 \right) \right\Vert^q \right] < \infty.
	\end{equation*}
\end{assumption}

\begin{example}
	Assumption~\ref{AssCondCDF} is satisfied for example if $\rho \in C(\mathbb{R})$ is bounded (e.g., sigmoid function $\rho(s) = \frac{1}{1+\exp(-s)}$) or if there exist $C,c > 0$ such that $\vert \rho(s) \vert \leq C \left( 1 + \vert s \vert^c \right)$, for all $s \in \mathbb{R}$, and $\widetilde{a}_1$ and $\widetilde{b}_1$ have moments of all orders (e.g., ReLU function $\rho(s) = \max(s,0)$ and normally distributed $\widetilde{a}_1, \widetilde{b}_1$).
\end{example}

Moreover, we assume that $(\widetilde{\Omega},\widetilde{\mathcal{A}},\widetilde{\mathbb{P}})$ supports two sequences of independent and identically distributed (i.i.d.)~random variables $(\widetilde{a}_j)_{j \in \mathbb{N}} \sim \widetilde{a}_1$ and $(\widetilde{b}_j)_{j \in \mathbb{N}} \sim \widetilde{b}_1$ used for the random weights and biases, and define the $\sigma$-algebra $\widetilde{\mathcal{F}}_{\rand} := \sigma(\lbrace \widetilde{a}_j, \widetilde{b}_j: j \in \mathbb{N} \rbrace)$. Now, we can introduce random neural networks.

\begin{definition}
	For $p \in [1,\infty)$ and $n,m \in \mathbb{N}$, we call $\widetilde{\Omega} \ni \widetilde{\omega} \mapsto \widetilde{\varphi}(\widetilde{\omega}) := \sum_{j=1}^m \widetilde{y}_j(\widetilde{\omega}) \varphi_j(\widetilde{\omega})^{\otimes n} \in L^{np}_{\diag}(X)^{\otimes n}$ a \emph{(tensor-valued) random neural network} if for every $\widetilde{\omega} \in \widetilde{\Omega}$ and $j = 1,...,m$ it holds that
	\begin{equation}
		\label{EqDefRandNN1}
		[0,T] \ni t \quad \mapsto \quad \widetilde{\varphi}_j(\widetilde{\omega})(t) = \rho\left( \widetilde{a}_j(\widetilde{\omega}) t + \widetilde{b}_j(\widetilde{\omega}) \right) \in \mathbb{R}^d,
	\end{equation}
	for some \emph{activation function} $\rho \in C(\mathbb{R})$ that is applied componentwise in \eqref{EqDefRandNN1}. Hereby, $(\widetilde{a}_j)_{j \in \mathbb{N}} \overset{i.i.d.}{\sim} \widetilde{a}_1$ are the \emph{random weights}, $(\widetilde{b}_j)_{j \in \mathbb{N}} \overset{i.i.d.}{\sim} \widetilde{b}_1$ are the \emph{random biases}, and $\widetilde{y}_1,...,\widetilde{y}_m: \widetilde{\Omega} \rightarrow \mathbb{R}$ are the \emph{linear readouts} that are assumed to be $\widetilde{\mathcal{F}}_{\rand}/\mathcal{B}(\mathbb{R})$-measurable and $\widetilde{\mathbb{P}}$-a.s.~bounded. Moreover, we denote by $\mathcal{RN}^\rho_{d,n}$ the set of tensor-valued random neural networks, whereas for $n = 0$ we set $\mathcal{RN}^\rho_{d,0} := \mathbb{R}$.
\end{definition}

In Lemma~\ref{LemmaRandNNWellDef}, we show that every random neural network $\widetilde{\varphi} \in \mathcal{RN}^\rho_{d,n}$ is well-defined in the Bochner space $L^r(\widetilde{\Omega},\widetilde{\mathcal{A}},\widetilde{\mathbb{P}};\overline{L^{np}(X)^{\otimes n}})$ (see, e.g., \cite{hytoenen16} for more details), where $\overline{L^{np}(X)^{\otimes n}}$ denotes the completion of $L^{np}(X)^{\otimes n}$. Now, we apply the strong law of large numbers for Banach space-valued random variables to obtain the following universal approximation result for random neural networks.

\begin{proposition}
	\label{PropRandUAT}
	Let $X$ be a continuous semimartingale satisfying Assumption~\ref{AssExpInt}, let $n \in \mathbb{N}_0$, $p,r \in [1,\infty)$, $g \in L^{np}(X)^{\otimes n}$, let $\rho \in C(\mathbb{R})$ be non-polynomial, and let Assumption~\ref{AssCondCDF} hold. Then, for every $\varepsilon > 0$ there exists some $\widetilde{\varphi} \in \mathcal{RN}^\rho_{d,n}$ such that $\widetilde{\mathbb{E}}\big[ \Vert g - \widetilde{\varphi}(\cdot) \Vert_{L^{np}(X)^{\otimes n}}^r \big]^{1/r} < \varepsilon$.
\end{proposition}

Proposition~\ref{PropRandUAT} implies the following approximation result for financial derivatives, which is similar to Theorem~\ref{ThmUAT} but now with random neural networks. The proof can be found in Section~\ref{SecProofsRandUAT}.

\begin{theorem}[Random universal approximation]
	\label{ThmRandUAT}
	Let $X$ be a continuous semimartingale satisfying Assumption~\ref{AssExpInt}, let $p,r \in [1,\infty)$, let $\rho \in C(\mathbb{R})$ be non-polynomial, and let Assumption~\ref{AssCondCDF} hold. Then, for every $G \in L^p(\Omega,\mathcal{F}_T,\mathbb{P})$ and $\varepsilon > 0$ there exists $N \in \mathbb{N}$ and $\widetilde{\varphi}_n \in \mathcal{RN}^\rho_{d,n}$, $n = 0,...,N$, such that $\big( \widetilde{\omega} \mapsto \sum_{n=0}^N J^\circ_n(\widetilde{\varphi}_n(\widetilde{\omega}))_T \big) \in L^r(\widetilde{\Omega},\widetilde{\mathcal{A}},\widetilde{\mathbb{P}};L^p(\Omega,\mathcal{F}_T,\mathbb{P}))$ satisfies $\widetilde{\mathbb{E}}\big[ \big\Vert G - \sum_{n=0}^N J^\circ_n(\widetilde{\varphi}_n(\widetilde{\omega}))_T \big\Vert_{L^p(\mathbb{P})}^r \big]^{1/r} < \varepsilon$.
\end{theorem}

\begin{remark}
	\label{RemDiscussion}
	The universal approximation results in Theorem~\ref{ThmUAT} (for fully trained neural networks) and in Theorem~\ref{ThmRandUAT} (for random neural networks) extend the following results in the literature:
	\begin{itemize}
		\item[(i)] By viewing financial derivatives as path-dependent functionals, \cite{perez19,lyons20,cartea22,primavera22,cuchiero23,bayer25} applied signature methods. Originating from rough path theory (see \cite{lyons98,friz10,friz20}), the signature satisfies a similar universality property than neural networks, but on path space. Since the (Stratonovich) signature of a semimartingale is equal to the iterated (Stratonovich) integrals with constant integrand equal to one, Theorem~\ref{ThmUAT} can be interpreted as $L^p$-universality of the signature.
		\item[(ii)] The random universal approximation result in Proposition~\ref{PropRandUAT} is similar to \cite{gonon20,gonon21,schmocker23}.
	\end{itemize}
\end{remark}

\subsection{$L^p$-hedging}
\label{SecLpHedgingRN}

For $p \in [1,\infty)$, we now consider again the $L^p$-hedging problem from Section~\ref{SecLpHedging}. However, instead of arbitrary $L^{np}(X)$-functions in Theorem~\ref{ThmLpHedging}, we now use (possibly random) neural networks, which are still able to approximately solve the $L^p$-hedging problem. First, we provide the result for fully trained neural networks, whose proof is given in Section~\ref{SecProofsLpHedgingRN}.

\begin{theorem}[$L^p$-hedging with fully trained neural networks]
	\label{ThmLpHedgingNN}
	Let $X$ be a continuous semimartingale satisfying Assumption~\ref{AssExpInt}, let $p \in [1,\infty)$, and let $\rho \in C(\mathbb{R})$ be non-polynomial. Then, for every $G \in L^p(\Omega,\mathcal{F}_T,\mathbb{P})$ and $\varepsilon > 0$ there exist some $\varphi_0 \in \mathbb{R}$ and $N,m_n \in \mathbb{N}$ as well as $\varphi_{n,j,0},\varphi_{n,j,1} \in \mathcal{NN}^\rho_{d,1}$, $n = 1,...,N$ and $j = 1,...,m_n$, such that
	\begin{equation*}
		\left\Vert G - \varphi_0 - \int_0^T \left( \vartheta^{\varphi_{1:N}}_t \right)^\top dX_t \right\Vert_{L^p(\mathbb{P})} \leq \varepsilon + \inf_{(c,\theta) \in \mathbb{R} \times \Theta^p(X)} \left\Vert G - c - \int_0^T \theta_t^\top dX_t \right\Vert_{L^p(\mathbb{P})},
	\end{equation*} 
	where $\vartheta^{\varphi_{1:N}}_t := \sum_{n=1}^N \sum_{j=1}^{m_n} \frac{W(\varphi_{n,j,0})_t^{n-1}}{(n-1)!} \varphi_{n,j,1}(t)$, for $t \in [0,T]$.
\end{theorem}

Next, we use random neural networks to approximately solve the $L^p$-hedging problem.

\begin{theorem}[$L^p$-hedging with random neural networks]
	\label{ThmLpHedgingRN}
	Let $X$ be a continuous semimartingale satisfying Assumption~\ref{AssExpInt}, let $p,r \in [1,\infty)$, let $\rho \in C(\mathbb{R})$ be non-polynomial, and let Assumption~\ref{AssCondCDF} hold. Then, for every $G \in L^p(\Omega,\mathcal{F}_T,\mathbb{P})$ and $\varepsilon > 0$ there exist some $\widetilde{\varphi}_0 \in \mathbb{R}$ and $N,m_n \in \mathbb{N}$ as well as $\widetilde{\varphi}_{n,j,0},\widetilde{\varphi}_{n,j,1} \in \mathcal{RN}^\rho_{d,1}$, $n = 1,...,N$ and $j = 1,...,m_n$, such that
	\begin{equation*}
		\widetilde{\mathbb{E}}\left[ \left\Vert G - \widetilde{\varphi}_0 - \int_0^T \left( \vartheta^{\widetilde{\varphi}_{1:N}(\cdot)}_t \right)^\top dX_t \right\Vert_{L^p(\mathbb{P})}^r \right]^\frac{1}{r} \leq \varepsilon + \inf_{(c,\theta) \in \mathbb{R} \times \Theta^p(X)} \left\Vert G - c - \int_0^T \theta_t^\top dX_t \right\Vert_{L^p(\mathbb{P})},
	\end{equation*} 
	where $\vartheta^{\widetilde{\varphi}_{1:N}(\widetilde{\omega})}_t := \sum_{n=1}^N \sum_{j=1}^{m_n} \frac{W(\widetilde{\varphi}_{n,j,0}(\widetilde{\omega}))_t^{n-1}}{(n-1)!} \widetilde{\varphi}_{n,j,1}(\widetilde{\omega})(t)$, for $t \in [0,T]$.
\end{theorem}

Hence, by using (possibly random) neural networks, we are indeed able to approximately solve the $L^p$-hedging problem, which coincides for $p = 2$ with the quadratic hedging problem described in \cite{schweizer99,pham00}.

\section{Numerical examples}
\label{SecNumEx}

In this section, we numerically solve the $L^p$-hedging problem described in Section~\ref{SecLpHedging}+\ref{SecLpHedgingRN}. To this end, we assume that $X$ is a continuous semimartingale satisfying Assumption~\ref{AssExpInt}. Then, for a given financial derivative $G \in L^p(\Omega,\mathcal{F}_T,\mathbb{P})$, we follow Algorithm~\ref{Alg} to learn the trading strategy from Theorem~\ref{ThmLpHedgingRN}.

\begin{algorithm}[!htbp]
	\DontPrintSemicolon
	\begin{small}
		\KwInput{Financial derivative $G \in L^p(\Omega,\mathcal{F}_T,\mathbb{P})$, for some $p \in [1,\infty)$, and $K,L,N,m_1,...,m_N \in \mathbb{N}$.}
		\KwOutput{Approximating $L^p$-optimal strategy $\vartheta^{\widetilde{\varphi}_{1:N}(\widetilde{\omega})}$ of $G \in L^p(\Omega,\mathcal{F}_T,\mathbb{P})$ from Theorem~\ref{ThmLpHedgingRN}.}
		
		\vspace{0.08cm}
		
		Generate $L$ discretized sample paths $(X_{t_k}(\omega_l))_{k=0,...,K}$ of $X$, $l=1,...,L$, where $0 = t_0 < ... < t_K = T$ is a partition of $[0,T]$, and compute for every $l = 1,...,L$ the corresponding payoff $G(\omega_l)$.
		\label{Alg1}
		
		\vspace{0.08cm}
		
		Let $\rho \in C(\mathbb{R})$ be non-polynomial and initialize the weights $(\widetilde{a}_{n,j,0})_{n,j \in \mathbb{N}}, (\widetilde{a}_{n,j,1})_{n,j \in \mathbb{N}} \overset{\text{i.i.d.}}{\sim} \widetilde{a}_1$ and biases $(\widetilde{b}_{n,j,0})_{n,j \in \mathbb{N}}, (\widetilde{b}_{n,j,1})_{n,j \in \mathbb{N}} \overset{\text{i.i.d.}}{\sim} \widetilde{b}_1$ randomly, i.e.~fix some $\widetilde{\omega} \in \widetilde{\Omega}$. Hereby, Assumption~\ref{AssCondCDF} should hold.
		\label{Alg2}
		
		\vspace{0.08cm}
		
		Let $\widetilde{\varphi}_0 \in \mathbb{R}$ and define for every $n = 1,...,N$ and $j = 1,...,m_n$ the random neurons
		\label{Alg3}
		\vspace{-0.05cm}
		\begin{equation*}
			\begin{aligned}
				[0,T] \ni t \quad \mapsto \quad \widetilde{\varphi}_{n,j,0}(\widetilde{\omega})(t) & := \rho\left( \widetilde{a}_{n,j,0}(\widetilde{\omega}) t + \widetilde{b}_{n,j,0}(\widetilde{\omega}) \right) \in \mathbb{R}^d, \\
				[0,T] \ni t \quad \mapsto \quad \widetilde{\varphi}_{n,j,1}(\widetilde{\omega})(t) & := \rho\left( \widetilde{a}_{n,j,1}(\widetilde{\omega}) t + \widetilde{b}_{n,j,1}(\widetilde{\omega}) \right) \in \mathbb{R}^d.
			\end{aligned}
			\vspace{-0.05cm}
		\end{equation*}
		
		\vspace{0.08cm}
		
		Compute for every $n = 1,...,N$, $j = 1,...,m_n$, $k = 0,...,K$, and $l = 1,...,L$ the stochastic integrals $W(\widetilde{\varphi}_{n,j,0}(\widetilde{\omega}))_{t_k}(\omega_l)$ and $W(\widetilde{\varphi}_{n,j,1}(\widetilde{\omega}))_{t_k}(\omega_l)$, e.g., by using Ito's formula, i.e.
		\label{Alg4}
		\vspace{-0.05cm}
		\begin{equation*}
			\begin{aligned}
				W(\widetilde{\varphi}_{n,j,0}(\widetilde{\omega}))_{t_k} & := \int_0^{t_k} \widetilde{\varphi}_{n,j,0}(\widetilde{\omega})(s)^\top dX_s \\
				& = \widetilde{\varphi}_{n,j,0}(\widetilde{\omega})(t_k)^\top X_{t_k} - \widetilde{\varphi}_{n,j,0}(\widetilde{\omega})(0)^\top X_0 - \int_0^{t_k} \widetilde{\varphi}_{n,j,0}(\widetilde{\omega})'(s)^\top X_s \, ds
			\end{aligned}
			\vspace{-0.05cm}
		\end{equation*}
		and similarly for $W(\widetilde{\varphi}_{n,j,1}(\widetilde{\omega}))_{t_k} \! := \! \int_0^{t_k} \widetilde{\varphi}_{n,j,1}(\widetilde{\omega})(s)^\top dX_s$, where $\rho \in C(\mathbb{R})$ must now be differentiable.
		
		\vspace{0.08cm}
		
		Let $y_{n,j} \in \mathbb{R}$, $n = 1,...,N$ and $j = 1,...,m_n$, and initialize the trading strategy $\vartheta^{\widetilde{\varphi}_{1:N}(\widetilde{\omega})} := \big( \vartheta^{\widetilde{\varphi}_{1:N}(\widetilde{\omega})}_t \big)_{t \in [0,T]}$ for every $t \in [0,T]$ by
		\vspace{-0.05cm}
		\begin{equation*}
			\vartheta^{\widetilde{\varphi}_{1:N}(\widetilde{\omega})}_t := \sum_{n=1}^N \sum_{j=1}^{m_n} y_{n,j} \frac{W(\widetilde{\varphi}_{n,j,0}(\widetilde{\omega}))_t^{n-1}}{(n-1)!} \widetilde{\varphi}_{n,j,1}(\widetilde{\omega})(t).
			\vspace{-0.05cm}
		\end{equation*}
		
		\vspace{0.08cm}
		
		Minimize the empirical $L^p$-hedging error over $\widetilde{\varphi}_0 \in \mathbb{R}$ and $y_{n,j} \in \mathbb{R}$, $n = 1,...,N$ and $j = 1,...,m_n$, i.e.
		\vspace{-0.05cm}
		\begin{equation}
			\label{EqAlg1}
			\inf_{\widetilde{\varphi}_0, y_{n,j} \in \mathbb{R}} \left( \sum_{l=1}^L \left\vert G(\omega_l) - \widetilde{\varphi}_0 - \left( \int_0^T \left( \vartheta^{\widetilde{\varphi}_{1:N}(\widetilde{\omega})}_t \right)^\top dX_t \right)(\omega_l) \right\vert^p \right)^\frac{1}{p},
			\vspace{-0.05cm}
		\end{equation}
		where the stochastic integral
		\vspace{-0.05cm}
		\begin{equation*}
			\int_0^T \left( \vartheta^{\widetilde{\varphi}_{1:N}(\widetilde{\omega})}_t \right)^\top dX_t = \sum_{n=1}^N \sum_{j=1}^{m_n} y_{n,j} \int_0^T \frac{W(\widetilde{\varphi}_{n,j,0}(\widetilde{\omega}))_t^{n-1}}{(n-1)!} \widetilde{\varphi}_{n,j,1}(\widetilde{\omega})(t)^\top dX_t
			\vspace{-0.05cm}
		\end{equation*}
		is approximated, e.g., by an Euler-Maruyama scheme.
		\label{Alg7}
		
		\vspace{0.08cm}
		
		Update the parameters $\widetilde{\varphi}_0 \in \mathbb{R}$ and $y_{n,j} \in \mathbb{R}$, $n = 1,...,N$ and $j = 1,...,m_n$, to the ones minimizing~\eqref{EqAlg1} and return the approximating $L^p$-optimal strategy $\vartheta^{\widetilde{\varphi}_{1:N}(\widetilde{\omega})} := \big( \vartheta^{\widetilde{\varphi}_{1:N}(\widetilde{\omega})}_t \big)_{t \in [0,T]}$.
		\label{Alg8}
	\end{small}
	\caption{Learning the $L^p$-optimal strategy of a financial derivative.}
	\label{Alg}
\end{algorithm}

\algorithmstyle{plain}
\LinesNotNumbered
\setlength{\algomargin}{0em}

\begin{algorithm}[!htbp]
	\justifying
	\small
	\textit{\textbf{Fully trained neural networks:} If fully trained neural networks are used instead of random neural networks, we omit the $\widetilde{\omega}$-dependence in Algorithm~\ref{Alg1} and replace Step~{\scriptsize\textbf{\ref{Alg3}}}+{\scriptsize\textbf{\ref{Alg7}}} by the following steps:}
\end{algorithm}

\setlength{\algomargin}{1.1em}

\begin{algorithm}[!htbp]
	\begin{small}
		\nlset{3} Let $\varphi_0 \in \mathbb{R}$ and define for every $n = 1,...,N$ and $j = 1,...,m_n$ the fully trained neurons
		\begin{equation*}
			\begin{aligned}
				[0,T] \ni t \quad \mapsto \quad \varphi_{n,j,0}(t) := \rho\left( a_{n,j,0} t + b_{n,j,0} \right) \in \mathbb{R}^d, \\
				[0,T] \ni t \quad \mapsto \quad \varphi_{n,j,1}(t) := \rho\left( a_{n,j,1} t + b_{n,j,1} \right) \in \mathbb{R}^d,
			\end{aligned}
		\end{equation*}
		for some $a_{n,j,0}, a_{n,j,1}, b_{n,j,0}, b_{n,j,1} \in \mathbb{R}^d$, $n = 1,...,N$ and $j = 1,...,m_n$.
		\vspace{0.08cm}
		
		\nlset{6} Minimize \eqref{EqAlg1} over $\varphi_0 \in \mathbb{R}$ as well as $a_{n,j,0}, a_{n,j,1}, b_{n,j,0}, b_{n,j,1}, y_{n,j} \in \mathbb{R}$, $j = 1,...,m_n$ and $n = 1,...,N$. Then, $\vartheta^{\varphi_{1:N}} := \big( \vartheta^{\varphi_{1:N}}_t \big)_{t \in [0,T]}$ is the approximating $L^p$-optimal strategy from Theorem~\ref{ThmLpHedgingNN}.
	\end{small}
\end{algorithm}

\setlength{\algomargin}{0em}

\begin{algorithm}[!htbp]
	\justifying
	\small
	\textit{\textbf{Comparison:} Random neural networks outperform fully trained neural networks in terms of computational efficiency and stability. Indeed, fully trained neural networks require an iterative algorithm for training (e.g., stochastic gradient descent, see \cite{gbc16}), whereas only the linear readout of a random neural network needs to be trained (e.g., least squares method). Moreover, the corresponding optimization problem for training a random neural network is convex, which is not the case for fully trained neural networks (see \cite{schmocker23}).}
	
	\hspace{-0.41cm}{\centering \rule{\textwidth}{0.05cm}}
\end{algorithm}
	
In the following numerical experiments\footnote{\label{footnote1}All numerical experiments have been implemented in \texttt{Python} on an average laptop (Lenovo ThinkPad X13 Gen2a with Processor AMD Ryzen 7 PRO 5850U and Radeon Graphics, 1901 Mhz, 8 Cores, 16 Logical Processors). The code can be found under the following link: \url{https://github.com/psc25/ChaoticHedging}}, we use random neural networks to learn the approximating $L^p$-optimal strategy of a financial derivative within a few seconds. To this end, we generate $L = 10^5$ sample paths of the semimartingale $X$, discretized over an equidistant time grid $(t_k)_{k = 0,...,K}$ of $K = 500$ time points, which are split up into $80\%$/$20\%$ for training and testing. Then, for each order of the chaos expansion $n = 0,...,N$, we initialize random neurons of size $m_n = 50$ with non-polynomial activation function $\rho(s) = \tanh(s)$.

In order to compare our results with the quadratic hedging approach in \cite{schweizer99,pham00}, we choose $p = 2$. Then, for a given financial derivative $G \in L^2(\Omega,\mathcal{F}_T,\mathbb{P})$ and different $N = 0,...,6$, we learn the approximating $L^p$-optimal strategy $\vartheta^{\widetilde{\varphi}_{1:N}(\widetilde{\omega})}$ from Theorem~\ref{ThmLpHedgingRN}. The following numerical examples empirically demonstrate that the pseudo-optimal $L^2$-strategy of $G \in L^2(\Omega,\mathcal{F}_T,\mathbb{P})$, as defined in \cite[p.~17]{schweizer99}, can indeed be approximated by $\vartheta^{\widetilde{\varphi}_{1:N}(\widetilde{\omega})}$.

\subsection{European call option with Brownian motion}
\label{SubsecExBM}

In the first example, we consider a one-dimensional Brownian motion $X := (X_t)_{t \in [0,T]}$, which is a continuous semimartingale satisfying Assumption~\ref{AssExpInt}. Then, we approximately solve the $L^2$-hedging problem for the European call option
\begin{equation}
	\label{EqEuroCall}
	G = \max(X_T-K_{\text{str}}, 0),
\end{equation}
where $K_{\text{str}} \in \mathbb{R}$ denotes the strike price. Since the market is complete and $\mathbb{P}$ is a (local) martingale measure, the $L^2$-optimal strategy is the replication strategy under $\mathbb{P}$. Hence, by using the Clark-Ocone formula in \cite{clark70,ocone84} (see also \cite[Theorem~4.1]{dinunno08}), the true hedging strategy $\theta = (\theta_t)_{t \in [0,T]}$ of $G$ is for every $t \in [0,T]$ given by $\theta_t = \mathbb{E}[D_t G \vert \mathcal{F}_t]$, where $D_t G$ denotes the Malliavin derivative of $G$ (see \cite[Definition~1.2.1]{nualart06}). Using the chain rule of the Malliavin derivative in \cite[Proposition~1.2.4]{nualart06}, it follows for every $t \in [0,T)$ that $D_t G = \mathds{1}_{[K_{\text{str}},\infty)}(X_T)$ and thus
\begin{equation}
	\label{EqEuroCallHedgStrat}
	\theta_t = \mathbb{E}\left[\mathds{1}_{[K_{\text{str}},\infty)}(X_T) \big\vert \mathcal{F}_t \right] = \int_{K_{\text{str}}}^\infty \frac{1}{\sqrt{2 \pi (T-t)}} e^{-\frac{(x-X_t)^2}{2(T-t)}} dx = \Phi\left(\frac{X_t-K_{\text{str}}}{\sqrt{T-t}}\right),
\end{equation}
where $\Phi: \mathbb{R} \rightarrow [0,1]$ denotes the cumulative distribution function of the standard normal distribution. For the numerical experiment in Figure~\ref{FigBM}, we choose the parameters $T = 1$ and $K_{\text{str}} = -0.5$.

\begin{figure}[h!]
	\centering
	\begin{minipage}[b][][b]{0.49\textwidth}
		\centering
		\includegraphics[height = 4.9cm,trim={0.63cm 0 0 0}]{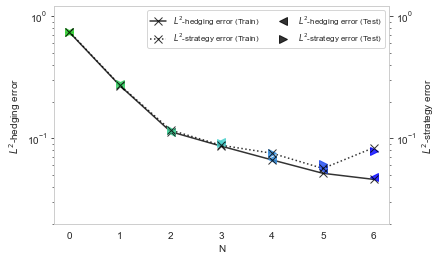}
		
		\subcaption{Learning performance}
		\vspace{0.2cm}
	\end{minipage}
	\begin{minipage}[b][][b]{0.49\textwidth}
		\centering
		\includegraphics[height = 4.9cm]{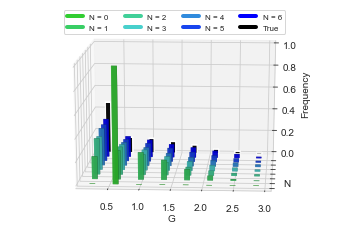}
		
		\subcaption{Payoff distribution on test set}
		\vspace{0.2cm}
	\end{minipage}
	\begin{minipage}[b][][b]{0.49\textwidth}
		\centering
		\includegraphics[height = 4.9cm]{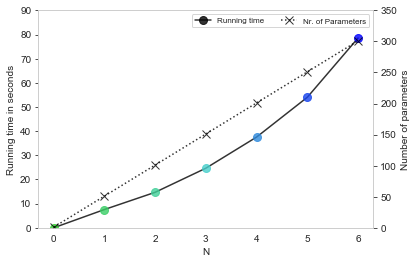}
		
		\subcaption{Running time and number of parameters}
		\vspace{0.2cm}
	\end{minipage}
	\begin{minipage}[b][][b]{0.49\textwidth}
		\centering
		\includegraphics[height = 4.9cm]{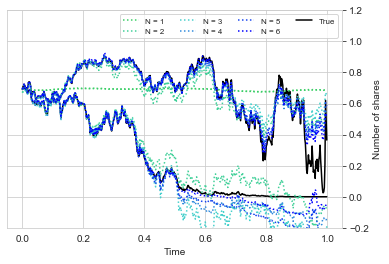}
		
		\subcaption{$\theta$ and $\vartheta^{\widetilde{\varphi}_{1:N}(\widetilde{\omega})}$ for two samples of test set}
		\vspace{0.2cm}
	\end{minipage}
	\caption{Learning the $L^2$-optimal strategy of the European call option $G$ defined in \eqref{EqEuroCall}, for $N = 0,...,6$. In (a), the $L^2(\mathbb{P})$-hedging error \eqref{EqAlg1} and the empirical $L^2(dt \otimes \mathbb{P})$-strategy error $\big( \frac{T}{KL} \sum_{l=1}^L \sum_{k=0}^K \big\vert \theta_{t_k}(\omega_l) - \vartheta^{\widetilde{\varphi}_{1:N}(\widetilde{\omega})}_{t_k}(\omega_l) \big\vert^2 \big)^{1/2}$ are displayed, where $\theta$ is defined in \eqref{EqEuroCallHedgStrat}. In (b), the distributions of $G$ (label ``True'') and $\widetilde{\varphi}_0 + \int_0^T \vartheta^{\widetilde{\varphi}_{1:N}(\widetilde{\omega})}_t dX_t$ (label ``$N = 0,...,6$'') are depicted on the test set. In (c), the running time\textsuperscript{\ref{footnote2}} and the number of estimated parameters are shown. In (d), we compare $\theta$ (label ``True'') and $\vartheta^{\widetilde{\varphi}_{1:N}(\widetilde{\omega})}$ (label ``$N = 1,...,6$'') on two samples of the test set.}
	\label{FigBM}
\end{figure}

\footnotetext[4]{\label{footnote2}The running time corresponds to the total execution time of Algorithm~\ref{Alg} including the computation of all inputs, solving the least squares problem (taking less than one second), and the calculation of the approximating $L^2$-optimal strategy.}

\subsection{Asian put option in a Vasi\v{c}ek model with stochastic volatility}
\label{SecExVasicek}

In the second example, we consider a Vasi\v{c}ek model with stochastic volatility, i.e.~the price process $X := (X_t)_{t \in [0,T]}$ follows the SDE
\begin{equation}
	\label{EqVasicek}
	\begin{cases}
		dX_t = \kappa \left( \mu - X_t \right) dt + \sqrt{\sigma_t} dB_t, & \\
		d\sigma_t \; = \alpha (\beta - \sigma_t) dt + \xi \sqrt{\sigma_t} d\widetilde{B}_t, &
	\end{cases}
	\quad t \in [0,T]
\end{equation}
with initial values $(X_0,\sigma_0) \in \mathbb{R} \times (0,\infty)$ and parameters $\kappa, \mu, \alpha, \beta, \xi > 0$. Hereby, the Brownian motions $B := (B_t)_{t \in [0,T]}$ and $\widetilde{B} := \big( \widetilde{B}_t \big)_{t \in [0,T]}$ are correlated with parameter $\rho \in [-1,1]$, i.e.~$\widetilde{B}_t = \rho B_t + \sqrt{1 - \rho^2} B^\perp_t$ for all $t \in [0,T]$ and another Brownian motion $B^\perp := \big( B^\perp_t \big)_{t \in [0,T]}$ which is independent of $B$. Since $\widetilde{X} := (X_t,\sigma_t)^\top_{t \in [0,T]}$ is a polynomial diffusion with diffusion coefficient $a: \mathbb{R} \times (0,\infty) \rightarrow \mathbb{S}^2_+$ of linear growth, Example~\ref{ExPolynDiff} ensures that $\widetilde{X}$ satisfies Assumption~\ref{AssExpInt}, which thus also applies to $X$. In this setting, we approximately solve the $L^2$-hedging problem for the Asian put option
\begin{equation}
	\label{EqAsianPut}
	G = \max\left( K_{\text{str}} - \frac{1}{T} \int_0^T X_t dt, 0 \right),
\end{equation}
where $K_{\text{str}} \in \mathbb{R}$ denotes the strike price. 

In order to compute the $L^2$-optimal strategy in the sense of quadratic hedging (see \cite{schweizer99,pham00}), we first observe that the market price of risk $t \mapsto \lambda_t := \sigma_t^{-1} \kappa (\mu - X_t)$ induces the minimal equivalent local martingale measure $\mathbb{Q} \sim \mathbb{P}$ whose density $\frac{d\mathbb{Q}}{d\mathbb{P}} := Z_T$ is the terminal value of the SDE $dZ_t = -Z_t \lambda_t \sqrt{\sigma_t} dB_t$, with $Z_0 = 1$. Hence, $X$ is by the Girsanov theorem a local martingale under $\mathbb{Q}$ satisfying the SDE $dX_t = \sqrt{\sigma_t} dB^\mathbb{Q}_t$, where $B^\mathbb{Q} := \big( B^\mathbb{Q}_t \big)_{t \in [0,T]}$ is a Brownian motion under $\mathbb{Q}$. Thus, \cite[Theorem~3.5]{schweizer99} shows that $t \mapsto V^\mathbb{Q}_t := \mathbb{E}^\mathbb{Q}[G \vert \mathcal{F}_t]$ is the value process of the pseudo-optimal $L^2$-strategy in the sense of \cite[p.~17]{schweizer99}. For the calculation of $\mathbb{E}^\mathbb{Q}[G \vert \mathcal{F}_t]$, we use the notation $I_{r,s} := \frac{1}{T} \int_r^s X_t dt$, the Fourier arguments in \cite{carr99} together with the put-call parity, and the affine transform formulas in \cite[Corollary~4.11]{kellerressel09} and \cite[Section~2.3]{duffie00} to conclude for every $t \in [0,T]$ that
\begin{equation*}
	\begin{aligned}
		V^\mathbb{Q}_t & = \mathbb{E}^\mathbb{Q}[G \vert \mathcal{F}_t] = \mathbb{E}^\mathbb{Q}\left[ \max\left( K_{\text{str}} - I_{0,t} - I_{t,T}, 0 \right) \vert \mathcal{F}_t \right] \\
		& = (K_{\text{str}} - I_{0,t}) \mathbb{E}^\mathbb{Q}\left[ \mathds{1}_{(-\infty,K_{\text{str}}-I_{0,t}]}(I_{t,T}) \big\vert \mathcal{F}_t \right] - \mathbb{E}^\mathbb{Q}\left[ I_{t,T} \mathds{1}_{(-\infty,K_{\text{str}}-I_{0,t}]}(I_{t,T}) \big\vert \mathcal{F}_t \right] \\
		& = (K_{\text{str}} - I_{0,t}) \left( \frac{1}{2} - \int_0^\infty \re\left( \tfrac{\exp\big(-\mathbf{i} u (K_{\text{str}}-I_{0,t})\big) \mathbb{E}\left[ \exp\left( \mathbf{i} u I_{t,T} \right) \big\vert \mathcal{F}_t \right]}{\mathbf{i} \pi u} \right) du \right) \\
		& \quad\quad - \left( \frac{\mathbb{E}^\mathbb{Q}[I_{t,T}]}{2} + \frac{1}{2} \int_0^\infty \re\left( \tfrac{\exp\big( -\mathbf{i} u (K-I_{0,t}) \big) \mathbb{E}\left[ \mathbf{i} I_{t,T} \exp\left( \mathbf{i} u I_{t,T} \right) \big\vert \mathcal{F}_t \right]}{\mathbf{i} \pi u} \right) du \right) \\
		& = (K_{\text{str}} - I_{0,t}) \left( \frac{1}{2} - \int_0^\infty \re\left( \tfrac{\exp\big( -\mathbf{i} u (K_{\text{str}}-I_{0,t}) \big) \exp\left( \phi^{(0)}_{T-t}(u)^\top \widetilde{X}_t \right)}{\mathbf{i} \pi u} \right) du \right) \\
		& \quad\quad - \frac{T-t}{2T} X_t - \int_0^\infty \re\left( \tfrac{\exp\big(-\mathbf{i} u (K_{\text{str}}-I_{0,t})\big) \exp\left( \phi^{(0)}_{T-t}(u)^\top \widetilde{X}_t \right) \phi^{(1)}_{T-t}(u)^\top \widetilde{X}_t}{\pi u} \right) du,
	\end{aligned}
\end{equation*}
where $\phi^{(0)}_h, \phi^{(1)}_h: \mathbb{C} \rightarrow \mathbb{C}^2$ satisfy the Riccati equations $\frac{\partial}{\partial h} \phi^{(0)}_h(u) = \big( \mathbf{i} u, \frac{1}{2} \phi^{(0)}_h(u)^\top A \phi^{(0)}_h(u) \big)^\top$ and $\frac{\partial}{\partial h} \phi^{(1)}_h(u) = \big( \mathbf{i}, \phi^{(0)}_h(u)^\top A \phi^{(1)}_h(u) \big)^\top$, respectively, with initial values $\phi^{(0)}_0(u) = \phi^{(1)}_0(u) = 0$ and matrix $A := (a_{i,j})_{i,j=1,2} \in \mathbb{R}^2$, where $a_{11} = 1$, $a_{12} = a_{21} = \xi \rho$, and $a_{22} = \xi^2$. Hence, by using Ito's formula together with $\frac{d}{dt} \langle \widetilde{X}, X \rangle_t = \sigma_t (1, \xi \rho)^\top$, we conclude that the $L^2$-optimal strategy $\theta$ is given as
\begin{equation}
	\label{EqAsianPutHedgStrat}
	\begin{aligned}
		& t \quad \mapsto \quad \theta_t = \frac{\tfrac{d}{dt} \langle V^\mathbb{Q}, X \rangle_t}{\tfrac{d}{dt} \langle X \rangle_t} \\
		& \quad\quad = - \frac{1}{\sigma_t} \int_0^\infty \re\left( \tfrac{\tfrac{d}{dt} \big\langle (K_{\text{str}}-I_{0,\cdot}) \exp\big( -\mathbf{i} u (K_{\text{str}}-I_{0,\cdot}) \big) \exp\left( \phi^{(0)}_{T-\,\cdot}(u)^\top \widetilde{X} \right), X \big\rangle_t}{\mathbf{i} \pi u} \right) du - \frac{\tfrac{d}{dt} \big\langle \frac{T-\,\cdot}{2T} X, X \big\rangle_t}{\sigma_t} \\
		& \quad\quad\quad\quad - \frac{1}{\sigma_t} \int_0^\infty \re\left( \tfrac{\tfrac{d}{dt} \big\langle \exp\big( -\mathbf{i} u (K_{\text{str}}-I_{0,\cdot}) \big) \exp\left( \phi^{(0)}_{T-\,\cdot}(u)^\top \widetilde{X} \right) \phi^{(1)}_{T-\,\cdot}(u)^\top \widetilde{X}, X \big\rangle_t}{\pi u} \right) du \\
		& \quad\quad = - (K_{\text{str}}-I_{0,t}) \int_0^\infty \re\left( \tfrac{\exp\big( -\mathbf{i} u (K_{\text{str}}-I_{0,t}) \big) \exp\left( \phi^{(0)}_{T-t}(u)^\top \widetilde{X}_t \right) \phi^{(0)}_{T-t}(u)^\top (1,\xi \rho)^\top}{\mathbf{i} \pi u} \right) du - \frac{T-t}{2T} \\
		& \quad\quad\quad\quad - \int_0^\infty \re\left( \tfrac{ \exp\big( -\mathbf{i} u (K_{\text{str}}-I_{0,t}) \big) \exp\left( \phi^{(0)}_{T-t}(u)^\top \widetilde{X}_t \right) \left( \phi^{(1)}_{T-t}(u)^\top \widetilde{X}_t \phi^{(0)}_{T-t}(u) + \phi^{(1)}_{T-t}(u) \right)^\top (1,\xi \rho)^\top}{\pi u} \right) du.
	\end{aligned}
\end{equation}
We then apply the fractional fast Fourier transform (FFT) in \cite{chourdakis05} to compute the integrals in \eqref{EqAsianPutHedgStrat}.

For the numerical experiment in Figure~\ref{FigVS}, we choose the initial values $(X_0,\sigma_0) = (100,4)$ as well as the parameters $T = 1$, $K_{\text{str}} = 101$, $\kappa = 0.5$, $\mu = 100$, $\alpha = 1$, $\beta = 5$, $\xi = 1$, and $\rho = -0.7$.

\begin{figure}[h!]
	\centering
	\begin{minipage}[b][][b]{0.49\textwidth}
		\centering
		\includegraphics[height = 4.9cm,trim={0.63cm 0 0 0}]{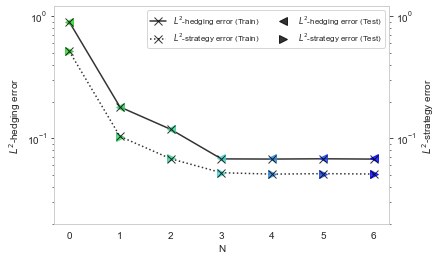}
		
		\subcaption{Learning performance}
		\vspace{0.2cm}
	\end{minipage}
	\begin{minipage}[b][][b]{0.49\textwidth}
		\centering
		\includegraphics[height = 4.9cm]{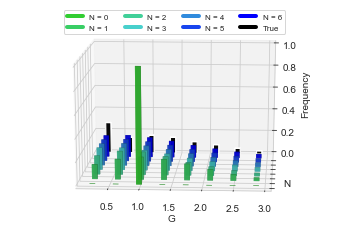}
		
		\subcaption{Payoff distribution on test set}
		\vspace{0.2cm}
	\end{minipage}
	\begin{minipage}[b][][b]{0.49\textwidth}
		\centering
		\includegraphics[height = 4.9cm]{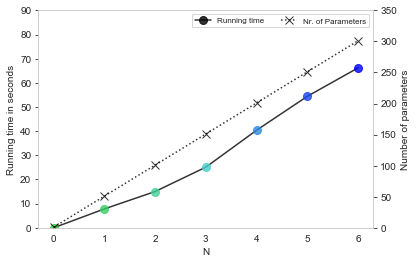}
		
		\subcaption{Running time and number of parameters}
		\vspace{0.2cm}
	\end{minipage}
	\begin{minipage}[b][][b]{0.49\textwidth}
		\centering
		\includegraphics[height = 4.9cm]{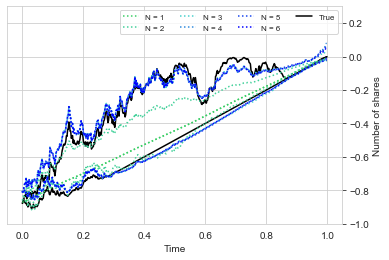}
		
		\subcaption{$\theta$ and $\vartheta^{\widetilde{\varphi}_{1:N}(\widetilde{\omega})}$ for two samples of test set}
		\vspace{0.2cm}
	\end{minipage}
	\caption{Learning the $L^2$-optimal strategy of the Asian put option $G$ defined in \eqref{EqAsianPut}, for $N = 0,...,6$. In (a), the $L^2(\mathbb{P})$-hedging error \eqref{EqAlg1} and the empirical $L^2(dt \otimes \mathbb{P})$-strategy error $\big( \frac{T}{KL} \sum_{l=1}^L \sum_{k=0}^K \big\vert \theta_{t_k}(\omega_l) - \vartheta^{\widetilde{\varphi}_{1:N}(\widetilde{\omega})}_{t_k}(\omega_l) \big\vert^2 \big)^{1/2}$ are displayed, where $\theta$ is defined in \eqref{EqAsianPutHedgStrat}. In (b), the distributions of $G$ (label ``True'') and $\widetilde{\varphi}_0 + \int_0^T \vartheta^{\widetilde{\varphi}_{1:N}(\widetilde{\omega})}_t dX_t$ (label ``$N = 0,...,6$'') are depicted on the test set. In (c), the running time\textsuperscript{\ref{footnote2}} and the number of estimated parameters are shown. In (d), we compare $\theta$ (label ``True'') and $\vartheta^{\widetilde{\varphi}_{1:N}(\widetilde{\omega})}$ (label ``$N = 1,...,6$'') on two samples of the test set.}
	\label{FigVS}
\end{figure}

\subsection{Basket option in Wishart affine stochastic correlation model}

In the third example (see Figure~\ref{FigAD}), we consider a $d$-dimensional affine model with stochastic correlation, i.e.~the stock price $X := (X_t)_{T \in [0,T]}$ follows the SDE
\begin{equation*}
	\begin{cases}
		dX_t = \kappa (\mu - X_t) dt + \sqrt{\Sigma_t} dB_t, \\
		d\Sigma_t \;\! = \left( \beta U^\top U + V \Sigma_t + \Sigma_t V^\top \right) dt + \sqrt{\Sigma_t} d\widetilde{B}_t U + U^\top d\widetilde{B}_t^\top \sqrt{\Sigma_t},
	\end{cases}
	\quad\quad t \in [0,T],
\end{equation*}
with initial values $(X_0,\Sigma_0) \in \mathbb{R}^d \times \mathbb{S}^{d-1}$ as well as parameters $\mu \in \mathbb{R}^d$, $\kappa > 0$, $\beta > d-1$, $U \in \mathbb{R}^{d \times d}$ being invertible, and $V \in \mathbb{R}^{d \times d}$. Hereby, the $d$-dimensional Brownian motion $B := (B_t)_{t \in [0,T]}$ and the ($d \times d$)-dimensional Brownian motion $\widetilde{B} := \big( \widetilde{B}_t \big)_{t \in [0,T]}$ are correlated with vector $\rho \in [-1,1]^d$ satisfying $\Vert \rho \Vert \leq 1$, i.e.~$B_t = \widetilde{B}_t \rho + \sqrt{1 - \rho^\top \rho} B^\perp_t$ for all $t \in [0,T]$ and a $d$-dimensional Brownian motion $B^\perp := \big(B^\perp_t\big)_{t \in [0,T]}$ which is independent of $\widetilde{B}$. Since $\widetilde{X} := \big( X_t, \vech(\Sigma_t) \big)_{t \in [0,T]}$ is an $\mathbb{R}^{d+d^2}$-valued polynomial diffusion with diffusion coefficient $a: \mathbb{R}^{d+d^2} \rightarrow \mathbb{S}^{d+d^2}_+$ of linear growth, where $\vech(\Sigma_t) := (\Sigma^{1,1}_t,\Sigma^{1,2}_t,...,\Sigma^{d,d}_t)_{t \in [0,T]}^\top \in \mathbb{R}^{d^2}$, Example~\ref{ExPolynDiff} ensures that $\widetilde{X}$ satisfies Assumption~\ref{AssExpInt}, which thus also applies to $X$. In this setting, we approximately solve the $L^2$-hedging problem for the Basket option
\begin{equation}
	\label{EqBasketOption}
	G = \max\left( K_{\text{str}} - w^\top X_T, 0 \right),
\end{equation}
where $K_{\text{str}} \in \mathbb{R}$ denotes the strike price and $w \in \mathbb{R}^d$ is a fixed vector.

In order to compute the optimal strategy in the sense of quadratic hedging (see \cite{schweizer99,pham00}), we follow the arguments of the previous section. The market price of risk $t \mapsto \lambda_t := \Sigma_t^{-1} \kappa (\mu - X_t)$ induces the minimal equivalent local martingale measure $\mathbb{Q} \sim \mathbb{P}$ whose density $\frac{d\mathbb{Q}}{d\mathbb{P}} := Z_T$ is the terminal value of the SDE $dZ_t = -Z_t \lambda_t^\top \sqrt{\Sigma_t} dB_t$, with $Z_0 = 1$. Hence, the process $X$ is a local martingale under $\mathbb{Q}$ satisfying $dX_t = \sqrt{\Sigma_t} dB^\mathbb{Q}_t$, where $B^\mathbb{Q} := \big( B^\mathbb{Q}_t \big)_{t \in [0,T]}$ is a $d$-dimensional Brownian motion under $\mathbb{Q}$. Thus, \cite[Theorem~3.5]{schweizer99} shows that $t \mapsto V^\mathbb{Q}_t := \mathbb{E}^\mathbb{Q}[G \vert \mathcal{F}_t]$ is the value process of the pseudo-optimal $L^2$-strategy in the sense of \cite[p.~17]{schweizer99}. For the calculation of the conditional expectation $\mathbb{E}^\mathbb{Q}[G \vert \mathcal{F}_t]$, we use the Fourier arguments in \cite{carr99} to conclude for every $t \in [0,T]$ that
\begin{equation*}
	\begin{aligned}
		V^\mathbb{Q}_t & = \mathbb{E}^\mathbb{Q}[G \vert \mathcal{F}_t] = \mathbb{E}^\mathbb{Q}\left[ \max\left( K_{\text{str}} - w^\top X_T, 0 \right) \Big\vert \mathcal{F}_t \right] \\
		& = K_{\text{str}} \mathbb{E}^\mathbb{Q}\left[ \mathds{1}_{\left\lbrace w^\top X_T \leq K_{\text{str}} \right\rbrace} \Big\vert \mathcal{F}_t \right] - \mathbb{E}^\mathbb{Q}\left[ w^\top X_T \mathds{1}_{\left\lbrace w^\top X_T \leq K_{\text{str}} \right\rbrace} \Big\vert \mathcal{F}_t \right] \\
		& = K_{\text{str}} \left( \frac{1}{2} - \int_0^\infty \re\left( \tfrac{e^{-\mathbf{i} u K} \mathbb{E}^\mathbb{Q}\left[ \exp\left( \mathbf{i} u w^\top X_T \right) \big\vert \mathcal{F}_t \right]}{\mathbf{i} \pi u} \right) du \right) \\
		& \quad\quad - \left( \frac{\mathbb{E}^\mathbb{Q}\left[ w^\top X_T \big\vert \mathcal{F}_t \right]}{2} + \int_0^\infty \re\left( \tfrac{e^{-\mathbf{i} u K} \mathbb{E}^\mathbb{Q}\left[ \mathbf{i} w^\top X_T \exp\left( \mathbf{i} u w^\top X_T \right) \big\vert \mathcal{F}_t \right]}{\pi u} \right) du \right) \\
		& = K_{\text{str}} \left( \frac{1}{2} - \int_0^\infty \re\left( \tfrac{e^{-\mathbf{i} u K} \Phi^{(0)}_{T-t,w}(u;X_t,\Sigma_t)}{\mathbf{i} \pi u} \right) du \right) \\
		& \quad\quad -\frac{w^\top X_t}{2} - \int_0^\infty \re\left( \tfrac{e^{-\mathbf{i} u K} \Phi^{(1)}_{T-t,w}(u;X_t,\Sigma_t)}{\pi u} \right) du.
	\end{aligned}
\end{equation*}
Hereby, $\mathbb{R} \ni u \mapsto \Phi^{(0)}_{h,w}(u;x,s) := \mathbb{E}^\mathbb{Q}\big[ \exp\left( \mathbf{i} u w^\top X_{t+h} \right) \big\vert (X_t,\Sigma_t) = (x,s) \big] \in \mathbb{C}$ denotes the conditional characteristic function of $w^\top X_{t+h}$ given that $(X_t,\Sigma_t) = (x,s) \in \mathbb{R}^d \times \mathbb{S}^d_+$, and where $\mathbb{R} \ni u \mapsto \Phi^{(1)}_{h,w}(u;x,s) := \mathbb{E}^\mathbb{Q}\big[ \mathbf{i} w^\top X_{t+h} \exp\left( \mathbf{i} u w^\top X_{t+h} \right) \big\vert (X_t,\Sigma_t) = (x,s) \big] \in \mathbb{C}$ is its derivative, which can be both computed via the affine transform formula in \cite{dafonseca07} (see also the code for more details). Hence, by using Ito's formula, the notations $\nabla_x \Phi^{(r)}_{h,w}(u,x,s) := \big( \frac{\partial}{\partial x_i} \Phi^{(r)}_{h,w}(u,x,s) \big)_{i=1,...,d}^\top$, $\nabla_s \Phi^{(r)}_{h,w}(u,x,s) := \big( \partial^s_{i,j} \Phi^{(r)}_{h,w}(u,x,s) \big)_{i,j=1,...,d}$, and $\partial^s_{i,j} \Phi^{(r)}_{h,w}(u,x,s) := \frac{\partial}{\partial s_{i,j}} \Phi^{(r)}_{h,w}(u,x,s)$, for $r \in \lbrace 0,1 \rbrace$, and that $\frac{d}{dt} \langle X \rangle_t = \Sigma_t$ as well as $\frac{d}{dt} \langle \Sigma^{i,j}, X^k \rangle_t = \Sigma^{k,i}_t (U^\top \rho)_j + \Sigma^{k,j}_t (U^\top \rho)_i$, we conclude that the pseudo-optimal $L^2$-strategy $\theta$ is for every $t \in [0,T]$ given as
\begin{equation}
	\label{EqBasketHedgStrat}
	\begin{aligned}
		& t \quad \mapsto \quad \theta_t = \left( \frac{d}{dt} \langle X \rangle_t \right)^{-1} \frac{d}{dt} \langle V^\mathbb{Q}, X \rangle_t \\
		& \quad = - \Sigma_t^{-1} K_{\text{str}} \int_0^\infty \re\left( \tfrac{e^{-\mathbf{i} u K} \frac{d}{dt} \big\langle \Phi^{(0)}_{T-\cdot,w}(u;X,\Sigma), X \big\rangle_t}{\mathbf{i} \pi u} \right) du \\
		& \quad\quad - \Sigma_t^{-1} \frac{\frac{d}{dt} \langle w^\top X, X \rangle_t}{2} - \Sigma_t^{-1} \int_0^\infty \re\left( \tfrac{e^{-\mathbf{i} u K} \frac{d}{dt} \big\langle \Phi^{(1)}_{T-\cdot,w}(u;X,\Sigma), X \big\rangle_t}{\pi u} \right) du \\
		& \quad = - \Sigma_t^{-1} K_{\text{str}} \int_0^\infty \re\left( \tfrac{e^{-\mathbf{i} u K} \left( \frac{d}{dt} \langle X \rangle_t \nabla_x \Phi^{(0)}_{T-t,w}(u;X_t,\Sigma_t) + \sum_{i,j=1}^d \partial^s_{i,j} \Phi^{(0)}_{T-t,w}(u;X_t,\Sigma_t) \frac{d}{dt} \langle \Sigma^{i,j}, X \rangle_t \right)}{\mathbf{i} \pi u} \right) du \\
		& \quad\quad -\frac{w}{2} - \Sigma_t^{-1} \int_0^\infty \re\left( \tfrac{e^{-\mathbf{i} u K} \left( \frac{d}{dt} \langle X \rangle_t \nabla_x \Phi^{(1)}_{T-t,w}(u;X_t,\Sigma_t) + \sum_{i,j=1}^d \partial^s_{i,j} \Phi^{(1)}_{T-t,w}(u;X_t,\Sigma_t) \frac{d}{dt} \langle \Sigma^{i,j}, X \rangle_t \right)}{\pi u} \right) du \\
		& \quad = - K_{\text{str}} \int_0^\infty \re\left( \tfrac{e^{-\mathbf{i} u K} \left( \nabla_x \Phi^{(0)}_{T-t,w}(u;X_t,\Sigma_t) + \left( \nabla_s \Phi^{(0)}_{T-t,w}(u;X_t,\Sigma_t) + \nabla_s \Phi^{(0)}_{T-t,w}(u;X_t,\Sigma_t)^\top \right) U^\top \rho \right)}{\mathbf{i} \pi u} \right) du \\
		& \quad\quad -\frac{w}{2} - \int_0^\infty \re\left( \tfrac{e^{-\mathbf{i} u K} \left( \nabla_x \Phi^{(1)}_{T-t,w}(u;X_t,\Sigma_t) + \left( \nabla_s \Phi^{(1)}_{T-t,w}(u;X_t,\Sigma_t) + \nabla_s \Phi^{(1)}_{T-t,w}(u;X_t,\Sigma_t)^\top \right) U^\top \rho \right)}{\pi u} \right) du.
	\end{aligned}
\end{equation}
We then apply the fractional fast Fourier transform (FFT) in \cite{chourdakis05} to compute the integrals in \eqref{EqBasketHedgStrat}.

For the numerical experiment in Figure~\ref{FigAD}, we choose $d = 2$, $X_0 = \mu = (10,10)^\top$, $T = 1$, $K_{\text{str}} = 21$, $w := (1,1)^\top$, $\kappa = 1$, and $\beta = 3$, whereas the other parameters are randomly initialized.

\begin{figure}[h!]
	\centering
	\begin{minipage}[b][][b]{0.49\textwidth}
		\centering
		\includegraphics[height = 4.9cm,trim={0.63cm 0 0 0}]{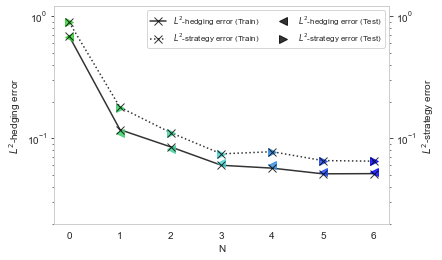}
		
		{\small {\bf (a)} Learning performance}
		\vspace{0.2cm}
	\end{minipage}
	\begin{minipage}[b][][b]{0.49\textwidth}
		\centering
		\includegraphics[height = 4.9cm]{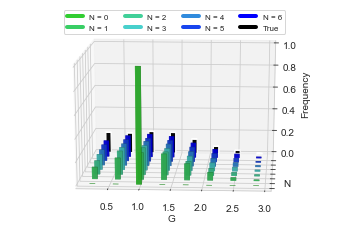}
		
		{\small {\bf (b)} Payoff distribution on test set}
		\vspace{0.2cm}
	\end{minipage}
	\begin{minipage}[b][][b]{0.49\textwidth}
		\centering
		\includegraphics[height = 4.9cm]{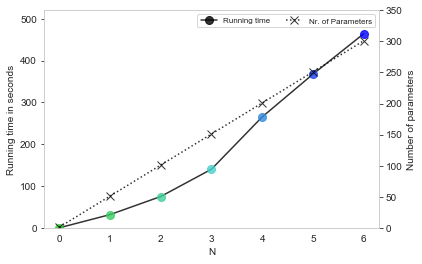}
		
		{\small {\bf (c)} Running time and number of parameters}
		\vspace{0.2cm}
	\end{minipage}
	\begin{minipage}[b][][b]{0.49\textwidth}
		\centering
		\includegraphics[height = 4.9cm]{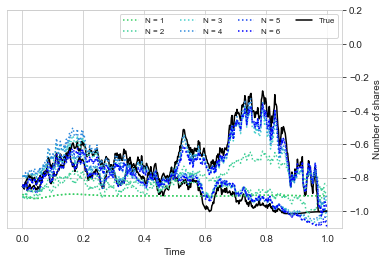}
		
		{\small {\bf (d)} $\theta^1$ and $\vartheta^{\widetilde{\varphi}_{1:N}(\widetilde{\omega}),1}$ for two samples of test set}
		\vspace{0.2cm}
	\end{minipage}
	\caption{Learning the $L^2$-optimal strategy of the Basket option $G$ defined in \eqref{EqBasketOption}, for $N = 0,...,6$. In (a), the the empirical $L^2(\mathbb{P})$-hedging error \eqref{EqAlg1} and the empirical $L^2(dt \otimes \mathbb{P})$-strategy error $\big( \frac{T}{KL} \sum_{l=1}^L \sum_{k=0}^K \big\Vert \theta_{t_k}(\omega_l) - \vartheta^{\widetilde{\varphi}_{1:N}(\widetilde{\omega})}_{t_k}(\omega_l) \big\Vert^2 \big)^{1/2}$ are displayed, where $\theta$ is defined in \eqref{EqBasketHedgStrat}. In (b), the distributions of $G$ (label ``True'') and $\widetilde{\varphi}_0 + \int_0^T \big( \vartheta^{\widetilde{\varphi}_{1:N}(\widetilde{\omega})}_t \big)^\top dX_t$ (label ``$N = 0,...,6$'') are depicted on the test set. In (c), the running time\textsuperscript{\ref{footnote2}} and the number of estimated parameters are shown. In (d), we compare $\theta^1$ (label ``True'') and $\vartheta^{\widetilde{\varphi}_{1:N}(\widetilde{\omega}),1}$ (label ``$N = 1,...,6$'') on two samples of the test set.}
	\label{FigAD}
\end{figure}

\section{Proofs}
\label{SecProofs}

\subsection{Proof of results in Section~\ref{SecExpIntSemimg}}
\label{SecProofsExpIntSemimg}

\begin{proof}[Proof of Example~\ref{ExPolynDiff}]
	We adapt the proof of \cite[Theorem~3.3]{filipovic16} and assume that $X := (X_t)_{t \in [0,T]}$ is an $E$-valued strong solution of the SDE \eqref{EqExPolynDiff0a}, where $E \subseteq \mathbb{R}^d$, where $E \ni x \mapsto b(x) := \beta_0 + \beta_1 x \in \mathbb{R}^d$ is the drift coefficient (with $\beta_0 \in \mathbb{R}^d$ and $\beta_1 \in \mathbb{R}^{d \times d}$), where $\sigma: E \rightarrow \mathbb{R}^{d \times d}$ is the volatility, and where $E \ni x \mapsto a(x) := \sigma(x) \sigma(x)^\top \in \mathbb{S}^d_+$ is the diffusion coefficient. Then, by using variation of constants, we conclude that $X_t = K_t + L_t$ for all $t \in [0,T]$, where $K_t := e^{\beta_1 t} X_0 + \int_0^t e^{\beta_1 (t-s)} \beta_0 ds$ and $L_t := e^{\beta_1 (T-t)} Y_t$, with $Y_t := \int_0^t \sigma^Y_s dB_s$ and $\sigma^Y_s := e^{\beta_1 (T-s)} \sigma\big( K_s + e^{-\beta_1(T-s)} Y_s \big)$. Moreover, by using the constant $C_2 := \sup_{t \in [0,T]} \big\Vert e^{\beta_1 (T-t)} \big\Vert_F > 0$ together with $\Vert \sigma(x) \Vert_F^2 = \trace\big( \sigma(x) \sigma(x)^\top \big) = \trace(a(x)) \leq \sqrt{d} \, \Vert a(x) \Vert_F$ for any $x \in E$, the linear growth condition in \eqref{EqExPolynDiff0b} (with constant $C_1 > 0$), and the constant $C_3 := C_1 C_2^2 \sqrt{d} \sup_{t \in [0,T]} \big( 1 + \Vert K_t \Vert + \big\Vert e^{-\beta_1(T-t)} \big\Vert_F \big) > 0$, it follows for every $t \in [0,T]$ that
	\begin{equation}
		\label{EqExPolynDiffProof1}
		\begin{aligned}
			\left\Vert \sigma^Y_t \right\Vert_F^2 & = \left\Vert e^{\beta_1 (T-t)} \sigma\left( K_t + e^{-\beta_1(T-t)} Y_t \right) \right\Vert_F^2 \leq \left\Vert e^{\beta_1 (T-t)} \right\Vert_F^2 \left\Vert \sigma\left( K_t + e^{-\beta_1(T-t)} Y_t \right) \right\Vert_F^2 \\
			& \leq C_2^2 \sqrt{d} \left\Vert a\left( K_t + e^{-\beta_1(T-t)} Y_t \right) \right\Vert_F \leq C_1 C_2^2 \sqrt{d} \left( 1 + \left\Vert K_t + e^{-\beta_1(T-t)} Y_t \right\Vert \right) \\
			& \leq C_1 C_2^2 \sqrt{d} \left( 1 + \Vert K_t \Vert + \left\Vert e^{-\beta_1(T-t)} \right\Vert_F \Vert Y_t \Vert \right) \leq C_3 \left( 1 + \Vert Y_t \Vert \right).
		\end{aligned}
	\end{equation}
	Now, let $f: \mathbb{R}^d \rightarrow \mathbb{R}$ be a smooth function such that $\sqrt{\Vert y \Vert} \leq f(y)$ for all $y \in \mathbb{R}^d$ with $\Vert y \Vert \leq 1$ as well as $f(y) = \sqrt{1+\Vert y \Vert}$ for all $y \in \mathbb{R}^d$ with $\Vert y \Vert > 1$. Then, the first and second order partial derivatives $\nabla f := \big( \frac{\partial f}{\partial y_i} \big)_{i=1,...,d}^\top: \mathbb{R}^d \rightarrow \mathbb{R}^d$ and $\nabla^2 f := \big( \frac{\partial^2 f}{\partial y_i \partial y_j} \big)_{i,j=1,...,d}: \mathbb{R}^d \rightarrow \mathbb{R}^{d \times d}$ can be computed on $\lbrace y \in \mathbb{R}^d: \Vert y \Vert > 1 \rbrace$ (see \cite[p.~954]{filipovic16}). Moreover, by applying Ito's formula on the process $t \mapsto Z_t := f(Y_t)$, we conclude that $dZ_t = \mu^Z_t dt + \sigma^Z_t dB_t$, $\mathbb{P}$-a.s., for all $t \in [0,T]$, where $\mu^Z_t := \frac{1}{2} \trace\big( \nabla^2 f(Y_t) \sigma^Y_t (\sigma^Y_t)^\top \big)$ and $\sigma^Z_t := \nabla f(Y_t)^\top \sigma^Y_t dB_t$. Hence, by using that $\nabla f: \mathbb{R}^d \rightarrow \mathbb{R}^d$ and $\nabla^2 f: \mathbb{R}^d \rightarrow \mathbb{R}^{d \times d}$ are uniformly bounded on $\mathbb{R}^d$ together with \eqref{EqExPolynDiffProof1}, there exists some $m, \rho > 0$ such that for every $t \in [0,T]$ it holds that $\mu^Z_t \leq m$ and $\big\Vert \sigma^Z_t \big\Vert \leq \rho$. Thus, for $\delta \in (0,1/(2\rho^2 T))$, we can apply \cite[Theorem~1.3]{hajek85} with the convex and non-decreasing function $\mathbb{R} \ni s \mapsto e^{\delta s^2} \in \mathbb{R}$ to conclude for every $t \in [0,T]$ and a corresponding normal distributed random variable $V_t \sim \mathcal{N}(f(0)+mt,\rho^2 t)$ that
	\begin{equation}
		\label{EqExPolynDiffProof2}
		\mathbb{E}\big[ e^{\delta Z_t^2} \big] \leq \mathbb{E}\big[ e^{\delta V_t^2} \big] \leq \mathbb{E}\big[ e^{\delta V_T^2} \big] < \infty.
	\end{equation}
	Next, we fix some $\varepsilon > 0$ with $C_2 \varepsilon \leq \delta$ and $C_1 \varepsilon T \leq \delta$. Then, by using the triangle inequality, the constants $C_4 := \sup_{t \in [0,T]} \Vert K_t \Vert > 0$ and $C_2 > 0$ (defined above), that $\Vert Y_t \Vert \leq f(Y_t)^2 = Z_t^2$ for any $t \in [0,T]$, and the inequality~\eqref{EqExPolynDiffProof2}, it follows for every $t \in [0,T]$ that
	\begin{equation}
		\label{EqExPolynDiffProof3}
		\mathbb{E}\left[ e^{\varepsilon \Vert X_t \Vert} \right] \leq \mathbb{E}\left[ e^{\varepsilon (\Vert K_t \Vert + \Vert L_t \Vert)} \right] \leq e^{C_4 \varepsilon} \mathbb{E}\left[ e^{C_2 \varepsilon \Vert Y_t \Vert} \right] \leq e^{C_4 \varepsilon} \mathbb{E}\left[ e^{\delta Z_t^2} \right] \leq e^{C_4 \varepsilon} \mathbb{E}\left[ e^{\delta V_T^2} \right] < \infty,
	\end{equation}
	which shows that $X$ satisfies Assumption~\ref{AssExpInt}~\ref{AssExpInt1}. Moreover, by using that the finite variation part $A$ satisfies $A_t = \int_0^t b(X_s) ds$ for any $t \in [0,T]$ together with Jensen's inequality and Fubini's theorem, Minkowski's inequality together with Jensen's inequality and $n := \lceil \gamma \rceil \in \mathbb{N}_0$, that $\frac{x^n}{n!} \leq \exp(x)$ for any $n \in \mathbb{N}_0$ and $x \in [0,\infty)$, and the inequality~\eqref{EqExPolynDiffProof3}, it holds for every $\gamma \in [1,\infty)$ that
	\begin{equation*}
		\begin{aligned}
			& \mathbb{E}\left[ \left( \int_0^T \Vert dA_t \Vert \right)^\gamma \right]^\frac{1}{\gamma} \leq T^{1-\frac{1}{\gamma}} \left( \int_0^T \mathbb{E}\left[ \Vert b(X_t) \Vert^\gamma \right] dt \right)^\frac{1}{\gamma} \leq T \sup_{t \in [0,T]} \mathbb{E}\left[ (\Vert \beta_0 \Vert + \Vert \beta_1 \Vert_F \Vert X_t \Vert)^\gamma \right]^\frac{1}{\gamma} \\
			& \quad\quad \leq \Vert \beta_0 \Vert T + \Vert \beta_1 \Vert_F T \sup_{t \in [0,T]} \mathbb{E}\left[ \Vert X_t \Vert^{n} \right]^\frac{1}{n} = \Vert \beta_0 \Vert T + \Vert \beta_1 \Vert_F T \frac{(n!)^\frac{1}{n}}{\varepsilon} \sup_{t \in [0,T]} \mathbb{E}\left[ \frac{(\varepsilon \Vert X_t \Vert)^{n}}{n!} \right]^\frac{1}{n} \\
			& \quad\quad \leq \Vert \beta_0 \Vert T + \Vert \beta_1 \Vert_F T \frac{(n!)^\frac{1}{n}}{\varepsilon} \sup_{t \in [0,T]} \mathbb{E}\left[ e^{\varepsilon \Vert X_t \Vert} \right]^\frac{1}{n} \leq \Vert \beta_0 \Vert T + \Vert \beta_1 \Vert_F T \frac{(n!)^\frac{1}{n}}{\varepsilon} e^{C_4 \frac{\varepsilon}{n}} \mathbb{E}\left[ e^{\delta V_T^2} \right]^\frac{1}{n} < \infty,
		\end{aligned}
	\end{equation*}
	which shows that $X$ satisfies Assumption~\ref{AssExpInt}~\ref{AssExpInt2}. In addition, by using that the local martingale part $M$ satisfies $M_t = \int_0^t \sigma(X_s) dB_s$ for any $t \in [0,T]$ together with Jensen's inequality, the linear growth condition in \eqref{EqExPolynDiff0b} (with constant $C_1 > 0$), Minkowski's inequality together with Jensen's inequality and $n := \lceil \gamma \rceil \in \mathbb{N}_0$, and the inequality~\eqref{EqExPolynDiffProof3}, we conclude for every $\gamma \in [1,\infty)$ that
	\begin{equation*}
		\begin{aligned}
			\mathbb{E}\left[ \Vert \langle M \rangle_T \Vert_F^\gamma \right]^\frac{1}{c} & \leq T^{1-\frac{1}{\gamma}} \left( \int_0^T \mathbb{E}\left[ \Vert a(X_t) \Vert_F^\gamma \right] dt \right)^\frac{1}{\gamma} \leq C_1 T \sup_{t \in [0,T]} \mathbb{E}\left[ \left( 1 + \Vert X_t \Vert \right)^\gamma \right]^\frac{1}{\gamma} \\
			& \leq C_1 T + C_1 T \sup_{t \in [0,T]} \mathbb{E}\left[ \Vert X_t \Vert^n \right]^\frac{1}{n} = C_1 T + C_1 T \frac{(n!)^\frac{1}{n}}{\varepsilon} \sup_{t \in [0,T]} \mathbb{E}\left[ \frac{(\varepsilon \Vert X_t \Vert)^n}{n!} \right]^\frac{1}{n} \\
			& \leq C_1 T + C_1 T \frac{(n!)^\frac{1}{n}}{\varepsilon} \sup_{t \in [0,T]} \mathbb{E}\left[ e^{\varepsilon \Vert X_t \Vert} \right] \leq C_1 T + C_1 T \frac{(n!)^\frac{1}{n}}{\varepsilon} e^{C_4 \frac{\varepsilon}{n}} \mathbb{E}\left[ e^{\delta V_T^2} \right]^\frac{1}{n} < \infty,
		\end{aligned}
	\end{equation*}
	which shows that $X$ satisfies Assumption~\ref{AssExpInt}~\ref{AssExpInt3}.
\end{proof}

\subsection{Proof of results in Section~\ref{SecStochIntItInt}.}

\subsubsection{Proof of results in Section~\ref{SecStochInt}}
\label{SecProofsStochInt}

\begin{proof}[Proof of Lemma~\ref{LemmaLpNorm}]
	For \ref{LemmaLpNorm1}, we first use that the space of c\`adl\`ag functions forms a vector space to conclude that we only need to show that $\Vert \cdot \Vert_{L^p(X)}$ is a norm. By definition, $\Vert \cdot \Vert_{L^p(X)}$ is absolutely homogeneous. Moreover, $\Vert \cdot \Vert_{L^p(X)}$ is positive definite as $L^p(X)$ is defined in terms of equivalence classes. In order to show the triangle inequality, let $f,g \in L^p(X)$. Then, by using the Kunita-Watanabe inequality to the continuous local martingales $t \mapsto M^f_t := \int_0^t f(s)^\top dM_s$ and $t \mapsto M^g_t := \int_0^t g(s)^\top dM_s$, and the identity $x + 2\sqrt{xy} + y = (\sqrt{x}+\sqrt{y})^2$ for any $x,y \geq 0$, it holds $\mathbb{P}$-a.s.~that
	\begin{equation*}
		\begin{aligned}
			& \int_0^T (f(t)+g(t))^\top d\langle M \rangle_t (f(t)+g(t)) \\
			& \quad\quad = \int_0^T f(t)^\top d\langle M \rangle_t f(t) + 2 \underbrace{\int_0^T f(t)^\top d\langle M \rangle_t g(t)}_{= \int_0^T d\langle M^f, M^g \rangle_t} + \int_0^T g(t)^\top d\langle M \rangle_t g(t) \\
			& \quad\quad \leq \big\langle M^f \big\rangle_T + 2 \big\langle M^f \big\rangle_T^\frac{1}{2} \big\langle M^g \big\rangle_T^\frac{1}{2} + \big\langle M^g \big\rangle_T = \left( \big\langle M^f \big\rangle_T^\frac{1}{2} + \big\langle M^g \big\rangle_T^\frac{1}{2} \right)^2.
		\end{aligned}
	\end{equation*}
	Finally, by using the triangle inequality and Minkowski's inequality, we conclude that
	\begin{equation*}
		\begin{aligned}
			& \Vert f + g \Vert_{L^p(X)} = \mathbb{E}\left[ \left( \int_0^T \left\vert (f(t)+g(t))^\top dA_t \right\vert \right)^p \right]^\frac{1}{p} + \mathbb{E}\left[ \left( \int_0^T (f(t)+g(t))^\top d\langle M \rangle_t (f(t)+g(t)) \right)^\frac{p}{2} \right]^\frac{1}{p} \\
			& \quad \leq \mathbb{E}\left[ \left( \int_0^T \left\vert f(t)^\top dA_t \right\vert + \int_0^T \left\vert g(t)^\top dA_t \right\vert \right)^p \right]^\frac{1}{p} + \mathbb{E}\left[ \left( \big\langle M^f \big\rangle_T^\frac{1}{2} + \big\langle M^g \big\rangle_T^\frac{1}{2} \right)^p \right]^\frac{1}{p} \\
			& \quad \leq \mathbb{E}\left[ \left( \int_0^T \left\vert f(t)^\top dA_t \right\vert \right)^p \right]^\frac{1}{p} + \mathbb{E}\left[ \left( \int_0^T \left\vert g(t)^\top dA_t \right\vert \right)^p \right]^\frac{1}{p} + \mathbb{E}\left[ \big\langle M^f \big\rangle_T^\frac{p}{2} \right]^\frac{1}{p} + \mathbb{E}\left[ \big\langle M^g \big\rangle_T^\frac{p}{2} \right]^\frac{1}{p} \\
			& \quad = \Vert f \Vert_{L^p(X)} + \Vert g \Vert_{L^p(X)},
		\end{aligned}
	\end{equation*}
	which shows \ref{LemmaLpNorm1}.
	
	For \ref{LemmaLpNorm2}, we first recall that every fixed $g := (g_1,...,g_d)^\top \in L^p(X)$ is $\mathcal{B}([0,T])/\mathcal{B}(\mathbb{R}^d)$-measurable and bounded (see Remark~\ref{RemBoundedMbl}), i.e.~$\Vert g \Vert_\infty := \sup_{t \in [0,T]} \Vert g(t) \Vert < \infty$. Moreover, by using that $X$ satisfies Assumption~\ref{AssExpInt}, we can define the finite constant
	\begin{equation}
		\label{EqLemmaLpIneqProof1}
		C_{p,T,X} := \mathbb{E}\left[ \left( \int_0^T \Vert dA_t \Vert \right)^p \right]^\frac{1}{p} + \sqrt{3d} \mathbb{E}\left[ \left\Vert \langle M \rangle_T \right\Vert_F^\frac{p}{2} \right]^\frac{1}{p} < \infty.
	\end{equation}
	In addition, by using the polarization identity $\langle M^i, M^j \rangle_t = \frac{1}{2} \big( \langle M^i + M^j \rangle_t - \langle M^i \rangle_t - \langle M^j \rangle_t \big)$ and the inequality $\sum_{i,j=1}^d c_{i,j} \leq d \big( \sum_{i,j=1}^d c_{i,j}^2 \big)^{1/2}$ for any $c_{i,j} \geq 0$, we conclude that
	\begin{equation}
		\label{EqLemmaLpIneqProof2}
		\begin{aligned}
			& \int_0^T g(t)^\top d\langle M \rangle_t g(t) = \frac{1}{2} \left\vert \sum_{i,j=1}^d \int_0^T g_i(t) \left( d\langle M^i + M^j \rangle_t - d\langle M^i \rangle_t - d\langle M^j \rangle_t \right) g_j(t) \right\vert \\
			& \leq \frac{1}{2} \sum_{i,j=1}^d \left( \left\vert \int_0^T g_i(t) d\langle M^i + M^j \rangle_t g_j(t) \right\vert + \left\vert \int_0^T g_i(t) d\langle M^i \rangle_t g_j(t) \right\vert + \left\vert \int_0^T g_i(t) d\langle M^j \rangle_t g_j(t) \right\vert \right) \\
			& \leq \frac{\Vert g \Vert_\infty^2}{2} \sum_{i,j=1}^d \left( \langle M^i + M^j \rangle_T + \langle M^i \rangle_T + \langle M^j \rangle_T \right) \leq 3 \Vert g \Vert_\infty^2 \sum_{i,j=1}^d \vert \langle M^i, M^j \rangle_T \vert \leq 3 d \Vert g \Vert_\infty^2 \Vert \langle M \rangle_T \Vert_F.
		\end{aligned}
	\end{equation}
	Hence, by using the Cauchy-Schwarz inequality and the inequalities~\eqref{EqLemmaLpIneqProof1}+\eqref{EqLemmaLpIneqProof2}, we obtain
	\begin{equation*}
		\begin{aligned}
			\Vert g \Vert_{L^p(X)} & = \mathbb{E}\left[ \left( \int_0^T \left\vert g(t)^\top dA_t \right\vert \right)^p \right]^\frac{1}{p} + \mathbb{E}\left[ \left( \int_0^T g(t)^\top d\langle M \rangle_t g(t) \right)^\frac{p}{2} \right]^\frac{1}{p} \\
			& \leq \Vert g \Vert_\infty \mathbb{E}\left[ \left( \int_0^T \Vert dA_t \Vert \right)^p \right]^\frac{1}{p} + \sqrt{3d} \Vert g \Vert_\infty \mathbb{E}\left[ \left\Vert \langle M \rangle_T \right\Vert_F^\frac{p}{2} \right]^\frac{1}{p} \leq C_{p,T,X} \Vert g \Vert_\infty < \infty,
		\end{aligned}
	\end{equation*}
	which completes the proof.
\end{proof}

\begin{proof}[Proof of Lemma~\ref{LemmaBDG}]
	Fix some $g \in L^p(X)$. Then, by using Minkowski's inequality and the Burkholder-Davis-Gundy inequality (with constant $C_p > 0$), it follows that
	\begin{equation*}
		\begin{aligned}
			\mathbb{E}\left[ \sup_{t \in [0,T]} \left\vert W(g)_t \right\vert^p \right]^\frac{1}{p} & \leq \mathbb{E}\left[ \sup_{t \in [0,T]} \left\vert \int_0^t g(s)^\top dA_s \right\vert^p \right]^\frac{1}{p} + \mathbb{E}\left[ \sup_{t \in [0,T]} \left\vert \int_0^t g(s)^\top dM_s \right\vert^p \right]^\frac{1}{p} \\
			& \leq \mathbb{E}\left[ \left( \int_0^T \left\vert g(s)^\top dA_s \right\vert \right)^p \right]^\frac{1}{p} + C_p \mathbb{E}\left[ \left( \int_0^T g(s)^\top d\langle M \rangle_s g(s) \right)^\frac{p}{2} \right]^\frac{1}{p} \\
			& \leq \max\left( 1, C_p \right) \Vert g \Vert_{L^p(X)},
		\end{aligned}
	\end{equation*}
	which completes the proof.
\end{proof}

\subsubsection{Proof of results in Section~\ref{SecItInt}}
\label{SecProofsItInt}

\begin{proof}[Proof of Lemma~\ref{LemmaLpXnNorm}]
	First, we consider two Banach spaces $(Y,\Vert \cdot \Vert_Y)$ and $(Z,\Vert \cdot \Vert_Z)$. Then, by using \cite[Proposition~2.1]{ryan02}, the tensor product $(Y \otimes Z,\Vert \cdot \Vert_{Y \otimes Z})$ is again a Banach space with projective norm
	\begin{equation*}
		\Vert x \Vert_{Y \otimes Z} = \inf\left\lbrace \sum_{j=1}^m \Vert y_j \Vert_Y \Vert z_j \Vert_Z: x = \sum_{j=1}^m y_j \otimes z_j, \, m \in \mathbb{N}, \, y_1,...,y_m \in Y, \, z_1,...,z_m \in Z \right\rbrace,
	\end{equation*}
	where it holds that $\Vert y \otimes z \Vert_{Y \otimes Z} = \Vert y \Vert_Y \Vert z \Vert_Z$ for all $y \in Y$ and $z \in Z$. Moreover, by following the proof of \cite[Proposition~2.1]{ryan02}, we observe that the completeness of $(Y,\Vert \cdot \Vert_Y)$ and $(Z,\Vert \cdot \Vert_Z)$ has never been used. Hence, if $(Y,\Vert \cdot \Vert_Y)$ and $(Z,\Vert \cdot \Vert_Z)$ are normed vector space, then $(Y \otimes Z,\Vert \cdot \Vert_{Y \otimes Z})$ is a normed vector space with $\Vert y \otimes z \Vert_{Y \otimes Z} = \Vert y \Vert_Y \Vert z \Vert_Z$ for all $y \in Y$ and $z \in Z$.
	
	Thus, by using the above arguments iteratively, we can show that $(L^{np}(X)^{\otimes n},\Vert \cdot \Vert_{L^{np}(X)^{\otimes n}})$ is a normed vector space such that for every $g_1,...,g_n \in L^{np}(X)$ it holds that
	\begin{equation*}
		\Vert g_1 \otimes \cdots \otimes g_n \Vert_{L^{np}(X)^{\otimes n}} = \Vert g_1 \otimes \cdots \otimes g_{n-1} \Vert_{L^{np}(X)^{\otimes (n-1)}} \Vert g_n \Vert_{L^{np}(X)} = ... = \prod_{k=1}^n \Vert g_k \Vert_{L^{np}(X)},
	\end{equation*}
	which completes the proof.
\end{proof}

\begin{proof}[Proof of Lemma~\ref{LemmaItIntLinearBDG}]
	For \ref{LemmaItIntLinearBDG1}, we first define $L^{np}(X) \times \cdots \times L^{np}(X) \ni (g_1,...,g_n) \mapsto L(g_1,...,g_n) := J^\circ_n(g_1 \otimes \cdots \otimes g_n)_t \rightarrow L^0(\Omega,\mathcal{F}_t,\mathbb{P})$, which is by \eqref{EqRemItInt1} multilinear, where $L^0(\Omega,\mathcal{F}_t,\mathbb{P})$ denotes the space of $\mathcal{F}_t$-measurable random variables. Moreover, by using the universal property of the tensor product in \cite[Proposition~1.4]{ryan02}, the map
	\begin{equation*}
		L^{np}(X)^{\otimes n} \ni g = \sum_{j=1}^m g_{j,1} \otimes \cdots \otimes g_{j,n} \quad \mapsto \quad \widetilde{L}(g) := \sum_{j=1}^m L(g_{j,1},...,g_{j,n}) \in \mathbb{R}
	\end{equation*}
	is linear. However, since $J^\circ_n(g)_t = \sum_{j=1}^m J^\circ_n(g_{j,1} \otimes \cdots \otimes g_{j,n})_t = \sum_{j=1}^m L(g_{j,1},...,g_{j,n}) = \widetilde{L}(g)$ for all $g = \sum_{j=1}^m g_{j,1} \otimes \cdots \otimes g_{j,n} \in L^{np}(X)^{\otimes n}$, we conclude that $J^\circ_n(\cdot)_t = \widetilde{L}$ is linear on $L^{np}(X)^{\otimes n}$.
	
	For \ref{LemmaItIntLinearBDG2}, we use the linearity of $J^\circ_n(\cdot)_t: L^{np}(X)^{\otimes n} \rightarrow L^0(\Omega,\mathcal{F}_t,\mathbb{P})$ in \ref{LemmaItIntLinearBDG1} to conclude that it suffices to show that $J^\circ_n(g)_t = J^\circ_n(0)_t = 0$, $\mathbb{P}$-a.s., for each representation $g = \sum_{j=1}^m g_{k,1} \otimes \cdots \otimes g_{j,n}$ of $0 \in L^{np}(X)^{\otimes n}$. For every fixed $E \in \mathcal{A}$, we define the multilinear form $L^{np}(X) \times \cdots \times L^{np}(X) \ni (f_1,...,f_n) \mapsto L(f_1,...,f_n) := \mathbb{E}\left[ \mathds{1}_E J^\circ_n(f_1 \otimes \cdots \otimes f_n)_t \right] \in \mathbb{R}$. Then, it follows for every representation $g := \sum_{j=1}^m g_{j,1} \otimes \cdots \otimes g_{j,n} \in L^{np}(X)^{\otimes n}$ of $0 \in L^{np}(X)^{\otimes n}$ that
	\begin{equation*}
		\mathbb{E}\left[ \mathds{1}_E J^\circ_n(g)_t \right] = \sum_{j=1}^m \mathbb{E}\left[ \mathds{1}_E J^\circ_n(g_{j,1} \otimes \cdots \otimes g_{j,n})_t \right] = \sum_{j=1}^m L(g_{j,1},...,g_{j,n}) = L(g).
	\end{equation*}
	Since the value of $L(g)$ does not dependent on the representation of $0 \in L^{np}(X)^{\otimes n}$ (see \cite[p.~2]{ryan02}), we can choose $g = 0 \in L^{np}(X)^{\otimes n}$ to conclude that $\mathbb{E}\left[ \mathds{1}_E J^\circ_n(g)_t \right] = L(g) = 0$. However, as $E \in \mathcal{A}$ was chosen arbitrarily, we obtain $J^\circ_n(g)_t = 0$, $\mathbb{P}$-a.s., which completes the proof.
	
	For \ref{LemmaItIntLinearBDG3}, we use induction on $n \in \mathbb{N}$. For $n = 1$, we observe that $J^\circ_1(g)_t = W(g)_t$ for all $g \in L^p(X)$ and $t \in [0,T]$, whence the conclusion follows from Lemma~\ref{LemmaBDG}. Now, we fix some $n \in \mathbb{N} \cap [2,\infty)$ and $g_1,...,g_n \in L^{np}(X)$, assume that \ref{LemmaItIntLinearBDG3} holds true for $1,...,n-1$, and aim to show \ref{LemmaItIntLinearBDG3} for $n$. To this end, we use the notation $g_{1:(n-1)} := g_1 \otimes \cdots \otimes g_{n-1} \in L^{np}(X)^{\otimes (n-1)}$ and H\"older's inequality to conclude that
	\begin{equation}
		\label{EqLemmaBDGItIntProof1}
		\begin{aligned}
			& \mathbb{E}\left[ \sup_{t \in [0,T]} \left\vert \int_0^t J^\circ_{n-1}(g_{1:(n-1)})_s g_n(s)^\top dA_s \right\vert^p \right]^\frac{1}{p} \\
			& \quad\quad \leq \mathbb{E}\left[ \sup_{t \in [0,T]} \left\vert J^\circ_{n-1}(g_{1:(n-1)})_t \right\vert^p \left( \int_0^T \left\vert g_k(s)^\top dA_s \right\vert \right)^p \right]^\frac{1}{p} \\
			& \quad\quad \leq \mathbb{E}\left[ \sup_{t \in [0,T]} \left\vert J^\circ_{n-1}(g_{1:(n-1)})_t \right\vert^\frac{np}{n-1} \right]^\frac{n-1}{np} \mathbb{E}\left[ \left( \int_0^T \left\vert g_k(s)^\top dA_s \right\vert \right)^{np} \right]^\frac{1}{np}.
		\end{aligned}
	\end{equation}
	Moreover, the Burkholder-Davis-Gundy inequality (with constant $C_p > 0$) and H\"older's inequality imply that
	\begin{equation}
		\label{EqLemmaBDGItIntProof2}
		\begin{aligned}
			& \mathbb{E}\left[ \sup_{t \in [0,T]} \left\vert \int_0^t J^\circ_{n-1}(g_{1:(n-1)})_s g_k(s)^\top dM_s \right\vert^p \right]^\frac{1}{p} \\
			& \quad\quad \leq C_p \mathbb{E}\left[ \left\vert \int_0^T J^\circ_{n-1}(g_{1:(n-1)})_s^2 g_k(s)^\top d\langle M \rangle_s g_k(s) \right\vert^\frac{p}{2} \right]^\frac{1}{p} \\
			& \quad\quad \leq C_p \mathbb{E}\left[ \sup_{t \in [0,T]} \left\vert J^\circ_{n-1}(g_{1:(n-1)})_t \right\vert^p \left( \int_0^T g_k(s)^\top d\langle M \rangle_s g_k(s) \right)^\frac{p}{2} \right]^\frac{1}{p} \\
			& \quad\quad \leq C_p \mathbb{E}\left[ \sup_{t \in [0,T]} \left\vert J^\circ_{n-1}(g_{1:(n-1)})_t \right\vert^\frac{np}{n-1} \right]^\frac{n-1}{np} \mathbb{E}\left[ \left( \int_0^T g_k(s)^\top d\langle M \rangle_s g_k(s) \right)^\frac{np}{2} \right]^\frac{1}{np}. 
		\end{aligned}
	\end{equation}
	In addition, by using the notation $g_{1:(n-2)} := g_1 \otimes \cdots \otimes g_{n-2} \in L^{np}(X)^{\otimes (n-2)}$, the Kunita-Watanabe inequality, and H\"older's inequality, we conclude that
	\begin{equation}
		\label{EqLemmaBDGItIntProof3}
		\begin{aligned}
			& \mathbb{E}\left[ \sup_{t \in [0,T]} \left\vert \int_0^t J^\circ_{n-2}(g_{1:(n-2)})_s g_{n-1}(s)^\top d\langle X \rangle_s g_n(s) \right\vert^p \right]^\frac{1}{p} \\
			& \quad\quad \leq \mathbb{E}\left[ \sup_{t \in [0,T]} \left\vert J^\circ_{n-2}(g_{1:(n-2)})_t \right\vert^p \left\vert \int_0^T g_{n-1}(s)^\top d\langle M \rangle_s g_n(s) \right\vert^p \right]^\frac{1}{p} \\
			& \quad\quad \leq \mathbb{E}\left[ \sup_{t \in [0,T]} \left\vert J^\circ_{n-2}(g_{1:(n-2)})_t \right\vert^p \left( \int_0^T g_{n-1}(s)^\top d\langle M \rangle_s g_{n-1}(s) \right)^\frac{p}{2} \left( \int_0^T g_n(s)^\top d\langle M \rangle_s g_n(s) \right)^\frac{p}{2} \right]^\frac{1}{p} \\
			& \quad\quad \leq \mathbb{E}\left[ \sup_{t \in [0,T]} \left\vert J^\circ_{n-2}(g_{1:(n-2)})_t \right\vert^\frac{np}{n-2} \right]^\frac{n-2}{np} \mathbb{E}\left[ \left( \int_0^T g_{n-1}(s)^\top d\langle M \rangle_s g_{n-1}(s) \right)^\frac{np}{2} \right]^\frac{1}{np} \\
			& \quad\quad\quad\quad \cdot \mathbb{E}\left[ \left( \int_0^T g_n(s)^\top d\langle M \rangle_s g_n(s) \right)^\frac{np}{2} \right]^\frac{1}{np}.
		\end{aligned}
	\end{equation}
	Hence, by using the notation $g_{1:n} := g_1 \otimes \cdots \otimes g_n \in L^{np}(X)^{\otimes n}$, the identity~\eqref{EqRemItInt1} together with $dX_t = dM_t + dA_t$, $t \in [0,T]$, and Minkowski's inequality, the inequalities~\eqref{EqLemmaBDGItIntProof1}-\eqref{EqLemmaBDGItIntProof3}, and the induction hypothesis (i.e.~that \ref{LemmaItIntLinearBDG3} holds true for $1,...,n-1$ with exponent $\frac{np}{n-1}$) together with $\Vert g_{1:n_1} \Vert_{L^{np}(X)^{\otimes n_1}} = \prod_{k=1}^{n_1} \Vert g_k \Vert_{L^{np}(X)}$ for $n_1 \in \lbrace n-2, n-1 \rbrace$ (see Lemma~\ref{LemmaLpXnNorm}), we obtain that
	\begin{equation}
		\label{EqLemmaBDGItIntProof4}
		\begin{aligned}
			\mathbb{E}\left[ \sup_{t \in [0,T]} \left\vert J^\circ_n(g_{1:n})_t \right\vert^p \right]^\frac{1}{p} & \leq \mathbb{E}\left[ \sup_{t \in [0,T]} \left\vert \int_0^t J^\circ_{n-1}(g_{1:(n-1)})_s g_n(s)^\top dA_s \right\vert^p \right]^\frac{1}{p} \\
			& \quad + \mathbb{E}\left[ \sup_{t \in [0,T]} \left\vert \int_0^t J^\circ_{n-1}(g_{1:(n-1)})_s g_n(s)^\top dM_s \right\vert^p \right]^\frac{1}{p} \\
			& \quad + \frac{1}{2} \mathbb{E}\left[ \sup_{t \in [0,T]} \left\vert \int_0^t J^\circ_{n-2}(g_{1:(n-2)})_s g_{n-1}(s)^\top d\langle X \rangle_s g_n(s) \right\vert^p \right]^\frac{1}{p} \\
			& \leq \mathbb{E}\left[ \sup_{t \in [0,T]} \left\vert J^\circ_{n-1}(g_{1:(n-1)})_t \right\vert^\frac{np}{n-1} \right]^\frac{n-1}{np} \mathbb{E}\left[ \left( \int_0^T \left\vert g_k(s)^\top dA_s \right\vert \right)^{np} \right]^\frac{1}{np} \\
			& \quad + C_p \mathbb{E}\left[ \sup_{t \in [0,T]} \left\vert J^\circ_{n-1}(g_{1:(n-1)})_t \right\vert^\frac{np}{n-1} \right]^\frac{n-1}{np} \mathbb{E}\left[ \left( \int_0^T g_k(s)^\top d\langle M \rangle_s g_k(s) \right)^\frac{np}{2} \right]^\frac{1}{np} \\
			& \quad + \frac{1}{2} \mathbb{E}\left[ \sup_{t \in [0,T]} \left\vert J^\circ_{n-2}(g_{1:(n-2)})_t \right\vert^\frac{np}{n-2} \right]^\frac{n-2}{np} \mathbb{E}\left[ \left( \int_0^T g_{n-1}(s)^\top d\langle M \rangle_s g_{n-1}(s) \right)^\frac{np}{2} \right]^\frac{1}{np} \\
			& \quad\quad\quad \cdot \mathbb{E}\left[ \left( \int_0^T g_n(s)^\top d\langle M \rangle_s g_n(s) \right)^\frac{np}{2} \right]^\frac{1}{np} \\
			& \leq \left( \prod_{k=1}^{n-1} \max\left( 1, C_\frac{(n-1) \frac{np}{n-1}}{k} \right) \right) \left( \prod_{k=1}^{n-1} \Vert g_k \Vert_{L^{np}(X)} \right) \max(1,C_p) \Vert g_n \Vert_{L^{np}(X)} \\
			& \quad + \frac{1}{2} \left( \prod_{k=1}^{n-2} \Vert g_k \Vert_{L^{np}(X)} \right) \Vert g_{n-1} \Vert_{L^{np}(X)} \Vert g_n \Vert_{L^{np}(X)} \\
			& \leq \frac{3}{2} \left( \prod_{k=1}^n \max\left( 1, C_\frac{np}{k} \right) \right) \Vert g_1 \otimes \cdots \otimes g_n \Vert_{L^{np}(X)^{\otimes n}}.
		\end{aligned}
	\end{equation}
	For the general case of $g \in L^{np}(X)^{\otimes n}$, we fix some $\varepsilon > 0$ and assume that $\sum_{j=1}^m g_{j,1} \otimes \cdots \otimes g_{j,n}$ is a representation of $g \in L^{np}(X)^{\otimes n}$ such that $\sum_{j=1}^m \prod_{k=1}^n \Vert g_{j,k} \Vert_{L^{np}(X)} \leq \Vert g \Vert_{L^{np}(X)^{\otimes n}} + \varepsilon/C_{n,p}$. Then, by using Minkowski's inequality and \eqref{EqLemmaBDGItIntProof4}, it follows that
	\begin{equation*}
		\begin{aligned}
			\mathbb{E}\left[ \sup_{t \in [0,T]} \left\vert J^\circ_n(g)_t \right\vert^p \right]^\frac{1}{p} & \leq \sum_{j=1}^m \mathbb{E}\left[ \sup_{t \in [0,T]} \left\vert J^\circ_n(g_{j,1} \otimes \cdots \otimes g_{j,n})_t \right\vert^p \right]^\frac{1}{p} \\
			& \leq C_{n,p} \sum_{j=1}^m \prod_{k=1}^n \Vert g_{j,k} \Vert_{L^{np}(X)} \\
			& \leq C_{n,p} \Vert g \Vert_{L^{np}(X)^{\otimes n}} + \varepsilon.
		\end{aligned}
	\end{equation*}
	Since $\varepsilon > 0$ was chosen arbitrarily, we obtain the conclusion in \ref{LemmaItIntLinearBDG3}.
\end{proof}

\begin{proof}[Proof of Proposition~\ref{PropMon}]
	Fix some $g = \sum_{j=1}^m g_{j,1} \otimes \cdots \otimes g_{j,n} \in L^{np}_{\sym}(X)^{\otimes n}$ and $t \in [0,T]$. Then, we show the conclusion by induction over $n \in \mathbb{N}$. For $n = 1$, we observe that
	\begin{equation*}
		J^\circ_1(g)_t = \sum_{j=1}^m \int_0^t \circ dW(g_{j,1})_t = \sum_{j=1}^m W(g_{j,1})_t.
	\end{equation*}
	Now, we fix some $n \in \mathbb{N} \cap [2,\infty)$, assume that the conclusion holds true for $n-1$, and aim to show that it also holds true for $n$. Indeed, by using that $g = \sym(g)$ together with the linearity of $J^\circ_n: L^{np}(X)^{\otimes n} \rightarrow L^p(\Omega,\mathcal{F}_t,\mathbb{P})$, the set $\mathcal{S}_{n,k}$ consisting of all permutations $\sigma: \lbrace 1,...,n \rbrace \setminus \lbrace k \rbrace \rightarrow \lbrace 1,...,n \rbrace \setminus \lbrace k \rbrace$, $l = 1,...,n$, the induction hypothesis (i.e.~that the conclusion holds true for $n-1$), and the definition of the Stratonovich integral, it follows that
	\begin{equation*}
		\begin{aligned}
			J^\circ_n(g) & = J^\circ_n(\sym(g)) = \frac{1}{n!} \sum_{\sigma \in \mathcal{S}_n} \sum_{j=1}^m J^\circ_n(g_{j,\sigma(1)} \otimes \cdots \otimes g_{j,\sigma(n)}) \\
			& = \frac{1}{n!} \sum_{j=1}^m \sum_{k=1}^n \sum_{\sigma \in \mathcal{S}_{n,k}} \int_0^t J^\circ_{n-1}\left( g_{j,\sigma(1)} \otimes \cdots \otimes g_{j,\sigma(k)-1} \otimes g_{j,\sigma(k)+1} \otimes \cdots \otimes g_{j,\sigma(n)} \right)_s \circ dW(g_{j,\sigma(k)})_s \\
			& = \frac{1}{n} \sum_{j=1}^m \sum_{k=1}^n \int_0^t J^\circ_{n-1}\left( \sym\left( g_{j,1} \otimes \cdots \otimes g_{j,k-1} \otimes g_{j,k+1} \otimes \cdots \otimes g_{j,n} \right) \right)_s \circ dW(g_{j,k})_s \\
			& = \frac{1}{n} \sum_{j=1}^m \sum_{k=1}^n \int_0^t \frac{\prod_{l=1, \, l \neq k}^n W(g_{j,l})_s}{(n-1)!} \circ dW(g_{j,k})_s \\
			& = \frac{1}{n!} \sum_{j=1}^m \prod_{k=1}^n W(g_{j,k})_t,
		\end{aligned}
	\end{equation*}
	which completes the proof.
\end{proof}

\subsection{Proof of results in Section~\ref{SecChaosHedg}}

\subsubsection{Proof of results in Section~\ref{SecChaos}}
\label{SecProofsChaos}

\begin{proof}[Proof of Proposition~\ref{PropPolDense}]
	Let $p \in [1,\infty)$ and $q \in (1,\infty]$ with $\frac{1}{p} + \frac{1}{q} = 1$, where $\frac{1}{\infty} := 0$. Now, we assume by contradiction that $\Pol(X)$ defined in \eqref{EqPropPolDense1} is not dense in $L^p(\Omega,\mathcal{F}_T,\mathbb{P})$. Then, by the Hahn-Banach theorem, there exists a non-zero continuous linear functional $l: L^p(\Omega,\mathcal{F}_T,\mathbb{P}) \rightarrow \mathbb{R}$ that vanishes on $\Pol(X)$. Moreover, by the Riesz representation theorem, there exists some $Z \in L^q(\Omega,\mathcal{F}_T,\mathbb{P})$ such that for every $Y \in L^p(\Omega,\mathcal{F}_T,\mathbb{P})$ it holds that $l(Y) = \mathbb{E}[ZY]$. Hence, for every $Y \in \Pol(X)$, we have $l(Y) = \mathbb{E}[ZY] = 0$, e.g., $l(1) = \mathbb{E}[Z] = 0$. Thus, it suffices to show that $Z = 0 \in L^q(\Omega,\mathcal{F}_T,\mathbb{P})$, which contradicts the assumption that the continuous linear functional $l: L^p(\Omega,\mathcal{F}_T,\mathbb{P}) \rightarrow \mathbb{R}$ is non-zero. 
	
	To this end, we define for every fixed $m \in \mathbb{N}$, $i_{1:m} := (i_1,....,i_m) \in \lbrace 1,...,d \rbrace^m$, and $t_{1:m} := (t_1,...,t_m) \in [0,T]^m$ the random vector $\Omega \ni \omega \mapsto X^{i_{1:m}}_{t_{1:m}}(\omega) := \big( X^{i_1}_{t_1}(\omega),...,X^{i_m}_{t_m}(\omega) \big)^\top \in \mathbb{R}^m$ and show that $\mathbb{E}\big[Z f\big( X^{i_{1:m}}_{t_{1:m}} \big)\big] = 0$ for all $f \in \mathcal{S}(\mathbb{R}^m;\mathbb{C})$, where $\mathcal{S}(\mathbb{R}^m;\mathbb{C})$ denotes the Schwartz space (see \cite[p.~330]{folland92}). For this purpose, we define the tempered distribution $\big( f \mapsto U(f) := \mathbb{E}\big[ Z f\big( X^{i_{1:m}}_{t_{1:m}} \big) \big] \big) \in \mathcal{S}'(\mathbb{R}^m;\mathbb{C})$, where $\mathcal{S}'(\mathbb{R}^m;\mathbb{C})$ denotes the dual space of $\mathcal{S}(\mathbb{R}^m;\mathbb{C})$ (see \cite[p.~332]{folland92}). Moreover, we denote by $\big( f \mapsto \widehat{U}(f) := U(\widehat{f}) \big) \in \mathcal{S}'(\mathbb{R}^m;\mathbb{C})$ its Fourier transform in the sense of distributions (see \cite[Equation~9.28]{folland92}), where $\mathbb{R}^m \ni \lambda \mapsto \widehat{f}(\lambda) := \int_{\mathbb{R}^m} e^{-i\lambda^\top x} f(x) dx \in \mathbb{R}$ is the Fourier transform of $f \in \mathcal{S}(\mathbb{R}^m;\mathbb{C})$ (see \cite[Equation~7.1]{folland92}). Then, Fubini's theorem implies for every $f \in \mathcal{S}(\mathbb{R}^m;\mathbb{C})$ that
	\begin{equation}
		\label{EqPropPolDenseProof1}
		\widehat{U}(f) = U(\widehat{f}) = \mathbb{E}\left[ Z \widehat{f}\big( X^{i_{1:m}}_{t_{1:m}} \big) \right] = \mathbb{E}\left[ Z \int_{\mathbb{R}^m} e^{-\mathbf{i} \, (X^{i_{1:m}}_{t_{1:m}})^\top x} f(x) dx \right] = \int_{\mathbb{R}^m} \widehat{u}(x) f(x) dx,
	\end{equation}
	where $\mathbb{R}^m \ni x \mapsto \widehat{u}(x) := \mathbb{E}\big[ Z e^{-\mathbf{i} x^\top X^{i_{1:m}}_{t_{1:m}}} \big] \in \mathbb{C}$. Moreover, by using the Cauchy-Schwarz inequality, that $\big\Vert X^{i_{1:m}}_{t_{1:m}} \big\Vert \leq \sum_{k=1}^m \big\vert X^{i_k}_{t_k} \big\vert \leq \sum_{k=1}^m \Vert X_{t_k} \Vert$, the generalized H\"older inequality (with exponents $\frac{1}{q} + \sum_{k=1}^m \frac{1}{mp} = 1$), and that $X$ satisfies Assumption~\ref{AssExpInt} (with $\varepsilon > 0$), we conclude for every $z \in V := \big\lbrace z \in \mathbb{C}^m: \Vert \im(z) \Vert < \frac{\varepsilon}{mp} \big\rbrace$ (with $\im(z) := (\im(z_1),...,\im(z_m))^\top$, for $z \in \mathbb{C}^m$) that
	\begin{equation*}
		\begin{aligned}
			& \mathbb{E}\left[ \left\vert Z e^{-\mathbf{i} z^\top X^{i_{1:m}}_{t_{1:m}}} \right\vert \right] \leq \mathbb{E}\left[ \vert Z \vert e^{\Vert \im(z) \Vert \big\Vert X^{i_{1:m}}_{t_{1:m}} \big\Vert} \right]^\frac{1}{p} \leq \mathbb{E}\left[ \vert Z \vert e^{\Vert \im(z) \Vert \sum_{k=1}^m \Vert X_{t_k} \Vert} \right]^\frac{1}{p} \\
			& \quad\quad \leq \Vert Z \Vert_{L^q(\mathbb{P})} \prod_{k=1}^m \mathbb{E}\left[ e^{mp \Vert \im(z) \Vert \Vert X_{t_k} \Vert} \right]^\frac{1}{mp} \leq \Vert Z \Vert_{L^q(\mathbb{P})} \prod_{k=1}^m \mathbb{E}\left[ e^{\varepsilon \Vert X_{t_k} \Vert} \right]^\frac{1}{mp} < \infty.
		\end{aligned}
	\end{equation*}
	Therefore, the function $V \ni z \mapsto \widehat{v}(z) := \mathbb{E}\big[ Z e^{-\mathbf{i} z^\top X^{i_{1:m}}_{t_{1:m}}} \big] \in \mathbb{R}$ is holomorphic on $V$, which implies that the restriction $\widehat{u} := \widehat{v}\vert_{\mathbb{R}^m}: \mathbb{R}^m \rightarrow \mathbb{C}$ is real analytic. Moreover, by using that $\mathbb{E}[ZY] = 0$ for any $Y \in \Pol(X)$, it follows for every $n \in \mathbb{N}_0$ and $(j_1,...,j_n) \in \lbrace 1,...,d \rbrace^n$ that
	\begin{equation*}
		\frac{\partial^n \widehat{u}}{\partial z_{j_1} \cdots \partial z_{j_n}}(0) = (-\mathbf{i})^n \mathbb{E}\left[ Z X^{i_{j_1}}_{t_{j_1}} \cdots X^{i_{j_n}}_{t_{j_n}} \right] = 0.
	\end{equation*}
	Hence, by using that $\widehat{u}: \mathbb{R}^m \rightarrow \mathbb{C}$ is real analytic, the function $\widehat{u}: \mathbb{R}^m \rightarrow \mathbb{C}$ is constant equal to $\widehat{u}(0) = \mathbb{E}[Z] = 0$, which implies together with \eqref{EqPropPolDenseProof1} that $\widehat{U}(f) = \int_{\mathbb{R}^m} \widehat{u}(x) f(x) dx = \widehat{u}(0) \int_{\mathbb{R}^m} f(x) dx = 0$ for all $f \in \mathcal{S}(\mathbb{R}^m;\mathbb{C})$. Thus, by using that the Fourier transform is bijective on $\mathcal{S}'(\mathbb{R}^m;\mathbb{C})$, it follows that $U = 0 \in \mathcal{S}'(\mathbb{R}^m;\mathbb{C})$, which shows that $\mathbb{E}\big[Z f\big( X^{i_{1:m}}_{t_{1:m}} \big)\big] = 0$ for all $f \in \mathcal{S}(\mathbb{R}^m;\mathbb{C})$.
	
	Finally, we show that $Z = 0 \in L^q(\Omega,\mathcal{F}_T,\mathbb{P})$. To this end, we fix some $m \in \mathbb{N}$, $i_{1:m} := (i_1,....,i_m) \in \lbrace 1,...,d \rbrace^m$, $t_{1:m} := (t_1,...,t_m) \in [0,T]^m$, and $F \in \mathcal{B}(\mathbb{R}^m)$ and define again the random vector $\Omega \ni \omega \mapsto X^{i_{1:m}}_{t_{1:m}}(\omega) := \big( X^{i_1}_{t_1}(\omega),...,X^{i_m}_{t_m}(\omega) \big)^\top \in \mathbb{R}^m$. Then, by using Lusin's theorem (see \cite[Theorem~3.14]{rudin91}) and a mollification argument, the indicator function $\mathds{1}_F: \mathbb{R}^m \rightarrow \mathbb{R}$ can be approximated by a sequence of smooth compactly supported functions $(f_n)_{n \in \mathbb{N}} \subseteq C^\infty_c(\mathbb{R}^m) \subseteq \mathcal{S}(\mathbb{R}^m;\mathbb{C})$ with respect to $\Vert \cdot \Vert_{L^p(\mathbb{P} \circ (X^{i_{1:m}}_{t_{1:m}})^{-1})}$, where $\mathcal{B}(\mathbb{R}^m) \ni \widetilde{F} \mapsto \big( \mathbb{P} \circ (X^{i_{1:m}}_{t_{1:m}})^{-1} \big)\big[\widetilde{F}\big] := \mathbb{P}\big[ X^{i_{1:m}}_{t_{1:m}} \in \widetilde{F} \big] \in [0,1]$ denotes the push-forward measure. Hence, by using the triangle inequality, that $\mathbb{E}\big[Z f\big( X^{i_{1:m}}_{t_{1:m}} \big)\big] = 0$ for any $f \in \mathcal{S}(\mathbb{R}^m;\mathbb{C})$, and H\"older's inequality, we conclude that
	\begin{equation}
		\label{EqPropPolDenseProof2}
		\begin{aligned}
			\bigg\vert \mathbb{E}\left[ Z \mathds{1}_F\big(X^{i_{1:m}}_{t_{1:m}}\big)\right] \bigg\vert & \leq \lim_{n \rightarrow \infty} \bigg\vert \mathbb{E}\left[Z (\mathds{1}_F-f_n)\big(X^{i_{1:m}}_{t_{1:m}}\big)\right] \bigg\vert + \lim_{n \rightarrow \infty} \bigg\vert \underbrace{\mathbb{E}\left[Z f_n\big(X^{i_{1:m}}_{t_{1:m}}\big)\right]}_{=0} \bigg\vert \\
			& \leq \Vert Z \Vert_{L^q(\mathbb{P})} \lim_{n \rightarrow \infty} \big\Vert (\mathds{1}_F-f_n)\big(X^{i_{1:m}}_{t_{1:m}}\big) \big\Vert_{L^p(\mathbb{P})} \\
			& = \Vert Z \Vert_{L^q(\mathbb{P})} \lim_{n \rightarrow \infty} \Vert \mathds{1}_F-f_n \Vert_{L^p(\mathbb{P} \circ (X^{i_{1:m}}_{t_{1:m}})^{-1})} = 0.
		\end{aligned}
	\end{equation}
	Since $\mathbb{F} := (\mathcal{F}_t)_{t \in [0,T]}$ is the usual $\mathbb{P}$-augmented filtration$^\text{\ref{FootnoteAugmFilt}}$ generated by $X$, the $\sigma$-algebra $\mathcal{F}_T$ is equal to the $\mathbb{P}$-completion of $\sigma(\lbrace X_t: t \in [0,T] \rbrace)$. Thus, by using that $X$ has continuous sample paths, $\mathcal{F}_T$ is equal to the $\mathbb{P}$-completion of $\sigma\big( \bigcup_{m \in \mathbb{N}} \bigcup_{i_{1:m} \in \lbrace 1,...,d \rbrace^m} \bigcup_{t_{1:m} \in [0,T]^m} \big\lbrace \big( X^{i_{1:m}}_{t_{1:m}} \big)^{-1}(F): F \in \mathcal{B}(\mathbb{R}^m) \big\rbrace \big)$. Hence, by applying the Dynkin $\pi$-$\lambda$ theorem (see \cite[Theorem~2.1.2]{durrett10}, with corresponding $\pi$-system $\bigcup_{m \in \mathbb{N}} \bigcup_{i_{1:m} \in \lbrace 1,...,d \rbrace^m} \bigcup_{t_{1:m} \in [0,T]^m} \big\lbrace \big( X^{i_{1:m}}_{t_{1:m}} \big)^{-1}(F): F \in \mathcal{B}(\mathbb{R}^m) \big\rbrace$ and $\lambda$-system $\sigma\big(\big\lbrace E \in \mathcal{F}_T: \mathbb{E}[Z \mathds{1}_E] = 0 \big\rbrace \big)$), we conclude from \eqref{EqPropPolDenseProof2} that $\mathbb{E}[Z \mathds{1}_E] = 0$ for all $E \in \mathcal{F}_T$. This shows that $Z = 0$, $\mathbb{P}$-a.s., and thus $Z = 0 \in L^q(\Omega,\mathcal{F}_T,\mathbb{P})$.
\end{proof}

\begin{proof}[Proof of Theorem~\ref{ThmChaos}]
	For \ref{ThmChaos1}, we fix some $g_n = \sum_{j=1}^{m_n} g_{n,j}^{\otimes n} \in L^{np}_{\operatorname{diag}}(X)^{\otimes n}$. Then, by using the convention $g_{n,j}^{\otimes 0} := 1$ (if $n = 1$), Remark~\ref{RemItInt}, and Proposition~\ref{PropMon}, it follows $\mathbb{P}$-a.s.~that
	\begin{equation*}
		\begin{aligned}
			J^\circ_n(g_n)_T & = \sum_{j=1}^{m_n} J^\circ_n\left( g_{n,j}^{\otimes n} \right)_T = \sum_{j=1}^{m_n} \int_0^T J^\circ_{n-1}\left( g_{n,j}^{\otimes (n-1)} \right)_t \circ dW(g_{n,j})_t \\
			& = \sum_{j=1}^{m_n} \int_0^T J^\circ_{n-1}\left( g_{n,j}^{\otimes (n-1)} \right)_t dW(g_{n,j})_t + \frac{1}{2} \sum_{j=1}^{m_n} \int_0^T J^\circ_{n-1}\left( g_{n,j}^{\otimes (n-1)} \right)_t g_{n,j}(t)^\top d\langle X \rangle_t g_{n,j}(t) \\
			& = \sum_{j=1}^{m_n} \int_0^T \frac{W(g_{n,j})_t^{n-1}}{(n-1)!} g_{n,j}(t)^\top dX_t + \frac{1}{2} \sum_{j=1}^{m_n} \int_0^T \frac{W(g_{n,j})_t^{n-1}}{(n-1)!} g_{n,j}(t)^\top d\langle X \rangle_t g_{n,j}(t),
		\end{aligned}
	\end{equation*}
	which completes the proof of \ref{ThmChaos1}.
	
	For \ref{ThmChaos2}, let $p \in [1,\infty)$ and $q \in (1,\infty]$ with $\frac{1}{p} + \frac{1}{q} = 1$, where $\frac{1}{\infty} := 0$. Now, we assume by contradiction that $\bigoplus_{n \in \mathbb{N}_0} \big\lbrace J^\circ_n(g_n)_T: g_n \in L^{np}_{\diag}(X)^{\otimes n} \big\rbrace$ is not dense in $L^p(\Omega,\mathcal{F}_T,\mathbb{P})$. Then, by the Hahn-Banach theorem, there exists a non-zero continuous linear functional $l: L^p(\Omega,\mathcal{F}_T,\mathbb{P}) \rightarrow \mathbb{R}$ that vanishes on $\bigoplus_{n \in \mathbb{N}_0} \big\lbrace J^\circ_n(g_n)_T: g_n \in L^{np}_{\diag}(X)^{\otimes n} \big\rbrace$. Moreover, by the Riesz representation theorem, there exists some $Z \in L^q(\Omega,\mathcal{F}_T,\mathbb{P})$ such that for every $Y \in L^p(\Omega,\mathcal{F}_T,\mathbb{P})$ it holds that $l(Y) = \mathbb{E}[ZY]$, e.g., $l(Y) = \mathbb{E}[ZY] = 0$ for all $Y \in \bigoplus_{n \in \mathbb{N}_0} \big\lbrace J^\circ_n(g_n)_T: g_n \in L^{np}_{\diag}(X)^{\otimes n} \big\rbrace$. Thus, it suffices to show that $Z = 0 \in L^q(\Omega,\mathcal{F}_T,\mathbb{P})$, which contradicts the assumption that $l: L^p(\Omega,\mathcal{F}_T,\mathbb{P}) \rightarrow \mathbb{R}$ is non-zero.
	
	To this end, we fix some $m \in \mathbb{N}$, $(i_1,...,i_m) \in \lbrace 1,...,d \rbrace^m$, $\mathbf{k} := (k_1,...,k_m) \in \mathbb{N}_0^m$, and $(t_1,...,t_m) \in [0,T]^m$. Moreover, we recall that $X$ satisfies Assumption~\ref{AssExpInt} (with $\varepsilon > 0$) and let $\delta > 0$ such that $\delta \leq \frac{\varepsilon}{2mp}$. In addition, for every fixed $\lambda := (\lambda_1,...,\lambda_m)^\top \in (-\delta,\delta)^m$, we define the simple function	
	\begin{equation*}
		[0,T] \ni t \quad \mapsto \quad g_\lambda(t) := \sum_{l=1}^m \lambda_l \mathds{1}_{[0,t_l]}(t) e_{i_l} \in \mathbb{R}^d,
	\end{equation*}
	where $e_i$ denotes the $i$-th unit vector of $\mathbb{R}^d$, $i = 1,...,d$. Then, for every $\lambda := (\lambda_1,...,\lambda_m)^\top \in (-\delta,\delta)^m$ and $t \in [0,T]$, we observe that
	\begin{equation}
		\label{EqThmChaosProof1}
		W(g_\lambda)_t = \sum_{l=1}^m \lambda_l \int_0^t \mathds{1}_{[0,t_l]}(s) e_{i_l}^\top dX_s = \sum_{l=1}^m \lambda_l \left( X^{i_l}_{t \wedge t_l} - X^{i_l}_0 \right).
	\end{equation}
	Furthermore, by using that $g_\lambda: [0,T] \rightarrow \mathbb{R}^d$ is c\`agl\`ad, Lemma~\ref{LemmaLpNorm}~\ref{LemmaLpNorm2} ensures that $g_\lambda \in L^{np}(X)$, for all $n \in \mathbb{N}$. Furthermore, by using the identity~\eqref{EqThmChaosProof1}, that $X_0$ is $\mathbb{P}$-a.s.~constant equal to some $x_0 \in \mathbb{R}^d$, H\"older's inequality, that $\delta m p \leq \varepsilon$, and that $X$ satisfies Assumption~\ref{AssExpInt} (with $\varepsilon > 0$), it holds that
	\begin{equation}
		\label{EqThmChaosProof2}
		\begin{aligned}
			& \mathbb{E}\left[ \vert Z \vert e^{\vert W(g_\lambda)_T \vert} \right] \leq \mathbb{E}\left[ \vert Z \vert e^{\sum_{l=1}^m \vert \lambda_l \vert \big\vert X^{i_l}_{t_l} - X^{i_l}_0 \big\vert} \right] \leq e^{\delta m \Vert x_0 \Vert} \max_{l=1,...,m} \mathbb{E}\left[ \vert Z \vert e^{\delta m \Vert X_{t_l} \Vert} \right] \\
			& \quad\quad \leq \Vert Z \Vert_{L^q(\mathbb{P})} e^{\delta m \Vert x_0 \Vert} \max_{l=1,...,m} \mathbb{E}\left[ e^{\delta m p \Vert X_{t_l} \Vert} \right]^\frac{1}{p} \leq \Vert Z \Vert_{L^q(\mathbb{P})} e^{\delta m \Vert x_0 \Vert} \max_{l=1,...,m} \mathbb{E}\left[ e^{\varepsilon \Vert X_{t_l} \Vert} \right]^\frac{1}{p} < \infty.
		\end{aligned}
	\end{equation}
	Hence, by using that $\sum_{n=0}^N \frac{W(g_\lambda)_T^n}{n!}$ converges pointwise to $e^{W(g_\lambda)_T}$ together with $\big\vert Z \sum_{n=0}^N \frac{W(g_\lambda)_T^n}{n!} \big\vert \leq \vert Z \vert \sum_{n=0}^N \frac{\vert W(g_\lambda)_T \vert^n}{n!} \leq \vert Z \vert e^{\vert W(g_\lambda)_T \vert}$, where the latter is by \eqref{EqThmChaosProof2} integrable, we can use the dominated convergence theorem, Proposition~\ref{PropMon}, and that $J^\circ_n\left( g_\lambda^{\otimes n} \right)_T \in \bigoplus_{n \in \mathbb{N}_0} \big\lbrace J^\circ_n(g_n)_T: g_n \in L^{np}_{\diag}(X)^{\otimes n} \big\rbrace$ to conclude that
	\begin{equation}
		\label{EqThmChaosProof3}
		u(\lambda) := \mathbb{E}\left[ Z e^{W(g_\lambda)_T} \right] = \lim_{N \rightarrow \infty} \mathbb{E}\left[ Z \sum_{n=0}^N \frac{W(g_\lambda)_T^n}{n!} \right] = \lim_{N \rightarrow \infty} \sum_{n=0}^N \mathbb{E}\left[ Z J^\circ_n\left( g_\lambda^{\otimes n} \right)_T \right]  = 0.
	\end{equation}
	Since $\lambda \in (-\delta,\delta)^m$ was chosen arbitrarily, this shows that $u: (-\delta,\delta)^m \rightarrow \mathbb{R}$ vanishes on $(-\delta,\delta)^m$.
	
	Next, we compute the $\vert \mathbf{k} \vert$-th order partial derivative of $u: (-\delta,\delta)^m \rightarrow \mathbb{R}$, $k_l$-times into the direction $e_l$, $l = 1,...,m$, and aim to show that we can interchange them with the expectation of the random function $(-\delta,\delta)^m \times \Omega \ni (\lambda,\omega) \mapsto U(\lambda,\omega) := Z e^{W(g_\lambda)_T(\omega)} \in \mathbb{R}$, i.e.~that
	\begin{equation}
		\label{EqThmChaosProof4}
		\partial_{\mathbf{k}} u(\lambda) := \frac{\partial^{\vert \mathbf{k} \vert} u}{\partial \lambda_1^{k_1} \cdots \partial \lambda_m^{k_m}}(\lambda) = \mathbb{E}\left[ \frac{\partial^{\vert \mathbf{k} \vert} U}{\partial \lambda_1^{k_1} \cdots \partial \lambda_m^{k_m}}(\lambda,\cdot) \right] =: \mathbb{E}\left[ \partial_{\mathbf{k}} U(\lambda,\cdot) \right].
	\end{equation}
	If \eqref{EqThmChaosProof4} holds true, then we can immediately conclude from \eqref{EqThmChaosProof3} and \eqref{EqThmChaosProof1} that
	\begin{equation}
		\label{EqThmChaosProof5}
		\begin{aligned}
			0 & = \frac{\partial^{\vert \mathbf{k} \vert} u}{\partial \lambda_1^{k_1} \cdots \partial \lambda_m^{k_m}}(0) = \mathbb{E}\left[ \frac{\partial^{\vert \mathbf{k} \vert} U}{\partial \lambda_1^{k_1} \cdots \partial \lambda_m^{k_m}}(0,\cdot) \right] = \mathbb{E}\left[ \frac{\partial^{\vert \mathbf{k} \vert}}{\partial \lambda_1^{k_1} \cdots \partial \lambda_m^{k_m}} \left( Z e^{W(g_\lambda)_T} \right) \bigg\vert_{\lambda = 0} \right] \\
			& = \mathbb{E}\left[ Z \frac{\partial^{\vert \mathbf{k} \vert}}{\partial \lambda_1^{k_1} \cdots \partial \lambda_m^{k_m}} \left( \prod_{l=1}^m e^{\lambda_l \big( X^{i_l}_{t_l} - X^{i_l}_0 \big)} \right) \bigg\vert_{\lambda = 0} \right] = \mathbb{E}\left[ Z \prod_{l=1}^m \left( e^{\lambda_l \big( X^{i_l}_{t_l} - X^{i_l}_0 \big)} \left( X^{i_l}_{t_l} - X^{i_l}_0 \right)^{k_l} \right) \bigg\vert_{\lambda = 0} \right] \\
			& = \mathbb{E}\left[ Z \left( X^{i_1}_{t_1} - X^{i_1}_0 \right)^{k_1} \cdots \left( X^{i_m}_{t_m} - X^{i_m}_0 \right)^{k_m} \right].
		\end{aligned}
	\end{equation}
	In order to prove \eqref{EqThmChaosProof4}, we first use again \eqref{EqThmChaosProof1}, the inequality $e^{\delta \vert x \vert} \vert x \vert^j = \frac{j!}{\delta^j} e^{\delta \vert x \vert} \frac{\vert \delta x \vert^j}{j!} \leq \frac{j!}{\delta^j} e^{2 \delta \vert x \vert}$ for any $j \in \mathbb{N}_0$ and $x \in \mathbb{R}$, that $X_0$ is $\mathbb{P}$-a.s.~constant equal to $x_0 \in \mathbb{R}^d$, H\"older's inequality, that $2\delta mp \leq \varepsilon$ together with $X$ satisfying Assumption~\ref{AssExpInt} (with $\varepsilon > 0$) to obtain for every $\mathbf{j} := (j_1,...,j_m) \in \mathbb{N}^m_0$ that
	\begin{equation}
		\label{EqThmChaosProof6}
		\begin{aligned}
			& \mathbb{E}\left[ \sup_{\widetilde{\lambda} \in (-\delta,\delta)^m} \left\vert \partial_{\mathbf{j}} U(\widetilde{\lambda},\cdot) \right\vert \right] = \mathbb{E}\left[ \sup_{\widetilde{\lambda} \in (-\delta,\delta)^m} \left\vert \frac{\partial^{\vert \mathbf{j} \vert}}{\partial \lambda_1^{j_1} \cdots \partial \lambda_m^{j_m}} \left( Z e^{W(g_\lambda)_T} \right) \bigg\vert_{\lambda = \widetilde{\lambda}} \right\vert \right] \\
			& = \mathbb{E}\left[ \vert Z \vert \sup_{\widetilde{\lambda} \in (-\delta,\delta)^m} \left\vert \prod_{l=1}^m \left( e^{\widetilde{\lambda}_l \big( X^{i_l}_{t_l} - X^{i_l}_0 \big)} \left( X^{i_l}_{t_l} - X^{i_l}_0 \right)^{j_l} \right) \right\vert \right] \leq \frac{\prod_{l=1}^m j_l!}{\delta^{\vert \mathbf{j} \vert}} \mathbb{E}\left[ \vert Z \vert e^{2 \delta \sum_{l=1}^m \big\vert X^{i_l}_{t_l} - X^{i_l}_0 \big\vert} \right] \\
			& \leq \frac{\prod_{l=1}^m j_l!}{\delta^{\vert \mathbf{j} \vert}} e^{2 \delta m \Vert x_0 \Vert} \max_{l=1,...,m} \mathbb{E}\left[ \vert Z \vert e^{2 \delta m \Vert X_{t_l} \Vert} \right] \leq \frac{\prod_{l=1}^m j_l!}{\delta^{\vert \mathbf{j} \vert}} \Vert Z \Vert_{L^q(\mathbb{P})} e^{2 \delta m \Vert x_0 \Vert} \max_{l=1,...,m} \mathbb{E}\left[ e^{\varepsilon \Vert X_{t_l} \Vert} \right]^\frac{1}{p} < \infty.
		\end{aligned}
	\end{equation}
	Now, we show the identity \eqref{EqThmChaosProof5} by induction on $\vert \mathbf{j} \vert = 0,...,\vert \mathbf{k} \vert$, i.e.~for every $\mathbf{j} := (j_1,...,j_m) \in \mathbb{N}^m_0$ with $j_l \leq k_l$ for all $l = 1,...,m$ we show that for every $\lambda \in (-\delta,\delta)^m$ it holds that
	\begin{equation}
		\label{EqThmChaosProof7}
		\partial_{\mathbf{j}} u(\lambda) = \mathbb{E}\left[ \partial_{\mathbf{j}} U(\lambda,\cdot) \right].
	\end{equation}
	For the induction initialization $\mathbf{j} = (0,...,0)$ with $\vert \mathbf{j} \vert = 0$, there is nothing to prove. For the induction step, we assume without loss of generality that \eqref{EqThmChaosProof7} holds true for some $\mathbf{j} := (j_1,...,j_m) \in \mathbb{N}^m_0$ with $j_{l_0} < k_{l_0}$ for some $l_0 = 1,...,m$, $j_l = 0$ for all $l = 1,...,l_0-1$, and $j_l = k_l$ for all $l = l_0+1,...,m$. Then, for every fixed $(\lambda_1,...,\lambda_{l_0-1},\lambda_{l_0+1},...,\lambda_m) \in \mathbb{R}^{m-1}$ and $\omega \in \Omega$, the functions $(-\delta,\delta) \ni \lambda_{l_0} \mapsto \frac{\partial}{\partial \lambda_{l_0}} \partial_{\mathbf{j}} U\big((\lambda_1,...,\lambda_{l_0-1},\lambda_{l_0},\lambda_{l_0+1},...,\lambda_m),\omega\big)$ as well as $(-\delta,\delta) \ni \lambda_{l_0} \mapsto \mathbb{E}\big[ \frac{\partial}{\partial \lambda_{l_0}} \partial_{\mathbf{j}} U\big((\lambda_1,...,\lambda_{l_0-1},\lambda_{l_0},\lambda_{l_0+1},...,\lambda_m),\cdot\big) \big]$ exist and are continuous. Moreover, by using \eqref{EqThmChaosProof6}, we observe for every $\lambda_{l_0} \in (-\delta,\delta)$ that
	\begin{equation*}
		\mathbb{E}\left[ \left\vert \partial_{\mathbf{j}} U\big((\lambda_1,...,\lambda_{l_0-1},\lambda_{l_0},\lambda_{l_0+1},...,\lambda_m),\cdot\big) \right\vert \right] \leq \mathbb{E}\left[ \sup_{\widetilde{\lambda} \in (-\delta,\delta)^m} \left\vert \partial_{\mathbf{j}} U(\widetilde{\lambda},\cdot) \right\vert \right] < \infty.
	\end{equation*}
	as well as
	\begin{equation*}
		\mathbb{E}\left[ \sup_{\lambda_{l_0} \in (-\delta,\delta)} \left\vert \frac{\partial}{\partial \lambda_{l_0}} \partial_{\mathbf{j}} U\big((\lambda_1,...,\lambda_{l_0-1},\lambda_{l_0},\lambda_{l_0+1},...,\lambda_m),\cdot\big) \right\vert \right] \leq \mathbb{E}\left[ \sup_{\widetilde{\lambda} \in (-\delta,\delta)^m} \left\vert \partial_{\mathbf{j}+e_{l_0}} U(\widetilde{\lambda},\cdot) \right\vert \right] < \infty,
	\end{equation*}
	where $e_{l_0} := (0,...,0,1,0,...,0) \in \mathbb{N}^m_0$ with $1$ at position $l_0$. Hence, we can apply \cite[Theorem~A.5.2]{durrett10} to conclude that \eqref{EqThmChaosProof7} holds true with $\mathbf{j} + e_{l_0}$ instead of $\mathbf{j}$, proving the induction step and therefore \eqref{EqThmChaosProof4}.
	
	Finally, we show that $Z = 0 \in L^q(\Omega,\mathcal{F}_T,\mathbb{P})$. To this end, we fix some $m \in \mathbb{N}$, $(i_1,...,i_m) \in \lbrace 1,...,d \rbrace^m$, $(K_1,...,K_m) \in \mathbb{N}_0^m$, and $(t_1,...,t_m) \in [0,T]^m$. Then, by using $m$-times the binomial theorem (with $\binom{\mathbf{K}}{\mathbf{k}} := \frac{K_1!}{k_1! (K_1-k_1)!} \cdots \frac{K_m!}{k_m! (K_m-k_m)!}$), that $X_0$ is $\mathbb{P}$-a.s.~constant equal to $x_0 \in \mathbb{R}^d$,	and the identity~\eqref{EqThmChaosProof5}, it follows that
	\begin{equation*}
		\begin{aligned}
			& \mathbb{E}\left[ Z \left( X^{i_1}_{t_1} \right)^{K_1} \cdots \left( X^{i_m}_{t_m} \right)^{K_m} \right] = \mathbb{E}\left[ Z \left( X^{i_1}_{t_1} - X^{i_1}_0 + x^{i_1}_0 \right)^{K_1} \cdots \left( X^{i_m}_{t_m} - X^{i_m}_0 + x^{i_m}_0 \right)^{K_m} \right] \\
			& = \sum_{k_1=0}^{K_1} \cdots \sum_{k_m=0}^{K_m} \binom{\mathbf{K}}{\mathbf{k}} \mathbb{E}\left[ Z \left( X^{i_1}_{t_1} - X^{i_1}_0 \right)^{k_1} \cdots \left( X^{i_m}_{t_m} - X^{i_m}_0 \right)^{k_m} \right] \left( x^{i_1}_0 \right)^{K_1-k_1} \cdots \left( x^{i_m}_0 \right)^{K_m-k_m} = 0.
		\end{aligned}
	\end{equation*}
	However, since $\Pol(X)$ defined in \eqref{EqPropPolDense1} is by Proposition~\ref{PropPolDense} dense in $L^p(\Omega,\mathcal{F}_T,\mathbb{P})$, we conclude from the Hahn-Banach theorem that $Z = 0 \in L^q(\Omega,\mathcal{F}_T,\mathbb{P})$.
\end{proof}

\subsubsection{Proof of results in Section~\ref{SecLpHedging}}
\label{SecProofsLpHedging}

For the results in this section, we denote by $\mathbb{G} := (\mathcal{G}_t)_{t \in [0,T]}$ the $\mathbb{P}$-completion of the filtration\footnote{The $\mathbb{P}$-completion of the filtration generated by $X$ is defined as the smallest filtration $\mathbb{G} := (\mathcal{G}_t)_{t \in [0,T]}$ such that $\mathbb{G}$ is complete with respect to $(\Omega,\mathcal{A},\mathbb{P})$ and $X$ is $\mathbb{G}$-adapted.} generated by $X$. In particular, it holds for every $t \in [0,T)$ that $\bigcap_{s \in (t,T]} \mathcal{G}_s = \mathcal{F}_t$, whereas $\mathcal{G}_T = \mathcal{F}_T$. Moreover, we recall that the $\mathbb{F}$-predictable $\sigma$-algebra on $[0,T] \times \Omega$ coincides with the $\mathbb{G}$-predictable $\sigma$-algebra on $[0,T] \times \Omega$ (see \cite[Theorem~IV.97]{dellacherie78}). Furthermore, if for example the semimartingale $X$ is a strong Markov process, then $\mathbb{G} = \mathbb{F}$ (see \cite[Theorem~2.7.7]{karatzas98}).

\begin{lemma}
	\label{LemmaElem}
	Let $X$ be a continuous semimartingale satisfying Assumption~\ref{AssExpInt} and let $p \in [1,\infty)$. Then, the set of bounded elementary $\mathbb{F}$-predictable $\mathbb{R}^d$-valued processes
	\begin{equation}
		\label{EqLemmaElem1}
		\Theta^\infty_{e,d}(X) := \left\lbrace t \mapsto \sum_{k=1}^K \mathds{1}_{(t_{k-1},t_k]}(t) Z_k: 
		\begin{matrix}
			K \in \mathbb{N}, \, 0 = t_0 < t_1 < ... < t_K = T, \\
			Z_k \in L^\infty(\Omega,\mathcal{G}_{t_{k-1}},\mathbb{P};\mathbb{R}^d), \, k = 1,...,K
		\end{matrix}
		\right\rbrace
	\end{equation}
	is dense in $\Theta^p(X)$.
\end{lemma}
\begin{proof}
	First, we observe that $\Theta^\infty_{e,d}(X) \subseteq \Theta^p(X)$ holds by definition. In order to see that $\Theta^\infty_{e,d}(X)$ is dense in $\Theta^p(X)$, we first show that any given $\theta := \big( \theta^1_t, ..., \theta^d_t \big)_{t \in [0,T]}^\top \in \Theta^p(X)$ can be approximated by the sequence of bounded $\mathbb{F}$-predictable processes $\theta^{(n)} := \big( \mathds{1}_{\lbrace \vert \theta^1_t \vert \leq n \rbrace} \theta^1_t, ..., \mathds{1}_{\lbrace \vert \theta^d_t \vert \leq n \rbrace} \theta^d_t \big)_{t \in [0,T]}^\top$, $n \in \mathbb{N}$, with respect to $\Vert \cdot \Vert_{\Theta^p(X)}$. Indeed, by applying the dominated convergence theorem twice, it follows that\footnote{For every $i = 1,...,d$ the vector space $\Theta^p(X^i)$ consists of $\mathbb{F}$-predictable real-valued processes $\theta := (\theta_t)_{t \in [0,T]}$ such that $\Vert \theta \Vert_{\Theta^p(X^i)} := \mathbb{E}\big[ \big( \int_0^T \big\vert \theta^i_t dA^i_t \big\vert \big)^p \big]^{1/p} + \mathbb{E}\big[ \big( \int_0^T \big( \theta^i_t \big)^2 d\langle M^i \rangle_t \big)^{p/2} \big]^{1/p} < \infty$.}
	\begin{equation*}
		\begin{aligned}
			& \lim_{n \rightarrow \infty} \big\Vert \theta - \theta^{(n)} \big\Vert_{\Theta^p(X)} = \lim_{n \rightarrow \infty} \sum_{i=1}^d \big\Vert \theta^i - \theta^{(i,n)} \big\Vert_{\Theta^p(X^i)} \\
			& = \sum_{i=1}^d \lim_{n \rightarrow \infty} \mathbb{E}\left[ \left( \int_0^T \mathds{1}_{\lbrace \vert \theta^i_t \vert > n \rbrace} \left\vert \theta^i_t dA^i_t \right\vert \right)^p \right]^\frac{1}{p} + \sum_{i=1}^d \lim_{n \rightarrow \infty} \mathbb{E}\left[ \left( \int_0^T \mathds{1}_{\lbrace \vert \theta^i_t \vert > n \rbrace} \left( \theta^i_t \right)^2 d\langle M^i \rangle_t \right)^\frac{p}{2} \right]^\frac{1}{p} \\
			& = \sum_{i=1}^d \mathbb{E}\left[ \left( \lim_{n \rightarrow \infty} \int_0^T \mathds{1}_{\lbrace \vert \theta^i_t \vert > n \rbrace} \left\vert \theta^i_t dA^i_t \right\vert \right)^p \right]^\frac{1}{p} + \sum_{i=1}^d \mathbb{E}\left[ \left( \lim_{n \rightarrow \infty} \int_0^T \mathds{1}_{\lbrace \vert \theta^i_t \vert > n \rbrace} \left( \theta^i_t \right)^2 d\langle M^i \rangle_t \right)^\frac{p}{2} \right]^\frac{1}{p} = 0,
		\end{aligned}		
	\end{equation*}
	where $e_i$ denotes the $i$-th unit vector of $\mathbb{R}^d$. Hence, it suffices to show that every bounded $\mathbb{F}$-predictable process can be approximated by elementary $\mathbb{F}$-predictable processes from $\Theta^\infty_{e,d}(X)$.
	
	Now, for every fixed $i = 1,...,d$, we prove that the vector space of bounded $\mathbb{F}$-predictable real-valued processes is contained in $\mathcal{H}_i$ defined as the closure of $\Theta^\infty_{e,1}(X^i)$ (given as in \eqref{EqLemmaElem1} but with $d=1$) with respect to $\Vert \cdot \Vert_{\Theta^p(X^i)}$. To this end, we first observe that $\mathcal{E}_T := \lbrace \lbrace 0 \rbrace \times E: E \in \mathcal{G}_0 \rbrace \cup \lbrace (s,t] \times E: 0 \leq s < t \leq T, \, E \in \mathcal{G}_s \rbrace$ is a $\pi$-system with $[0,T] \times \Omega \in \mathcal{E}_T$, which generates the $\mathbb{F}$-predictable $\sigma$-algebra $\mathcal{P}_T := \sigma(\mathcal{E}_T)$ on $[0,T] \times \Omega$. Moreover, $\mathcal{H}_i$ is a vector space and for every $E \in \mathcal{E}_T$ it holds that $\mathds{1}_E \in \mathcal{H}_i$. In addition, for every sequence $(\vartheta^{(i,n)})_{n \in \mathbb{N}} \subseteq \mathcal{H}_i$ with $0 \leq \vartheta^{(i,n)}_t(\omega) \leq \vartheta^{(i,n+1)}_t(\omega)$ for all $(t,\omega) \in [0,T] \times \Omega$ and $n \in \mathbb{N}$ converging pointwise to some bounded $\mathbb{F}$-predictable real-valued process $\vartheta^i := \big( \vartheta^i_t \big)_{t \in [0,T]}$, we use the dominated convergence theorem twice to obtain that
	\begin{equation*}
		\begin{aligned}
			& \lim_{n \rightarrow \infty} \big\Vert \vartheta^i_t - \vartheta^{(i,n)}_t \big\Vert_{\Theta^p(X^i)} \\
			& \quad\quad = \lim_{n \rightarrow \infty} \mathbb{E}\left[ \left( \int_0^T \left\vert \left( \vartheta^i_t - \vartheta^{(i,n)}_t \right) dA^i_t \right\vert \right)^p \right]^\frac{1}{p} + \lim_{n \rightarrow \infty} \mathbb{E}\left[ \left( \int_0^T \left( \vartheta^i_t - \vartheta^{(i,n)}_t \right)^2 d\langle M^i \rangle_t \right)^\frac{p}{2} \right]^\frac{1}{p} \\
			& \quad\quad = \mathbb{E}\left[ \left( \lim_{n \rightarrow \infty} \int_0^T \left\vert \left( \vartheta^i_t - \vartheta^{(i,n)}_t \right) dA^i_t \right\vert \right)^p \right]^\frac{1}{p} + \mathbb{E}\left[ \left( \lim_{n \rightarrow \infty} \int_0^T \left( \vartheta^i_t - \vartheta^{(i,n)}_t \right)^2 d\langle M^i \rangle_t \right)^\frac{p}{2} \right]^\frac{1}{p} = 0,
		\end{aligned}
	\end{equation*}
	showing that $\vartheta^i \in \mathcal{H}_i$. Hence, we can apply the monotone class theorem (see \cite[Theorem~6.1.3]{durrett10}) to conclude that $\mathcal{H}_i$ contains all bounded $\mathbb{F}$-predictable real-valued processes.

	Finally, for every bounded $\mathbb{F}$-predictable process $\theta := \big( \theta^1_t, ..., \theta^d_t \big)_{t \in [0,T]}^\top$ and every $i = 1,...,d$, there exists by the previous step a sequence $\big( \theta^{(i,n)} \big)_{n \in \mathbb{N}} \subseteq \Theta^\infty_{e,1}(X^i)$ with $\lim_{n \rightarrow \infty} \big\Vert \theta^i - \theta^{(i,n)} \big\Vert_{\Theta^p(X^i)} = 0$. Hence, by defining $\theta^{(n)} := \big( \theta^{(1,n)}_t,...,\theta^{(d,n)}_t \big)_{t \in [0,T]}^\top \in \Theta^\infty_{e,d}(X)$, $n \in \mathbb{N}$, it follows that
	\begin{equation*}
		\lim_{n \rightarrow \infty} \big\Vert \theta - \theta^{(n)} \big\Vert_{\Theta^p(X)} = \sum_{i=1}^d \lim_{n \rightarrow \infty} \big\Vert \theta^i - \theta^{(n,i)} \big\Vert_{\Theta^p(X^i)} = 0,
	\end{equation*}
	which completes the proof.
\end{proof}

\begin{corollary}
	\label{CorLpDense}
	Let $X$ be a continuous semimartingale satisfying Assumption~\ref{AssExpInt}, let $p \in [1,\infty)$, and $t \in [0,T]$. Then, the following holds true:
	\begin{enumerate}
		\item\label{CorLpDense1} The set of polynomials
		\begin{equation*}
			\quad\quad\quad \Pol(X\vert_{[0,t]}) := \linspan\left\lbrace \left( X^{i_1}_{t_1} \right)^{k_1} \cdots \left( X^{i_m}_{t_m} \right)^{k_m}: \,\,
			\begin{matrix}
				m \in \mathbb{N}, \, (i_1,...,i_m) \in \lbrace 1,...,d \rbrace^m, \\
				(k_1,...,k_m) \in \mathbb{N}_0^m, \, (t_1,...,t_m) \in [0,t]^m
			\end{matrix}
			\right\rbrace
		\end{equation*}
		is dense in $L^p(\Omega,\mathcal{G}_t,\mathbb{P})$.
		\item\label{CorLpDense2} The direct sum $\bigoplus_{n \in \mathbb{N}_0} \big\lbrace J^\circ_n(g_n)_t: g_n \in L^{np}_{\diag}(X\vert_{[0,t]})^{\otimes n} \big\rbrace$ is dense in $L^p(\Omega,\mathcal{G}_t,\mathbb{P})$.
	\end{enumerate}
\end{corollary}
\begin{proof}
	For \ref{CorLpDense1}, we can follow the proof of Proposition~\ref{PropPolDense} (but with $t$ replacing $T$) and use that $\mathcal{G}_t$ is defined as the $\mathbb{P}$-completion of $\sigma(\lbrace X_s: s \in [0,t] \rbrace)$. For \ref{CorLpDense2}, we can follow the proof of Theorem~\ref{ThmChaos} (but with $t$ replacing $T$), where we now apply \ref{CorLpDense1} instead of Theorem~\ref{PropPolDense}.
\end{proof}

\begin{proposition}
	\label{PropGp}
	Let $X$ be a continuous semimartingale satisfying Assumption~\ref{AssExpInt} and let $p \in [1,\infty)$. Then, $\mathcal{G}^p(X) := \big\lbrace \int_0^T \theta_t^\top dX_t: \theta \in \Theta^p(X) \big\rbrace$ is contained in the closure of
	\begin{equation*}
		\mathcal{J}^p(X) := \left\lbrace \int_0^T \left( \vartheta^{g_{1:N}}_t \right)^\top dX_t: N, m_n \in \mathbb{N} \text{ and } g_{n,j,0},g_{n,j,1} \in L^{np}(X), \, n = 1,...,N, \, j = 1,...,m_n \right\rbrace
	\end{equation*}
	with respect to $\Vert \cdot \Vert_{L^p(\mathbb{P})}$, where $\vartheta^{g_{1:N}}_t := \sum_{n=1}^N \sum_{j=1}^{m_n} \frac{W(g_{n,j,0})_t^{n-1}}{(n-1)!} g_{n,j,1}(t)$, for $t \in [0,T]$.
\end{proposition}
\begin{proof}
	Fix some $\theta \in \Theta^p(X)$ and $\varepsilon > 0$. Then, by using Lemma~\ref{LemmaElem}, there exists an elementary $\mathbb{F}$-predictable process $\widetilde{\theta} \in \Theta^\infty_{e,d}(X)$ of the form $t \mapsto \sum_{k=1}^K \mathds{1}_{(t_{k-1},t_k]}(t) Z_k$, for some $K \in \mathbb{N}$, $0 = t_0 < t_1 < ... < t_K = T$, and $Z_k := (Z_{k,1},...,Z_{k,d})^\top \in L^\infty(\Omega,\mathcal{G}_{t_{k-1}},\mathbb{P};\mathbb{R}^d)$, $k = 1,...,K$, such that
	\begin{equation}
		\label{EqPropGpProof1}
		\big\Vert \theta - \widetilde{\theta} \big\Vert_{\Theta^p(X)} < \frac{\varepsilon}{2 \max(1,C_p)},
	\end{equation}
	where $C_p > 0$ is the constant in the upper Burkholder-Davis-Gundy inequality. Moreover, for every $k = 1,...,K$ and $i = 1,...,d$, we apply Corollary~\ref{CorLpDense}~\ref{CorLpDense2} on $Z_{k,i} \in L^\infty(\Omega,\mathcal{G}_{t_{k-1}},\mathbb{P}) \subseteq L^{2p}(\Omega,\mathcal{G}_{t_{k-1}},\mathbb{P})$ to obtain some $N, m_n \in \mathbb{N}$ and $h_{k,i,n,j} \in L^{2np}_{\diag}(X)$, $n = 0,...,N-1$ and $j = 1,...,m_n$, such that
	\begin{equation}
		\label{EqPropGpProof2}
		\mathbb{E}\left[ \left\vert Z_{k,i} - \sum_{n=0}^{N-1} \sum_{j=1}^{m_n} J^\circ_n\left( h_{k,i,n,j}^{\otimes n} \right)_{t_{k-1}} \right\vert^{2p} \right]^\frac{1}{2p} < \frac{\varepsilon}{2 \max(1,C_p) C_{A,M,p} d K},
	\end{equation}
	where $C_{A,M,p} := 1+\mathbb{E}\big[ \big( \int_0^T \Vert dA_t \Vert \big)^{2p} \big]^{1/(2p)} + \mathbb{E}\big[ \Vert \langle M \rangle_T \Vert_F^p \big]^{1/(2p)} > 0$. From this, we define for every $(k,i,j) \in \lbrace 1,...,K \rbrace \times \lbrace 1,...,d \rbrace \times \lbrace 1,...,m_n \rbrace$ the functions $g_{1,(k,i,j),0}, g_{1,(k,i,j),1} \in L^{np}(X)$ by
	\begin{equation*}
		\begin{aligned}
			\quad\quad\,\,\,\, [0,T] \ni t \quad \mapsto \quad g_{1,(k,i,j),0}(t) & := 0 \in \mathbb{R}^d \\
			\quad\quad\,\,\,\, [0,T] \ni t \quad \mapsto \quad g_{1,(k,i,j),1}(t) & := h_{k,i,0,j} \mathds{1}_{(t_{k-1},t_k]} e_i \in \mathbb{R}^d,
		\end{aligned}
	\end{equation*}
	where $e_i$ denotes the $i$-th unit vector of $\mathbb{R}^d$. Moreover, for every $n = 2,...,N$ and $(k,i,j) \in \lbrace 1,...,K \rbrace \times \lbrace 1,...,d \rbrace \times \lbrace 1,...,m_n \rbrace$, we define the functions $g_{n,(k,i,j),0}, g_{n,(k,i,j),1} \in L^{np}(X)$ by
	\begin{equation*}
		\begin{aligned}
			[0,T] \ni t \quad \mapsto \quad g_{n,(k,i,j),0}(t) & := h_{k,i,n-1,j} \in \mathbb{R}^d \\
			[0,T] \ni t \quad \mapsto \quad g_{n,(k,i,j),1}(t) & := \mathds{1}_{(t_{k-1},t_k]} e_i \in \mathbb{R}^d
		\end{aligned}
	\end{equation*}
	In addition, for every $k = 1,...,K$ and $i = 1,...,d$, we define $\widetilde{Z}_{k,i} := \sum_{n=0}^{N-1} \sum_{j=1}^{m_n} J_n\big(  h_{k,i,n,j}^{\otimes n} \big)_{t_{k-1}}$. Furthermore, we define $\vartheta^{g_{1:N}} := \big( \vartheta^{g_{1:N},1}_t, ..., \vartheta^{g_{1:N},1}_t \big)_{t \in [0,T]}^\top$ for every $i = 1,...,d$ and $t \in [0,T]$ by
	\begin{equation}
		\label{EqPropGpProof3}
		\begin{aligned}
			\vartheta^{g_{1:N},i}_t & := \sum_{n=1}^N \sum_{j=1}^{m_n} \frac{W\left( g_{n,(k,i,j),0} \right)_{t_{k-1}}^{n-1}}{(n-1)!} g_{n,(k,i,j),1}(t) \\
			& = \sum_{j=1}^{m_n} h_{k,i,0,j} \mathds{1}_{(t_{k-1},t_k]}(t) e_i + \sum_{n=2}^N \sum_{j=1}^{m_n} \frac{W\left( h_{k,i,n-1,j} \right)_{t_{k-1}}^{n-1}}{(n-1)!} \mathds{1}_{(t_{k-1},t_k]}(t) e_i \\
			& = \sum_{j=1}^{m_n} h_{k,i,0,j} \mathds{1}_{(t_{k-1},t_k]}(t) e_i + \sum_{n=2}^N \sum_{j=1}^{m_n} J^\circ_{n-1}\left(  h_{k,i,n-1,j}^{\otimes (n-1)} \right)_{t_{k-1}} \mathds{1}_{(t_{k-1},t_k]}(t) e_i \\
			& = \left( \sum_{n=0}^{N-1} \sum_{j=1}^{m_n} J^\circ_n\left(  h_{k,i,n,j}^{\otimes n} \right)_{t_{k-1}} \right) \mathds{1}_{(t_{k-1},t_k]}(t) e_i = \mathds{1}_{(t_{k-1},t_k]}(t) \widetilde{Z}_{k,i} e_i,
		\end{aligned}
	\end{equation}
	where the third equality follows from Proposition~\ref{PropMon}. Hence, by using \eqref{EqPropGpProof3}, that $\big\vert x^\top y \big\vert \leq \Vert x \Vert \Vert y \Vert$ and $x^\top \widetilde{M} x \leq \Vert x \Vert^2 \Vert \widetilde{M} \Vert_F$, H\"older's inequality, Minkowski's inequality together with $\big( \sum_{k=1}^K c_{k,i} \big)^{1/2} \leq \sum_{k=1}^K c_{k,i}^{1/2}$ for all $c_{k,i} \geq 0$, Proposition~\ref{PropMon}, and the inequality~\eqref{EqPropGpProof2}, it follows that
	\begin{equation}
		\label{EqPropGpProof4}
		\begin{aligned}
			& \left\Vert \widetilde{\theta} - \vartheta^{g_{1:N}} \right\Vert_{\Theta^p(X)} \\
			& = \sum_{i=1}^d \mathbb{E}\left[ \left( \int_0^T \left\vert \left( \widetilde{\theta}^i_t - \vartheta^{g_{1:N},i}_t \right) dA^i_t \right\vert \right)^p \right]^\frac{1}{p} + \sum_{i=1}^d \mathbb{E}\left[ \left( \int_0^T \left( \widetilde{\theta}^i_t - \vartheta^{g_{1:N},i}_t \right)^2 d\langle M^i \rangle_t \right)^\frac{p}{2} \right]^\frac{1}{p} \\
			& = \sum_{i=1}^d \mathbb{E}\left[ \left( \int_0^T \left\vert \left( \sum_{k=1}^K \mathds{1}_{(t_{k-1},t_k]}(t) \left( Z_{k,i} - \widetilde{Z}_{k,i} \right) \right) dA^i_t \right\vert \right)^p \right]^\frac{1}{p} \\
			& \quad\quad + \sum_{i=1}^d \mathbb{E}\left[ \left( \int_0^T \left( \sum_{k=1}^K \mathds{1}_{(t_{k-1},t_k]}(t) \left( Z_{k,i} - \widetilde{Z}_{k,i} \right) \right)^2 d\langle M^i \rangle_t \right)^\frac{p}{2} \right]^\frac{1}{p} \\
			& \leq \sum_{i=1}^d \mathbb{E}\left[ \left( \sum_{k=1}^K \left\vert Z_{k,i} - \widetilde{Z}_{k,i} \right\vert \int_{t_{k-1}}^{t_k} \left\vert dA^i_t \right\vert \right)^p \right]^\frac{1}{p} + \sum_{i=1}^d \mathbb{E}\left[ \left( \sum_{k=1}^K \left\vert Z_{k,i} - \widetilde{Z}_{k,i} \right\vert^2 \left( \langle M^i \rangle_{t_k} - \langle M^i \rangle_{t_{k-1}} \right) \right)^\frac{p}{2} \right]^\frac{1}{p} \\
			& \leq \sum_{i=1}^d \sum_{k=1}^K \mathbb{E}\left[ \vert Z_{k,i} - \widetilde{Z}_{k,i} \vert^{2p} \right] \mathbb{E}\left[ \left( \int_0^T \Vert dA_t \Vert \right)^{2p} \right]^\frac{1}{2p} + \sum_{i=1}^d \sum_{k=1}^K \mathbb{E}\left[ \vert Z_{k,i} - \widetilde{Z}_{k,i} \vert^{2p} \right] \mathbb{E}\left[ \Vert \langle M \rangle_T \Vert_F^p \right]^\frac{1}{2p} \\
			& \leq C_{A,M,p} \sum_{i=1}^d \sum_{k=1}^K \mathbb{E}\left[ \left\vert Z_{k,i} - \sum_{n=0}^{N-1} \sum_{j=1}^{m_n} J^\circ_n\left( h_{k,i,n,j}^{\otimes n} \right)_{t_{k-1}} \right\vert^{2p} \right]^\frac{1}{2p} \\
			& < C_{A,M,p} dK \frac{\varepsilon}{2 \max(1,C_p) C_{A,M,p} dK} = \frac{\varepsilon}{2 \max(1,C_p)}.
		\end{aligned}
	\end{equation}
	Thus, by using Minkowski's inequality, the Burkholder-Davis-Gundy inequality with constant $C_p > 0$, and the inequalities \eqref{EqPropGpProof1}+\eqref{EqPropGpProof4}, it follows that 
	\begin{equation*}
		\begin{aligned}
			& \mathbb{E}\left[ \left\vert \int_0^T \theta_t^\top dX_t - \int_0^T \left( \vartheta^{g_{1:N}}_t \right)^\top dX_t \right\vert^p \right]^\frac{1}{p} \leq \sum_{i=1}^d \mathbb{E}\left[ \left\vert \int_0^T \left( \theta^i_t - \vartheta^{g_{1:N},i}_t \right) dX^i_t \right\vert^p \right]^\frac{1}{p} \\
			& \quad\quad \leq \sum_{i=1}^d \mathbb{E}\left[ \left\vert \int_0^T \left( \theta^i_t - \vartheta^{g_{1:N},i}_t \right) dA^i_t \right\vert^p \right]^\frac{1}{p} + \sum_{i=1}^d \mathbb{E}\left[ \left\vert \int_0^T \left( \theta^i_t - \vartheta^{g_{1:N},i}_t \right) dM^i_t \right\vert^p \right]^\frac{1}{p} \\
			& \quad\quad \leq \sum_{i=1}^d \mathbb{E}\left[ \left\vert \int_0^T \left( \theta^i_t - \vartheta^{g_{1:N},i}_t \right) dA^i_t \right\vert^p \right]^\frac{1}{p} + C_p \sum_{i=1}^d \mathbb{E}\left[ \left( \int_0^T \left( \theta^i_t - \vartheta^{g_{1:N},i}_t \right)^2 d\langle M^i \rangle_t \right)^\frac{p}{2} \right]^\frac{1}{p} \\
			& \quad\quad \leq \max(1,C_p) \sum_{i=1}^d \big\Vert \theta^i - \vartheta^{g_{1:N},i} \big\Vert_{\Theta^p(X^i)} = \max(1,C_p) \big\Vert \theta - \vartheta^{g_{1:N}} \big\Vert_{\Theta^p(X)} \\
			& \quad\quad \leq \max(1,C_p) \big\Vert \theta - \widetilde{\theta} \big\Vert_{\Theta^p(X)} + \max(1,C_p) \big\Vert \widetilde{\theta} - \vartheta^{g_{1:N}} \big\Vert_{\Theta^p(X)} \\
			& \quad\quad < \max(1,C_p) \frac{\varepsilon}{2 \max(1,C_p)} + \max(1,C_p) \frac{\varepsilon}{2 \max(1,C_p)} = \varepsilon.
		\end{aligned}
	\end{equation*}
	Since $\theta \in \Theta^p(X)$ and $\varepsilon > 0$ were chosen arbitrarily, this shows that $\mathcal{G}^p(X)$ is contained in the closure of $\mathcal{J}^p(X)$ with respect to $\Vert \cdot \Vert_{L^p(\mathbb{P})}$.
\end{proof}

\begin{proof}[Proof of Theorem~\ref{ThmLpHedging}]	
	Fix some $G \in L^p(\Omega,\mathcal{F}_T,\mathbb{P})$ and $\varepsilon > 0$. Then, by definition of the infimum, there exists some $(\widetilde{c},\widetilde{\theta}) \in \mathbb{R} \times \Theta^p(X)$ such that
	\begin{equation}
		\label{EqThmLpHedgingProof1}
		\left\Vert G - \widetilde{c} - \int_0^T \widetilde{\theta}_t^\top dX_t \right\Vert_{L^p(\mathbb{P})} - \frac{\varepsilon}{2} < + \inf_{(c,\theta) \in \mathbb{R} \times \Theta^p(X)} \left\Vert G - c - \int_0^T \theta_t^\top dX_t \right\Vert_{L^p(\mathbb{P})}.
	\end{equation}
	Moreover, by applying Proposition~\ref{PropGp} to $\int_0^T \widetilde{\theta}_t^\top dX_t \in \mathcal{G}^p(X)$, there exists some $N,m_n \in \mathbb{N}$ and $g_{n,j,0},g_{n,j,1} \in L^{np}(X)$, $n = 1,...,N$ and $j = 1,...,m_n$ such that $\int_0^T \left( \vartheta^{g_{1:N}}_t \right)^\top dX_t$ with $\vartheta^{g_{1:N}}_t := \sum_{n=1}^N \sum_{j=1}^{m_n} \frac{W(g_{n,j,0})_t^{n-1}}{(n-1)!} g_{n,j,1}(t)$, $t \in [0,T]$, satisfies
	\begin{equation}
		\label{EqThmLpHedgingProof2}
		\left\Vert \int_0^T \widetilde{\theta}_t^\top dX_t - \int_0^T \left( \vartheta^{g_{1:N}}_t \right)^\top dX_t \right\Vert_{L^p(\mathbb{P})} < \frac{\varepsilon}{2}.
	\end{equation}
	Moreover, we set $g_0 := \widetilde{c} \in \mathbb{R}$. Hence, by using the inequality~\eqref{EqThmLpHedgingProof1}, the reverse triangle inequality of $\Vert \cdot \Vert_{L^p(\mathbb{P})}$, and the inequality~\eqref{EqThmLpHedgingProof2}, it follows that
	\begin{equation*}
		\begin{aligned}
			& \left\Vert G - g_0 - \int_0^T \left( \vartheta^{g_{1:N}}_t \right)^\top dX_t \right\Vert_{L^p(\mathbb{P})} - \inf_{(c,\theta) \in \mathbb{R} \times \Theta^p(X)} \left\Vert G - c - \int_0^T \theta_t^\top dX_t \right\Vert_{L^p(\mathbb{P})} \\
			& \quad\quad < \left\Vert G - g_0 - \int_0^T \left( \vartheta^{g_{1:N}}_t \right)^\top dX_t \right\Vert_{L^p(\mathbb{P})} - \left\Vert G - \widetilde{c} - \int_0^T \widetilde{\theta}_t^\top dX_t \right\Vert_{L^p(\mathbb{P})} + \frac{\varepsilon}{2} \\
			& \quad\quad \leq \left\Vert \int_0^T \widetilde{\theta}_t^\top dX_t - \int_0^T \left( \vartheta^{g_{1:N}}_t \right)^\top dX_t \right\Vert_{L^p(\mathbb{P})} + \frac{\varepsilon}{2} < \frac{\varepsilon}{2} + \frac{\varepsilon}{2} = \varepsilon,
		\end{aligned}
	\end{equation*}
	which completes the proof.
\end{proof}

\subsection{Proof of results in Section~\ref{SecUATs}}
\label{SecProofsUATs}

\subsubsection{Proof of results in Section~\ref{SecUAT}}
\label{SecProofsUAT}

\begin{lemma}
	\label{LemmaContDenseLp}
	Let $X$ be a continuous semimartingale satisfying Assumption~\ref{AssExpInt} and let $p \in [1,\infty)$. Then, $C([0,T];\mathbb{R}^d)$ is dense in $L^p(X)$.
\end{lemma}
\begin{proof}
	First, we show that the set of step functions is dense in $L^p(X)$, where a step function is of the form  $[0,T] \ni t \mapsto \mathds{1}_{\lbrace 0 \rbrace}(t) v_0 + \sum_{l=1}^m \mathds{1}_{(t_{l-1},t_l]}(t) v_l \in \mathbb{R}^d$, for some $m \in \mathbb{N}$, $0 = t_0 < t_1 < ... < t_m = T$, and $v_0,...,v_m \in \mathbb{R}^d$. To this end, we fix some $g \in L^p(X)$ and $\varepsilon > 0$. Then, by using that $g \in D_l([0,T];\mathbb{R}^d)$, i.e.~that $[0,T] \ni t \mapsto g(T-t) \in \mathbb{R}^d$ is c\`adl\`ag, we can apply \cite[Lemma~12.1]{billingsley99} to conclude that there exists some $0 = t_0 < t_1 < ... < t_m = T$ and $v_1,...,v_m \in \mathbb{R}^d$ such that for every $l = 1,...,m$ it holds that $\sup_{u,v \in (t_{l-1},t_l]} \Vert g(u) - g(v) \Vert < C_{p,T,X}^{-1} \varepsilon$, where the constant $C_{p,T,X} > 0$ was introduced in Lemma~\ref{LemmaLpNorm}~\ref{LemmaLpNorm2}. Hence, by using Lemma~\ref{LemmaLpNorm}~\ref{LemmaLpNorm2}, it follows for the step function $[0,T] \ni t \mapsto s(t) := \mathds{1}_{\lbrace 0 \rbrace}(t) g(0) + \sum_{l=1}^m \mathds{1}_{(t_{l-1},t_l]}(t) g(t_l) \in \mathbb{R}^d$ that
	\begin{equation*}
		\begin{aligned}
			\Vert g - s \Vert_{L^p(X)} & \leq C_{p,T,X} \Vert g - s \Vert_\infty \leq C_{p,T,X} \left( \Vert g(0) - s(0) \Vert + \max_{l=1,...,m} \sup_{u \in (t_{l-1},t_l]} \Vert g(u) - s(u) \Vert \right) \\
			& = C_{p,T,X} \max_{l=1,...,m} \sup_{u \in (t_{l-1},t_l]} \Vert g(u) - g(t_l) \Vert < C_{p,T,X} \frac{\varepsilon}{C_{p,T,X}} = \varepsilon.
		\end{aligned}
	\end{equation*}
	Since $\varepsilon > 0$ and $g \in L^p(X)$ were chosen arbitrarily, this shows that step functions are dense in $L^p(X)$.
	
	Next, we show that every fixed step function $[0,T] \ni t \mapsto s(t) := \mathds{1}_{\lbrace 0 \rbrace}(t) v_0 + \sum_{l=1}^m \mathds{1}_{(t_{l-1},t_l]}(t) v_l \in \mathbb{R}^d$ can be approximated by some element in $C([0,T];\mathbb{R}^d)$ with respect to $\Vert \cdot \Vert_{L^p(X)}$, where $m \in \mathbb{N}$, $0 = t_0 < t_1 < ... < t_m = T$, and $v_0,...,v_m \in \mathbb{R}^d$. To this end, we define for every fixed $\delta \in (0,t_{\min})$ with $t_{\min} := \min_{l=1,...,m} (t_l-t_{l-1})$ the function
	\begin{equation*}
		[0,T] \ni t \quad \mapsto \quad f_\delta(t) :=
		\begin{cases}
			v_{l-1} (1-\delta) + v_l \delta, & t \in (t_{l-1},t_{l-1}+\delta], \\
			v_l, & t \in (t_{l-1}+\delta,t_l].
		\end{cases}
	\end{equation*}
	Then, $f_\delta: [0,T] \rightarrow \mathbb{R}^d$ is continuous with $\Vert f_\delta \Vert_\infty = \Vert s \Vert_\infty$ and converges pointwise to the fixed step function $s: [0,T] \rightarrow \mathbb{R}^d$, as $\delta \rightarrow 0$. Moreover, by using that $X$ is a continuous semimartingale, there exists some $E \in \mathcal{A}$ with $\mathbb{P}[E]=1$ such that for every $\omega \in E$ the paths $t \mapsto X_t(\omega)$, $t \mapsto A_t(\omega)$, and $t \mapsto \langle M \rangle_t(\omega)$ are continuous. Now, for every fixed $\omega \in E$, we use the dominated convergence theorem (with $f_\delta \rightarrow s$ pointwise, as $\delta \rightarrow 0$, and $\Vert f_\delta \Vert_\infty = \Vert s \Vert_\infty$ for all $\delta > 0$) to conclude that
	\begin{equation}
		\label{EqLemmaContDenseLpProof1}
		\lim_{\delta \rightarrow 0} \int_0^T \Vert s(t) - f_\delta(t) \Vert \Vert dA_t(\omega) \Vert = 0 \quad \text{and} \quad \lim_{\delta \rightarrow 0} \int_0^T (s(t)-f_\delta(t))^\top d\langle M \rangle_t(\omega) (s(t)-f_\delta(t)) = 0.
	\end{equation}
	In addition, by using \eqref{EqLemmaLpIneqProof2}, we observe for every $\delta \in (0,1)$ that, $\mathbb{P}$-a.s.,
	\begin{equation}
		\label{EqLemmaContDenseLpProof2}
		\begin{aligned}
			\left( \int_0^T \Vert s(t) - f_\delta(t) \Vert \Vert dA_t \Vert \right)^p & \leq (2 \Vert s \Vert_\infty)^p \left( \int_0^T \Vert dA_t \Vert \right)^p, \quad \text{and} \\
			\left( \int_0^T (s(t)-f_\delta(t))^\top d\langle M \rangle_t (s(t)-f_\delta(t)) \right)^\frac{p}{2} & \leq (3d)^\frac{p}{2} (2\Vert s \Vert_\infty)^p \Vert \langle M \rangle_T \Vert_F^\frac{p}{2},
		\end{aligned}
	\end{equation}
	where both right-hand sides do not depend on $\delta \in (0,1)$ and are $\mathbb{P}$-integrable due to \eqref{EqLemmaLpIneqProof1} (using that $X$ satisfies Assumption~\ref{AssExpInt}). Hence, by another application of the dominated convergence theorem (with $\big( \int_0^T \Vert s(t) - f_\delta \Vert \Vert dA_t \Vert \big)^p \rightarrow 0$ and $\big( \int_0^T (s(t)-f_\delta(t))^\top d\langle M \rangle_t (s(t)-f_\delta(t)) \big)^{p/2} \rightarrow 0$, as $\delta \rightarrow 0$, $\mathbb{P}$-a.s. (see \eqref{EqLemmaContDenseLpProof1}), and \eqref{EqLemmaContDenseLpProof2}), we have
	\begin{equation*}
		\begin{aligned}
			\lim_{\delta \rightarrow 0} \Vert s - f_\delta \Vert_{L^p(X)} & \leq \lim_{\delta \rightarrow 0} \mathbb{E}\left[ \left( \int_0^T \Vert s(t)-f_\delta(t) \Vert \Vert dA_t \Vert \right)^p \right]^\frac{1}{p} \\
			& \quad\quad + \lim_{\delta \rightarrow 0} \mathbb{E}\left[ \left( \int_0^T (s(t)-f_\delta(t))^\top d\langle M \rangle_t (s(t)-f_\delta(t)) \right)^\frac{p}{2} \right]^\frac{1}{p} = 0,
		\end{aligned}
	\end{equation*}
	which completes the proof.
\end{proof}

\begin{lemma}
	\label{LemmaUAT}
	Let $X$ be a continuous semimartingale satisfying Assumption~\ref{AssExpInt}, let $p \in [1,\infty)$, and $\rho \in C(\mathbb{R})$ be non-polynomial. Then, $\mathcal{NN}^\rho_{d,1}$ is dense in $L^p(X)$.
\end{lemma}
\begin{proof}
	For $p \in [1,\infty)$, let $g \in L^p(X)$ and fix some $\varepsilon > 0$. Then, by using Lemma~\ref{LemmaContDenseLp}, there exists some $f := (f_1,...,f_d)^\top \in C([0,T];\mathbb{R}^d)$ such that
	\begin{equation}
		\label{EqLemmaUATProof1}
		\Vert g - f \Vert_{L^p(X)} < \frac{\varepsilon}{2}.
	\end{equation}
	Moreover, for every fixed $i = 1,...,d$, we apply the universal approximation theorem in\cite[Theorem~1]{leshno93} to conclude that there exists a neural network $\mathbb{R} \ni t \mapsto \varphi_i(t) := \sum_{n=1}^N y_{n,i} \rho(a_{n,i} t + b_{n,i}) \in \mathbb{R}$, for some $a_{n,i},,b_{n,i}, y_{n,i} \in \mathbb{R}$, $n = 1,...,N$, such that
	\begin{equation}
		\label{EqLemmaUATProof2}
		\Vert f_i - \varphi_i \Vert_\infty := \sup_{t \in [0,T]} \vert g_i(t) - \varphi_i(t) \vert < \frac{\varepsilon}{2 C_{p,T,X} \sqrt{d}},
	\end{equation}
	where the constant $C_{p,T,X} > 0$ was introduced in Lemma~\ref{LemmaLpNorm}~\ref{LemmaLpNorm2}. From this, we define the neural network $[0,T] \ni t \mapsto \varphi := (\varphi_1,...,\varphi_d)^\top = \sum_{n=1}^N y_n \rho(a_n t + b_n) \in \mathbb{R}^d$ satisfying $\varphi \in \mathcal{NN}^\rho_{d,1}$, where $\rho \in C(\mathbb{R})$ is applied componentwise, and where $a_n := (a_{n,1},...,a_{n,d})^\top \in \mathbb{R}^d$, $b_n := (b_{n,1},...,b_{n,d})^\top \in \mathbb{R}^d$, and $y_n := (y_{n,1},...,y_{n,d})^\top \in \mathbb{R}^d$, $n = 1,...,N$. Then, by using Lemma~\ref{LemmaLpNorm}~\ref{LemmaLpNorm2}, Minkowski's integral inequality, and the inequalities \eqref{EqLemmaUATProof1}+\eqref{EqLemmaUATProof2}, it follows that
	\begin{equation*}
		\begin{aligned}
			\Vert g - \varphi \Vert_{L^p(X)} & \leq \Vert g - f \Vert_{L^p(X)} + \Vert f - \varphi \Vert_{L^p(X)} \\
			& \leq \Vert g - f \Vert_{L^p(X)} + C_{p,T,X} \Vert f - \varphi \Vert_\infty \\
			& \leq \Vert g - f \Vert_{L^p(X)} + C_{p,T,X} \sup_{t \in [0,T]} \left( \sum_{i=1}^d \vert f_i(t) - \varphi_i(t) \vert^2 \right)^\frac{1}{2} \\
			& \leq \Vert g - f \Vert_{L^p(X)} + C_{p,T,X} \sqrt{d} \max_{i=1,...,d} \sup_{t \in [0,T]} \vert f_i(t) - \varphi_i(t) \vert \\
			& < \frac{\varepsilon}{2} + C_{p,T,X} \sqrt{d} \frac{\varepsilon}{2 C_{p,T,X} \sqrt{d}} = \varepsilon.
		\end{aligned}
	\end{equation*}
	Since $\varepsilon > 0$ and $g \in L^p(X)$ were chosen arbitrarily, this shows that $\mathcal{NN}^\rho_{d,1}$ is dense in $L^p(X)$.
\end{proof}

\begin{proof}[Proof of Proposition~\ref{PropUAT}]
	For $n = 0$ and $p \in [1,\infty)$, we first observe that $\mathcal{NN}^\rho_{d,0} := \mathbb{R}$ is dense in $L^{np}(X)^{\otimes 0} := \mathbb{R}$. Moreover, for $n = 1$, the conclusion follows from Lemma~\ref{LemmaUAT} by using the identification $L^{np}(X)^{\otimes 1} \cong L^p(X)$. Now, for $n \geq 2$, we fix some $\varepsilon \in (0,1)$ and $g \in L^{np}_{\diag}(X)^{\otimes n}$, which has a representation $g = \sum_{j=1}^m g_j^{\otimes n}$, for some $g_1,...,g_m \in L^{np}(X)$, such that
	\begin{equation}
		\label{EqPropUATProof1}
		\left\Vert g - \sum_{j=1}^m g_j^{\otimes n} \right\Vert_{L^{np}(X)^{\otimes n}} < \frac{\varepsilon}{2}.
	\end{equation}
	Then, for every fixed $j = 1,...,m$, we apply Lemma~\ref{LemmaUAT} to obtain some $\varphi_j \in \mathcal{NN}^\rho_{d,1}$ such that
	\begin{equation*}
		\Vert g_j - \varphi_j \Vert_{L^{np}(X)} < \frac{\varepsilon}{2 m n \left( 1 + \Vert g_j \Vert_{L^{np}(X)} \right)^{n-1}} \leq 1,
	\end{equation*}
	which implies that $\Vert \varphi_j \Vert_{L^{np}(X)} \leq 1 + \Vert g_j \Vert_{L^{np}(X)}$. Using this, the telescoping sum $g_j^{\otimes n} - \varphi_j^{\otimes n} = \sum_{l=1}^n g_j^{\otimes (n-l)} \otimes (g_j - \varphi_j) \otimes \varphi_j^{\otimes (l-1)}$, the triangle inequality, and Lemma~\ref{LemmaLpXnNorm}, it follows that
	\begin{equation}
		\label{EqPropUATProof2}
		\begin{aligned}
			\left\Vert g_j^{\otimes n} - \varphi_j^{\otimes n} \right\Vert_{L^{np}(X)^{\otimes n}} & \leq \sum_{l=1}^n \left\Vert g_j^{\otimes (n-l)} \otimes \left( g_j - \varphi_j \right) \otimes g_j^{\otimes (l-1)} \right\Vert_{L^{np}(X)^{\otimes n}} \\
			& \leq \sum_{l=1}^n \Vert g_j \Vert_{L^{np}(X)}^{n-l} \Vert g_j - \varphi_j \Vert_{L^{np}(X)} \Vert \varphi_j \Vert_{L^{np}(X)}^{l-1} \\
			& \leq n \left( 1 + \Vert g_j \Vert_{L^{np}(X)} \right)^{n-1} \Vert g_j - \varphi_j \Vert_{L^{np}(X)} < \frac{\varepsilon}{2m}.
		\end{aligned}
	\end{equation}
	Hence, by using \eqref{EqPropUATProof1} and \eqref{EqPropUATProof2}, we conclude for $\varphi := \sum_{j=1}^m \varphi_j^{\otimes n} \in \mathcal{NN}^\rho_{d,n}$ that
	\begin{equation*}
		\Vert g - \varphi \Vert_{L^{np}(X)^{\otimes n}} \leq \left\Vert g - \sum_{j=1}^m g_j^{\otimes n} \right\Vert_{L^{np}(X)^{\otimes n}} + \sum_{j=1}^m \left\Vert g_j^{\otimes n} - \varphi_j^{\otimes n} \right\Vert_{L^{np}(X)^{\otimes n}} < \frac{\varepsilon}{2} + \sum_{j=1}^m \frac{\varepsilon}{2m} = \varepsilon.
	\end{equation*}
	Since $\varepsilon > 0$ and $g \in L^{np}_{\diag}(X)^{\otimes n}$ were chosen arbitrarily, $\mathcal{NN}^\rho_{d,n}$ is dense in $L^{np}_{\diag}(X)^{\otimes n}$.
\end{proof}

\begin{proof}[Proof of Theorem~\ref{ThmUAT}]
	Fix some $G \in L^p(\Omega,\mathcal{F}_T,\mathbb{P})$ and $\varepsilon > 0$. Then, by using Theorem~\ref{ThmChaos}, there exists some $N \in \mathbb{N}$ and $g_n \in L^{np}_{\diag}(X)^{\otimes n}$, $n = 0,...,N$, such that
	\begin{equation}
		\label{EqThmUATProof1}
		\left\Vert G - \sum_{n=0}^N J^\circ_n(g_n)_T \right\Vert_{L^p(\mathbb{P})} < \frac{\varepsilon}{2}.
	\end{equation}
	Now, for every $n = 1,...,N$, we apply Proposition~\ref{PropUAT} to obtain some $\varphi_n \in \mathcal{NN}^\rho_{d,n}$ such that
	\begin{equation}
		\label{EqThmUATProof2}
		\left\Vert g_n - \varphi_n \right\Vert_{L^{np}(X)^{\otimes n}} < \frac{\varepsilon}{2 C_{n,p} N},
	\end{equation}
	where the constant $C_{n,p} > 0$ was introduced in Lemma~\ref{LemmaItIntLinearBDG}~\ref{LemmaItIntLinearBDG3}. Hence, by using $\varphi_0 := g_0 \in \mathbb{R}$, Minkowksi's inequality, Lemma~\ref{LemmaItIntLinearBDG}~\ref{LemmaItIntLinearBDG3}, and the inequalities \eqref{EqThmUATProof1}+\eqref{EqThmUATProof2}, it follows that
	\begin{equation*}
		\begin{aligned}
			\left\Vert G - \sum_{n=0}^N J^\circ_n(\varphi_n)_T \right\Vert_{L^p(\mathbb{P})} & \leq \left\Vert G - \sum_{n=0}^N J^\circ_n(g_n)_T \right\Vert_{L^p(\mathbb{P})} + \left\Vert \sum_{n=0}^N J^\circ_n(g_n-\varphi_n)_T \right\Vert_{L^p(\mathbb{P})} \\
			& \leq \left\Vert G - \sum_{n=0}^N J^\circ_n(g_n)_T \right\Vert_{L^p(\mathbb{P})} + \sum_{n=1}^N \left\Vert J^\circ_n(g_n-\varphi_n)_T \right\Vert_{L^p(\mathbb{P})} \\
			& \leq \left\Vert G - \sum_{n=0}^N J^\circ_n(g_n)_T \right\Vert_{L^p(\mathbb{P})} + \sum_{n=1}^N C_{n,p} \Vert g_n - \varphi_n \Vert_{L^{np}(X)^{\otimes n}} \\
			& < \frac{\varepsilon}{2} + \sum_{n=1}^N C_{n,p} \frac{\varepsilon}{2 C_{n,p} N} = \varepsilon,
		\end{aligned}
	\end{equation*}
	which completes the proof.
\end{proof}

\subsubsection{Proof of results in Section~\ref{SecRandUAT}}
\label{SecProofsRandUAT}

\begin{lemma}
	\label{LemmaRandNNWellDef}
	Let $X$ be a continuous semimartingale satisfying Assumption~\ref{AssExpInt}, let $p,r \in [1,\infty)$, $n \in \mathbb{N}$, let $\rho \in C(\mathbb{R})$, and let Assumption~\ref{AssCondCDF} hold. Then, $\mathcal{RN}^\rho_{d,n} \subseteq L^r(\widetilde{\Omega},\widetilde{\mathcal{A}},\widetilde{\mathbb{P}};\overline{L^{np}(X)^{\otimes n}})$.
\end{lemma}
\begin{proof}
	First, we observe that $(\overline{L^{np}(X)^{\otimes n}},\Vert \cdot \Vert_{L^{np}(X)^{\otimes n}})$ is as completion of $L^{np}(X)^{\otimes n}$ a Banach space, which we claim to be separable. To this end, we use the polynomial functions $\Pol_\mathbb{Q}([0,T];\mathbb{R}^d) := \big\lbrace [0,T] \ni t \mapsto \sum_{n=0}^N q_n t^n \in \mathbb{R}^d: N \in \mathbb{N}_0, \, q_1,...,q_N \in \mathbb{Q}^d \big\rbrace \subseteq D_l([0,T];\mathbb{R}^d)$ to define the countable set
	\begin{equation*}
		\Pol_\mathbb{Q}([0,T];\mathbb{R}^d)^{\otimes n} := \left\lbrace p_1 \otimes \cdots \otimes p_n \in L^{np}(X)^{\otimes n}: p_1,...,p_n \in \Pol_\mathbb{Q}([0,T];\mathbb{R}^d) \right\rbrace.
	\end{equation*}
	Then, by applying the Stone-Weierstrass theorem componentwise, it follows that $\Pol_\mathbb{Q}([0,T];\mathbb{R}^d)$ is dense in $C([0,T];\mathbb{R}^d)$ with respect to $\Vert \cdot \Vert_\infty$. Hence, by combining Lemma~\ref{LemmaLpNorm}~\ref{LemmaLpNorm2} with Lemma~\ref{LemmaContDenseLp} (as in the proof of Lemma~\ref{LemmaUAT}), we conclude that $\Pol_\mathbb{Q}([0,T];\mathbb{R}^d)$ is dense in $L^p(X)$. Thus, by similar arguments as in the proof of Proposition~\ref{PropUAT}, it follows that $\Pol_\mathbb{Q}([0,T];\mathbb{R}^d)^{\otimes n}$ is dense in $L^{np}(X)^{\otimes n}$.
	
	Next, we show that every $\widetilde{\varphi} \in \mathcal{RN}^\rho_{d,n}$ is $\widetilde{\mathbb{P}}$-strongly measurable. By linearity of $\mathcal{RN}^\rho_{d,n}$, it suffices to prove for every fixed $j \in \mathbb{N}$ as well as for every fixed $\widetilde{\mathcal{F}}_{\rand}/\mathcal{B}(\mathbb{R})$-measurable and $\widetilde{\mathbb{P}}$-a.s.~bounded random variable $\widetilde{y}_n: \widetilde{\Omega} \rightarrow \mathbb{R}$ that
	\begin{equation}
		\label{EqLemmaRandNNWellDefProof1}
		\widetilde{\Omega} \ni \widetilde{\omega} \quad \mapsto \quad \widetilde{R}_j(\widetilde{\omega}) := \widetilde{y}_j(\omega) \rho\left( \widetilde{a}_j(\omega) \cdot + \widetilde{b}_j(\omega) \right)^{\otimes n} \in L^{np}(X)^{\otimes n} \subseteq \overline{L^{np}(X)^{\otimes n}}
	\end{equation} 
	is $\widetilde{\mathbb{P}}$-strongly measurable, where $\widetilde{y}_j(\omega) \rho\big( \widetilde{a}_j(\omega) \cdot + \widetilde{b}_j(\omega) \big)$ denotes $[0,T] \ni t \mapsto \widetilde{y}_j(\omega) \rho\big( \widetilde{a}_j(\omega) t + \widetilde{b}_j(\omega) \big) \in \mathbb{R}^d$ with $\rho \in C(\mathbb{R})$ applied componentwise. To this end, we first show for every fixed $y \in \mathbb{R}$ that the map
	\begin{equation}
		\label{EqLemmaRandNNWellDefProof2}
		\mathbb{R}^d \times \mathbb{R}^d \ni (a,b) \quad \mapsto \quad y \rho(a \cdot + b) \in L^{np}(X)
	\end{equation}
	is continuous. Indeed, fix some $\varepsilon > 0$ and let $(a_k,b_k)_{k \in \mathbb{N}} \subseteq \mathbb{R}^d \times \mathbb{R}^d$ be an arbitrary sequence converging to any fixed $(a,b) \in \mathbb{R}^d \times \mathbb{R}^d$, where we use the notation $a_k := (a_{k,1},...,a_{k,d})^\top \in \mathbb{R}^d$ and $a := (a_1,...,a_d)^\top \in \mathbb{R}^d$, and analogously for the vectors $b_k,b \in \mathbb{R}^d$. Then, we can define the finite constants $C_a := \Vert a \Vert + \sup_{k \in \mathbb{N}} \Vert a_k \Vert \geq 0 $ and $C_b := \Vert b \Vert + \sup_{k \in \mathbb{N}} \Vert b_k \Vert \geq 0$ as well as the compact set $K := \left\lbrace yt + z:  t \in [0,T], \, y \in [-C_a,C_a], \, z \in [-C_b,C_b] \right\rbrace$. Hence, by using that $\rho \in C(\mathbb{R})$ is continuous, thus uniformly continuous on the compact subset $K \subseteq \mathbb{R}$, there exists some $\delta > 0$ such that for every $s_1,s_2 \in K$ with $\vert s_1 - s_2 \vert < \delta$ it holds that
	\begin{equation}
		\label{EqLemmaRandNNWellDefProof3}
		\vert \rho(s_1) - \rho(s_2) \vert < \frac{\varepsilon}{(1+\vert y \vert) C_{np,T,X} \sqrt{d}},
	\end{equation}
	where $C_{np,T,X} > 0$ is the constant introduced in Lemma~\ref{LemmaLpNorm}~\ref{LemmaLpNorm2}. Now, by using that $(a_k,b_k)_{k \in \mathbb{N}} \subseteq \mathbb{R}^d \times \mathbb{R}^d$ converges to $(a,b) \in \mathbb{R}^d \times \mathbb{R}^d$, there exists some $k_0 \in \mathbb{N}$ such that for every $k \in \mathbb{N} \cap [k_0,\infty)$ it holds that $\Vert (a,b) - (a_k,b_k) \Vert < \delta/(1+T)$. This implies for every $i = 1,...,d$ and $t \in [0,T]$ that
	\begin{equation}
		\label{EqLemmaRandNNWellDefProof4}
		\big\vert (a_i t + b_i) - (a_{k,i} t + b_{k,i}) \big\vert \leq \vert a_i - a_{k,i} \vert T + \vert b_i - b_{k,i} \vert \leq (1+T) \Vert (a,b) - (a_k,b_k) \Vert < \delta.
	\end{equation}
	Thus, by using Lemma~\ref{LemmaLpNorm}~\ref{LemmaLpNorm2} and by combining \eqref{EqLemmaRandNNWellDefProof3} with \eqref{EqLemmaRandNNWellDefProof4}, it follows for every $k \in \mathbb{N} \cap [k_0,\infty)$ that
	\begin{equation*}
		\begin{aligned}
			& \Vert y \rho(a \cdot + b) - y \rho(a_k \cdot + b_k) \Vert_{L^{np}(X)} \leq C_{np,T,X} \vert y \vert \Vert \rho(a \cdot + b) - \rho(a_k \cdot + b_k) \Vert_\infty \\
			& \quad\quad \leq \vert y \vert C_{np,T,X} \sup_{t \in [0,T]} \left( \sum_{i=1}^d \vert \rho(a_i t + b_i) - \rho(a_{k,i} t + b_{k,i}) \vert^2 \right)^\frac{1}{2} \\
			& \quad\quad \leq \vert y \vert C_{np,T,X} \sqrt{d} \max_{i=1,...,d} \sup_{t \in [0,T]} \vert \rho(a_i t + b_i) - \rho(a_{k,i} t + b_{k,i}) \vert \\
			& \quad\quad < \vert y \vert C_{np,T,X} \sqrt{d} \frac{\varepsilon}{(1+\vert y \vert) C_{np,T,X} \sqrt{d}} = \varepsilon.
		\end{aligned}
	\end{equation*}
	Since $\varepsilon > 0$ was chosen arbitrarily, this shows that \eqref{EqLemmaRandNNWellDefProof2} is continuous. Now, we use this to conclude that
	\begin{equation}
		\label{EqLemmaRandNNWellDefProof5}
		\mathbb{R}^d \times \mathbb{R}^d \times \mathbb{R} \ni (a,b,y) \quad \mapsto \quad y \rho(a \cdot + b)^{\otimes n} \in L^{np}(X)^{\otimes n}
	\end{equation}
	is also continuous. To this end, let $(a_k,b_k,y_k)_{k \in \mathbb{N}} \subseteq \mathbb{R}^d \times \mathbb{R}^d \times \mathbb{R}$ be an arbitrary sequence converging to any fixed $(a,b,y) \in \mathbb{R}^d \times \mathbb{R}^d \times \mathbb{R}$. Then, by using the triangle inequality, the telescoping sum $\rho(a \cdot + b)^{\otimes n} - \rho(a_k \cdot + b_k)^{\otimes n} = \sum_{l=1}^n \rho(a \cdot + b)^{\otimes (n-l)} \otimes (\rho(a \cdot + b) - \rho(a_k \cdot + b_k)) \otimes \rho(a_k \cdot + b_k)^{\otimes (l-1)}$, Lemma~\ref{LemmaLpXnNorm}, and that \eqref{EqLemmaRandNNWellDefProof2} is continuous, it follows that
	\begin{equation*}
		\begin{aligned}
			& \lim_{k \rightarrow \infty} \left\Vert y \rho(a \cdot + b)^{\otimes n} - y_k \rho(a_k \cdot + b_k)^{\otimes n} \right\Vert_{L^{np}(X)^{\otimes n}} \\
			& \quad\quad \leq \lim_{k \rightarrow \infty} \vert y - y_k \vert \left\Vert \rho(a \cdot + b)^{\otimes n} \right\Vert_{L^{np}(X)^{\otimes n}} + \lim_{k \rightarrow \infty} \vert y_k \vert \left\Vert \rho(a \cdot + b)^{\otimes n} - \rho(a_k \cdot + b_k)^{\otimes n} \right\Vert_{L^{np}(X)^{\otimes n}} \\
			& \quad\quad \leq \vert y \vert \lim_{k \rightarrow \infty} \sum_{l=1}^n \left\Vert \rho(a \cdot + b)^{\otimes (n-l)} \otimes \left( \rho(a \cdot + b) - \rho(a_k \cdot + b_k) \right) \otimes \rho(a_k \cdot + b_k)^{\otimes (l-1)} \right\Vert_{L^{np}(X)^{\otimes n}} \\
			& \quad\quad = \vert y \vert \lim_{k \rightarrow \infty} \sum_{l=1}^n \Vert \rho(a \cdot + b) \Vert_{L^{np}(X)}^{n-l} \left\Vert \rho(a \cdot + b) - \rho(a_k \cdot + b_k) \right\Vert_{L^{np}(X)} \Vert \rho(a_k \cdot + b_k) \Vert_{L^{np}(X)}^{l-1} \\
			& \quad\quad = n \vert y \vert \Vert \rho(a \cdot + b) \Vert_{L^{np}(X)}^{n-1} \lim_{k \rightarrow \infty} \left\Vert \rho(a \cdot + b) - \rho(a_k \cdot + b_k) \right\Vert_{L^{np}(X)} = 0,
		\end{aligned}
	\end{equation*}
	which shows that \eqref{EqLemmaRandNNWellDefProof5} is indeed continuous. Hence, the map \eqref{EqLemmaRandNNWellDefProof1} is measurable as composition of the $\widetilde{\mathcal{F}}/\mathcal{B}(\mathbb{R}^d \times \mathbb{R}^d \times \mathbb{R})$-measurable map $\widetilde{\Omega} \ni \widetilde{\omega} \mapsto (\widetilde{a}_j(\widetilde{\omega}),\widetilde{b}_j(\widetilde{\omega}),\widetilde{y}_j(\widetilde{\omega})) \in \mathbb{R}^d \times \mathbb{R}^d \times \mathbb{R}$ and the continuous map \eqref{EqLemmaRandNNWellDefProof5}. Since $(\overline{L^{np}(X)^{\otimes n}},\Vert \cdot \Vert_{L^{np}(X)^{\otimes n}})$ is by the first step separable, it follows from \cite[Corollary~1.1.10 \& Proposition 1.1.16.(2)]{hytoenen16} that the map \eqref{EqLemmaRandNNWellDefProof1} is $\widetilde{\mathbb{P}}$-strongly measurable.
	
	Finally, we show that $\mathcal{RN}^\rho_{d,n} \subseteq L^r(\widetilde{\Omega},\widetilde{\mathcal{A}},\widetilde{\mathbb{P}};\overline{L^{np}(X)^{\otimes n}})$. Again, by linearity of $\mathcal{RN}^\rho_{d,n}$, it suffices to prove that $\widetilde{R}_j \in L^r(\widetilde{\Omega},\widetilde{\mathcal{A}},\widetilde{\mathbb{P}};\overline{L^{np}(X)^{\otimes n}})$, where $\widetilde{R}_j: \widetilde{\Omega} \rightarrow \overline{L^{np}(X)^{\otimes n}}$ is defined in \eqref{EqLemmaRandNNWellDefProof1}. Indeed, by using that $\widetilde{y}_j: \widetilde{\Omega} \rightarrow \mathbb{R}$ is $\widetilde{\mathbb{P}}$-a.s.~bounded, Lemma~\ref{LemmaLpXnNorm} together with the fact that $(\widetilde{a}_j,\widetilde{b}_j) \sim (\widetilde{a}_1,\widetilde{b}_1)$ are identically distributed, Lemma~\ref{LemmaLpNorm}~\ref{LemmaLpNorm2}, and Assumption~\ref{AssCondCDF}, it follows that
	\begin{equation*}
		\begin{aligned}
			\widetilde{\mathbb{E}}\left[ \left\Vert \widetilde{R}_j \right\Vert_{L^{np}(X)^{\otimes n}}^r \right] & = \widetilde{\mathbb{E}}\left[ \left\Vert \widetilde{y}_j \rho\left( \widetilde{a}_j \cdot + \widetilde{b}_j \right)^{\otimes n} \right\Vert_{L^{np}(X)^{\otimes n}}^r \right] \\
			& \leq \left\Vert \widetilde{y}_j \right\Vert_{L^\infty(\widetilde{\mathbb{P}})}^r \widetilde{\mathbb{E}}\left[ \left\Vert \rho\left( \widetilde{a}_j \cdot + \widetilde{b}_j \right)^{\otimes n} \right\Vert_{L^{np}(X)^{\otimes n}}^r \right] \\
			& = \left\Vert \widetilde{y}_j \right\Vert_{L^\infty(\widetilde{\mathbb{P}})}^r \widetilde{\mathbb{E}}\left[ \left\Vert \rho\left( \widetilde{a}_1 \cdot + \widetilde{b}_1 \right) \right\Vert_{L^{np}(X)}^{nr} \right] \\
			& \leq C_{np,T,X} \left\Vert \widetilde{y}_j \right\Vert_{L^\infty(\widetilde{\mathbb{P}})}^r \widetilde{\mathbb{E}}\left[ \left\Vert \rho\left( \widetilde{a}_1 \cdot + \widetilde{b}_1 \right) \right\Vert_\infty^{nr} \right] < \infty.
		\end{aligned}
	\end{equation*}
	This shows that $\widetilde{R}_j \in L^r(\widetilde{\Omega},\widetilde{\mathcal{A}},\widetilde{\mathbb{P}};\overline{L^{np}(X)^{\otimes n}})$ and thus $\mathcal{RN}^\rho_{d,n} \subseteq L^r(\widetilde{\Omega},\widetilde{\mathcal{A}},\widetilde{\mathbb{P}};\overline{L^{np}(X)^{\otimes n}})$.
\end{proof}

\begin{lemma}
	\label{LemmaRandUAT}
	Let $X$ be a continuous semimartingale satisfying Assumption~\ref{AssExpInt}, let $p,r \in [1,\infty)$, $g \in L^p(X)$, let $\rho \in C(\mathbb{R})$ be non-polynomial, and let Assumption~\ref{AssCondCDF} hold. Then, for every $\varepsilon > 0$ there exists some $\widetilde{\varphi} \in \mathcal{RN}^\rho_{d,1}$ such that $\widetilde{\mathbb{E}}\big[ \Vert g - \widetilde{\varphi} \Vert_{L^p(X)}^r \big]^{1/r} < \varepsilon$.
\end{lemma}
\begin{proof}
	For $p,r \in [1,\infty)$, let $g \in L^p(X)$ and fix some $\varepsilon > 0$. Since $\rho \in C(\mathbb{R})$ is non-polynomial, there exists by Lemma~\ref{LemmaUAT} a fully trained neural network $\varphi := \sum_{n=1}^N y_n \rho(a_n \cdot + b_n) \in \mathcal{NN}^\rho_{d,1}$, with some $a_n := (a_{n,1},...,a_{n,d})^\top \in \mathbb{R}^d$, $b_n := (b_{n,1},...,b_{n,d})^\top \in \mathbb{R}^d$, and $y_n \in \mathbb{R}$, $n = 1,...,N$, such that
	\begin{equation}
		\label{EqLemmaRandUATProof1}
		\Vert g - \varphi \Vert_{L^{np}(X)} < \frac{\varepsilon}{3}.
	\end{equation}
	Hereby, $y_n \rho(a_n \cdot + b_n)$ denotes the function $[0,T] \ni t \mapsto y_n \rho(a_n t + b_n) \in \mathbb{R}^d$ with $\rho \in C(\mathbb{R})$ being applied componentwise. Now, for every fixed $n = 1,...,N$ and $m,j \in \mathbb{N}$, we define the map
	\begin{equation}
		\label{EqLemmaRandUATProof2}
		\widetilde{\Omega} \ni \widetilde{\omega} \quad \mapsto \quad \widetilde{R}_{n,m,j}(\widetilde{\omega}) := \widetilde{y}_{n,m,j}(\widetilde{\omega}) \rho\left( \widetilde{a}_j(\widetilde{\omega}) \cdot + \widetilde{b}_j(\widetilde{\omega}) \right) \in L^p(X) \subseteq \overline{L^p(X)},
	\end{equation}
	where $\rho\big( \widetilde{a}_j(\widetilde{\omega}) \cdot + \widetilde{b}_j(\widetilde{\omega}) \big)$ denotes the function $[0,T] \ni t \mapsto \rho\big( \widetilde{a}_j(\omega) t + \widetilde{b}_j(\omega) \big) \in \mathbb{R}^d$ and where
	\begin{equation*}
		\widetilde{\Omega} \ni \widetilde{\omega} \quad \mapsto \quad \widetilde{y}_{n,m,j}(\widetilde{\omega}) := \frac{y_n}{C_{n,m}} \mathds{1}_{E_{n,m}} (\widetilde{a}_j(\widetilde{\omega}),\widetilde{b}_j(\widetilde{\omega})) \in \mathbb{R}
	\end{equation*}
	with $E_{n,m} := \big\lbrace (x,y) \in \mathbb{R}^d \times \mathbb{R}^d: \Vert (x,y)-(a_n,b_n) \Vert < \frac{1}{m} \big\rbrace \in \mathcal{B}(\mathbb{R}^d \times \mathbb{R}^d)$ and $C_{n,m} := \widetilde{\mathbb{P}}\big[ \big\lbrace \widetilde{\omega} \in \widetilde{\Omega}: (\widetilde{a}_j(\widetilde{\omega}),\widetilde{b}_j(\widetilde{\omega})) \in E_{n,m} \big\rbrace \big] > 0$ due to Assumption~\ref{AssCondCDF}. Moreover, by using that the random variable $\widetilde{y}_{n,m,j}: \widetilde{\Omega} \rightarrow \mathbb{R}$ is $\widetilde{\mathcal{F}}_{\rand}/\mathcal{B}(\mathbb{R})$-measurable and $\widetilde{\mathbb{P}}$-a.s.~bounded, we have $\widetilde{R}_{n,m,j} \in \mathcal{RN}^\rho_{d,1}$ and thus $\widetilde{R}_{n,m,j} \in L^r(\widetilde{\Omega},\widetilde{\mathcal{A}},\widetilde{\mathbb{P}};\overline{L^p(X)})$ by Lemma~\ref{LemmaRandNNWellDef}.
	
	Now, we show that there exists some $m_0 \in \mathbb{N}$ such that for every fixed $n = 1,...,N$ the expectation $\widetilde{\mathbb{E}}\big[ \widetilde{R}_{n,m_0,1} \big] \in \overline{L^p(X)}$ is $\frac{\varepsilon}{3N}$-close to the function $y_n \rho(a_n \cdot + b_n) \in \overline{L^p(X)}$ with respect to $\Vert \cdot \Vert_{L^p(X)}$. To this end, we first observe that $E_{n,m+1} \subseteq E_{n,m} \subseteq E_{n,1}$ holds true, for all $m \in \mathbb{N}$, and that the set $K := \left\lbrace xt + z:  t \in [0,T], \, (x,z) \in E_{n,1} \right\rbrace$ is compact. Hence, by using that $\rho \in C(\mathbb{R})$ is continuous, thus uniformly continuous on the compact subset $K \subseteq \mathbb{R}$, there exists some $\delta > 0$ such that for every $s_1,s_2 \in K$ with $\vert s_1 - s_2 \vert < \delta$ it holds that
	\begin{equation}
		\label{EqLemmaRandUATProof3}
		\vert \rho(s_1) - \rho(s_2) \vert < \frac{\varepsilon}{3N (1+\vert y_n \vert) C_{p,T,X} \sqrt{d}},
	\end{equation}
	where $C_{p,T,X} > 0$ is the constant introduced in Lemma~\ref{LemmaLpNorm}~\ref{LemmaLpNorm2}. Moreover, by choosing $m_0 \in \mathbb{N}$ large enough such that $m_0 \geq (T+1)/\delta$, we conclude for every $i = 1,...,d$, $t \in [0,T]$, and $(x,z) \in E_{n,m_0}$ that
	\begin{equation}
		\label{EqLemmaRandUATProof4}
		\big\vert (a_{n,i} t + b_n) - (x_i t + z_i) \big\vert \leq \vert a_{n,i} - x_i \vert T + \vert b_{n,i} - z_i \vert \leq (T+1) \Vert (a_n,b_n) - (x,y) \Vert < \frac{T+1}{m_0} \leq \delta,
	\end{equation}
	where we use the notation $x := (x_1,...,x_d)^\top \in \mathbb{R}$ and $z := (z_1,...,z_d)^\top \in \mathbb{R}$. Hence, by using \cite[Proposition~1.2.2]{hytoenen16}, Lemma~\ref{LemmaLpNorm}~\ref{LemmaLpNorm2}, and the inequalities \eqref{EqLemmaRandUATProof3}+\eqref{EqLemmaRandUATProof4} together with $(y,z) \in E_{n,m_0}$, it follows that
	\begin{equation}
		\begin{aligned}
			& \left\Vert y_n \rho(a_n \cdot + b_n) - \widetilde{\mathbb{E}}\left[ \widetilde{R}_{n,m_0,1} \right] \right\Vert_{L^p(X)} \\
			& \quad\quad = \left\Vert \widetilde{\mathbb{E}}\left[ \frac{y_n}{C_{n,m_0}} \mathds{1}_{E_{n,m_0}} (\widetilde{a}_1,\widetilde{b}_1) \left( \rho(a_n \cdot + b_n) - \rho\left( \widetilde{a}_1 \cdot + \widetilde{b}_1 \right) \right) \right] \right\Vert_{L^p(X)} \\
			& \quad\quad \leq \frac{\vert y_n \vert}{C_{n,m_0}} \widetilde{\mathbb{E}}\left[ \mathds{1}_{E_{n,m_0}}(\widetilde{a}_1,\widetilde{b}_1) \left\Vert \rho(a_n \cdot + b_n) - \rho\left( \widetilde{a}_1 \cdot + \widetilde{b}_1 \right) \right\Vert_{L^p(X)} \right] \\
			& \quad\quad \leq \frac{\vert y_n \vert}{C_{n,m_0}} \widetilde{\mathbb{E}}\left[ \mathds{1}_{E_{n,m_0}}(\widetilde{a}_1,\widetilde{b}_1) \right] \sup_{(x,z) \in E_{n,m_0}} \left\Vert \rho(a_n \cdot + b_n) - \rho(x \cdot + z) \right\Vert_{L^p(X)} \\
			& \quad\quad \leq \frac{\vert y_n \vert}{C_{n,m_0}} C_{n,m_0} C_{p,T,X} \sup_{(x,z) \in E_{n,m_0}} \Vert \rho(a_n t + b_n) - \rho(x t + z) \Vert_\infty \\
			& \quad\quad = \vert y_n \vert C_{p,T,X} \sup_{(x,z) \in E_{n,m_0}} \sup_{t \in [0,T]} \left( \sum_{i=1}^d \vert \rho(a_{n,i} t + b_n) - \rho(x_i t + z_i) \vert^2 \right)^\frac{1}{2} \\
			& \quad\quad \leq \vert y_n \vert C_{p,T,X} \sqrt{d} \max_{i=1,...,d} \sup_{(t,(x,z)) \in [0,T] \times E_{n,m_0}} \vert \rho(a_{n,i} t + b_n) - \rho(x_i t + z_i) \vert \\
			& \quad\quad < \vert y_n \vert C_{p,T,X} \sqrt{d} \frac{\varepsilon}{3N (1+\vert y_n \vert) C_{p,T,X} \sqrt{d}} \leq \frac{\varepsilon}{3N},
		\end{aligned}
	\end{equation}
	which shows that $\widetilde{\mathbb{E}}\big[ \widetilde{R}_{n,m_0,1} \big] \in \overline{L^p(X)}$ is $\frac{\varepsilon}{3N}$-close to $y_n \rho(a_n \cdot + b_n) \in \overline{L^p(X)}$.
	
	Next, we approximate for every fixed $n = 1,...,N$ the constant random variable $\big( \widetilde{\omega} \mapsto \widetilde{\mathbb{E}}\big[\widetilde{R}_{n,m_0,1}\big] \big) \in L^1(\widetilde{\Omega},\widetilde{\mathcal{A}},\widetilde{\mathbb{P}};\overline{L^p(X)})$ by the average of the i.i.d.~sequence $\big( \widetilde{\omega} \mapsto \widetilde{R}_{n,m_0,j}(\widetilde{\omega}) \big) \in L^1(\widetilde{\Omega},\widetilde{\mathcal{A}},\widetilde{\mathbb{P}};\overline{L^p(X)})$ defined in \eqref{EqLemmaRandUATProof2}. More precisely, by applying the strong law of large numbers for Banach space-valued random variables in \cite[Theorem~3.3.10]{hytoenen16}, we conclude that
	\begin{equation}
		\label{EqLemmaRandUATProof5}
		\frac{1}{J} \sum_{j=1}^J \widetilde{R}_{n,m_0,j} \quad \overset{J \rightarrow \infty}{\longrightarrow} \quad \widetilde{\mathbb{E}}\left[ \widetilde{R}_{n,m_0,1} \right] \quad\quad \text{in} \quad L^1(\widetilde{\Omega},\widetilde{\mathcal{A}},\widetilde{\mathbb{P}};\overline{L^p(X)}) \quad \text{and} \quad \widetilde{\mathbb{P}}\text{-a.s.}.
	\end{equation}
	In order to generalize the convergence to $L^r(\widetilde{\Omega},\widetilde{\mathcal{A}},\widetilde{\mathbb{P}};\overline{L^p(X)})$, we define the sequence of non-negative random variables $(Z_{n,J})_{J \in \mathbb{N}}$ by $Z_{n,J} := \big\Vert \widetilde{\mathbb{E}}\big[ \widetilde{R}_{n,m_0,1} \big] - \frac{1}{J} \sum_{j=1}^J \widetilde{R}_{n,m_0,j} \big\Vert_{L^p(X)}^r$, $J \in \mathbb{N}$. Then, by using the triangle inequality of $\Vert \cdot \Vert_{L^p(X)}$, Minkowski's inequality, \cite[Proposition~1.2.2]{hytoenen16}, that $\widetilde{R}_{n,m_0,j} \sim \widetilde{R}_{n,m_0,1}$ is identically distributed, Jensen's inequality, and that $\widetilde{R}_{n,m_0,1} \in \mathcal{RN}^\rho_{d,1} \subseteq L^q(\widetilde{\Omega},\widetilde{\mathcal{A}},\widetilde{\mathbb{P}};\overline{L^p(X)})$ for any $q \in [1,\infty)$ (see Lemma~\ref{LemmaRandNNWellDef}), it follows for every $J \in \mathbb{N}$ and $q \in [1,\infty)$ that
	\begin{equation*}
		\begin{aligned}
			\widetilde{\mathbb{E}}\left[ Z_{n,J}^\frac{q}{r} \right]^\frac{1}{q} & = \widetilde{\mathbb{E}}\left[ \left\Vert \widetilde{\mathbb{E}}\left[ \widetilde{R}_{n,m_0,1} \right] - \frac{1}{J} \sum_{j=1}^J \widetilde{R}_{n,m_0,j} \right\Vert_{L^p(X)}^q \right]^\frac{1}{q} \\
			& \leq \left\Vert \widetilde{\mathbb{E}}\left[ \widetilde{R}_{n,m_0,1} \right] \right\Vert_{L^p(X)} + \widetilde{\mathbb{E}}\left[ \left( \frac{1}{J} \sum_{j=1}^J \left\Vert \widetilde{R}_{n,m_0,j} \right\Vert_{L^p(X)} \right)^q \right]^\frac{1}{q} \\
			& \leq \widetilde{\mathbb{E}}\left[ \left\Vert \widetilde{R}_{n,m_0,1} \right\Vert_{L^p(X)} \right] + \frac{1}{J} \sum_{j=1}^J \widetilde{\mathbb{E}}\left[ \left\Vert \widetilde{R}_{n,m_0,j} \right\Vert_{L^p(X)}^q \right]^\frac{1}{q} \\
			& \leq 2 \widetilde{\mathbb{E}}\left[ \left\Vert \widetilde{R}_{n,m_0,1} \right\Vert_{L^p(X)}^q \right]^\frac{1}{q} < \infty.
		\end{aligned}
	\end{equation*}
	Since the right-hand side is finite and does not depend on $J \in \mathbb{N}$, we conclude that $\sup_{J \in \mathbb{N}} \widetilde{\mathbb{E}}\big[ Z_{n,J}^{q/r} \big] < \infty$ for all $q \in (1,\infty)$. Hence, the family of random variables $(Z_{n,J})_{J \in \mathbb{N}}$ is by de la Vall\'ee-Poussin's theorem in \cite[Theorem~6.19]{klenke14} uniformly integrable. Thus, by using \eqref{EqLemmaRandUATProof5}, i.e.~that $Z_{n,J} \rightarrow 0$, $\widetilde{\mathbb{P}}$-a.s., as $J \rightarrow \infty$, together with Vitali's convergence theorem, it follows that
	\begin{equation*}
		\lim_{J \rightarrow \infty} \mathbb{E}\left[ \left\Vert \widetilde{\mathbb{E}}\left[ \widetilde{R}_{n,m_0,1} \right] - \frac{1}{J} \sum_{j=1}^J \widetilde{R}_{n,m_0,j} \right\Vert_{L^p(X)}^r \right] = \lim_{J \rightarrow \infty} \widetilde{\mathbb{E}}\left[ Z_{n,J} \right] = 0.
	\end{equation*}
	From this, we conclude that there exists some $J_n \in \mathbb{N}$ such that
	\begin{equation}
		\label{EqLemmaRandUATProof6}
		\widetilde{\mathbb{E}}\left[ \left\Vert \widetilde{\mathbb{E}}\left[ \widetilde{R}_{n,m_0,1} \right] - \frac{1}{J_n} \sum_{j=1}^{J_n} \widetilde{R}_{n,m_0,j} \right\Vert_{L^p(X)}^r \right]^\frac{1}{r} < \frac{\varepsilon}{3N},
	\end{equation}
	which shows that the average of the i.i.d.~sequence $(\widetilde{R}_{n,m_0,j})_{j = 1,...,J_n}$ is $\frac{\varepsilon}{3N}$-close to $\widetilde{\mathbb{E}}\big[\widetilde{R}_{n,m_0,1}\big]$.
	
	Finally, we define the random neural network $\widetilde{\varphi} := \sum_{n=1}^N \frac{1}{J_n} \sum_{j=1}^{J_n} \widetilde{R}_{n,m_0,j} \in \mathcal{RN}^\rho_{d,1}$. Then, by combining \eqref{EqLemmaRandUATProof1}, \eqref{EqLemmaRandUATProof4}, and \eqref{EqLemmaRandUATProof6} with the triangle inequality and Minkowski's inequality, it follows that
	\begin{equation*}
		\begin{aligned}
			\widetilde{\mathbb{E}}\left[ \left\Vert g - \widetilde{\varphi} \right\Vert_{L^p(X)}^r \right]^\frac{1}{r} & \leq \widetilde{\mathbb{E}}\left[ \left\Vert g - \varphi \right\Vert_{L^p(X)}^r \right]^\frac{1}{r} + \widetilde{\mathbb{E}}\left[ \left\Vert \varphi - \sum_{n=1}^N \widetilde{\mathbb{E}}\left[ \widetilde{R}_{n,m_0,1} \right] \right\Vert_{L^p(X)}^r \right]^\frac{1}{r} \\
			& \quad\quad + \widetilde{\mathbb{E}}\left[ \left\Vert \sum_{n=1}^N \widetilde{\mathbb{E}}\left[ \widetilde{R}_{n,m_0,1} \right] - \sum_{n=1}^N \frac{1}{J_n} \sum_{j=1}^{J_n} \widetilde{R}_{n,m_0,j} \right\Vert_{L^p(X)}^r \right]^\frac{1}{r} \\
			& \leq \frac{\varepsilon}{3} + \sum_{n=1}^N \left\Vert \rho(a_n \cdot + b_n) - \widetilde{\mathbb{E}}\left[ \widetilde{R}_{n,m_0,1} \right] \right\Vert_{L^p(X)} \\
			& \quad\quad + \sum_{n=1}^N \widetilde{\mathbb{E}}\left[ \left\Vert \widetilde{\mathbb{E}}\left[ \widetilde{R}_{n,m_0,1} \right] - \frac{1}{J_n} \sum_{j=1}^{J_n} \widetilde{R}_{n,m_0,j} \right\Vert_{L^p(X)}^r \right]^\frac{1}{r} \\
			& \leq \frac{\varepsilon}{3} + N \frac{\varepsilon}{3N} + N \frac{\varepsilon}{3N} = \varepsilon,
		\end{aligned}
	\end{equation*}
	which completes the proof.
\end{proof}

\begin{proof}[Proof of Proposition~\ref{PropRandUAT}]
	For $n = 0$ and $p \in [1,\infty)$, the conclusion holds trivially true as $\mathcal{RN}^\rho_{d,0} = \mathbb{R}$. Moreover, for $n = 1$, the conclusion follows from Lemma~\ref{LemmaRandUAT} by using the identification $L^{np}(X)^{\otimes 1} \cong L^p(X)$. Now, for $n \geq 2$, we fix some $\varepsilon \in (0,1)$ and $g \in L^{np}_{\diag}(X)^{\otimes n}$, which has a representation $g = \sum_{j=1}^m g_j^{\otimes n}$, for some $g_1,...,g_m \in L^{np}(X)$, such that
	\begin{equation}
		\label{EqPropRandUATProof1}
		\left\Vert g - \sum_{j=1}^m g_j^{\otimes n} \right\Vert_{L^{np}(X)^{\otimes n}} < \frac{\varepsilon}{2}.
	\end{equation}
	Then, for every fixed $j = 1,...,m$, we apply Lemma~\ref{LemmaUAT} to obtain some $\widetilde{\varphi}_j \in \mathcal{RN}^\rho_{d,1}$ such that
	\begin{equation*}
		\widetilde{\mathbb{E}}\left[ \left\Vert g_j - \widetilde{\varphi}_j(\cdot) \right\Vert_{L^{np}(X)}^{nr} \right]^\frac{1}{nr} < \frac{\varepsilon}{2 m n \left( 1 + \Vert g_j \Vert_{L^{np}(X)} \right)^{n-1}} \leq 1,
	\end{equation*}
	which implies that $\widetilde{\mathbb{E}}\big[ \Vert \widetilde{\varphi}_j(\cdot) \Vert_{L^{np}(X)}^{nr} \big]^{nr} \leq 1 + \Vert g_j \Vert_{L^{np}(X)}$. Using this, the telescoping sum $g_j^{\otimes n} - \widetilde{\varphi}_j(\widetilde{\omega})^{\otimes n} = \sum_{l=1}^n g_j^{\otimes (n-l)} \otimes (g_j - \widetilde{\varphi}_j(\widetilde{\omega})) \otimes \widetilde{\varphi}_j(\widetilde{\omega})^{\otimes (l-1)}$, Minkowski's inequality, Lemma~\ref{LemmaLpXnNorm}, H\"older's inequality, and Jensen's inequality, it follows that
	\begin{equation}
		\label{EqPropRandUATProof2}
		\begin{aligned}
			\widetilde{\mathbb{E}}\left[ \left\Vert g_j^{\otimes n} - \widetilde{\varphi}_j(\cdot)^{\otimes n} \right\Vert_{L^{np}(X)^{\otimes n}}^r \right]^\frac{1}{r} & \leq \sum_{l=1}^n \widetilde{\mathbb{E}}\left[ \left\Vert g_j^{\otimes (n-l)} \otimes \left( g_j - \widetilde{\varphi}_j(\cdot) \right) \otimes \widetilde{\varphi}_j(\cdot)^{\otimes (l-1)} \right\Vert_{L^{np}(X)^{\otimes n}}^r \right]^\frac{1}{r} \\
			& \leq \sum_{l=1}^n \Vert g_j \Vert_{L^{np}(X)}^{n-l} \widetilde{\mathbb{E}}\left[ \Vert g_j - \widetilde{\varphi}_j(\cdot) \Vert_{L^{np}(X)}^r \Vert \widetilde{\varphi}_j(\cdot) \Vert_{L^{np}(X)}^{(l-1)r} \right]^\frac{1}{r} \\
			& \leq \sum_{l=1}^n \Vert g_j \Vert_{L^{np}(X)}^{n-l} \widetilde{\mathbb{E}}\left[ \Vert g_j - \widetilde{\varphi}_j(\cdot) \Vert_{L^{np}(X)}^{lr} \right]^\frac{1}{lr} \widetilde{\mathbb{E}}\left[ \Vert \widetilde{\varphi}_j(\cdot) \Vert_{L^{np}(X)}^{lr} \right]^\frac{l-1}{lr} \\
			& \leq \widetilde{\mathbb{E}}\left[ \Vert g_j - \widetilde{\varphi}_j(\cdot) \Vert_{L^{np}(X)}^{nr} \right]^\frac{1}{nr} \sum_{l=1}^n \Vert g_j \Vert_{L^{np}(X)}^{n-l} \widetilde{\mathbb{E}}\left[ \Vert \widetilde{\varphi}_j(\cdot) \Vert_{L^{np}(X)}^{nr} \right]^\frac{l-1}{nr} \\
			& \leq \widetilde{\mathbb{E}}\left[ \Vert g_j - \widetilde{\varphi}_j(\cdot) \Vert_{L^{np}(X)}^{nr} \right]^\frac{1}{nr} \sum_{l=1}^n \Vert g_j \Vert_{L^{np}(X)}^{n-l} \left( 1 + \Vert g_j \Vert_{L^{np}(X)} \right)^{l-1} \\
			& < \frac{\varepsilon}{2 m n \left( 1 + \Vert g_j \Vert_{L^{np}(X)} \right)^{n-1}} n \left( 1 + \Vert g_j \Vert_{L^{np}(X)} \right)^{n-1} = \frac{\varepsilon}{2 m}.
		\end{aligned}
	\end{equation}
	Hence, by using \eqref{EqPropRandUATProof1} and \eqref{EqPropRandUATProof2}, we conclude for $\widetilde{\varphi} := \sum_{j=1}^m \widetilde{\varphi}_j^{\otimes n} \in \mathcal{RN}^\rho_{d,n}$ that
	\begin{equation*}
		\begin{aligned}
			\widetilde{\mathbb{E}}\left[ \left\Vert g - \widetilde{\varphi}(\cdot) \right\Vert_{L^{np}(X)^{\otimes n}}^r \right]^\frac{1}{r} & \leq \widetilde{\mathbb{E}}\left[ \left\Vert g - \sum_{j=1}^m g_j^{\otimes n} \right\Vert_{L^{np}(X)^{\otimes n}}^r \right]^\frac{1}{r} + \sum_{j=1}^m \widetilde{\mathbb{E}}\left[ \left\Vert g_j^{\otimes n} - \widetilde{\varphi}_j(\cdot)^{\otimes n} \right\Vert_{L^{np}(X)^{\otimes n}}^r \right]^\frac{1}{r} \\
			& < \frac{\varepsilon}{2} + \sum_{j=1}^m \frac{\varepsilon}{2m} = \varepsilon,
		\end{aligned}
	\end{equation*}
	which completes the proof.
\end{proof}

\begin{proof}[Proof of Theorem~\ref{ThmRandUAT}]
	Fix some $G \in L^p(\Omega,\mathcal{F}_T,\mathbb{P})$ and $\varepsilon > 0$. Then, by using Theorem~\ref{ThmChaos}, there exists some $N \in \mathbb{N}$ and $g_n \in L^{np}_{\diag}(X)^{\otimes n}$, $n = 0,...,N$, such that
	\begin{equation}
		\label{EqThmRandUATProof1}
		\left\Vert G - \sum_{n=0}^N J^\circ_n(g_n)_T \right\Vert_{L^p(\mathbb{P})} < \frac{\varepsilon}{2}.
	\end{equation}
	Now, for every $n = 1,...,N$, we apply Proposition~\ref{PropRandUAT} to obtain some $\widetilde{\varphi}_n \in \mathcal{RN}^\rho_{d,n}$ with
	\begin{equation}
		\label{EqThmRandUATProof2}
		\widetilde{\mathbb{E}}\left[ \left\Vert g_n - \widetilde{\varphi}_n(\cdot) \right\Vert_{L^{np}(X)^{\otimes n}}^r \right]^\frac{1}{r} < \frac{\varepsilon}{2 C_{n,p} N},
	\end{equation}
	where the constant $C_{n,p} > 0$ was introduced in Lemma~\ref{LemmaItIntLinearBDG}~\ref{LemmaItIntLinearBDG3}. From this, we define the map $\widetilde{\Omega} \ni \widetilde{\omega} \mapsto \sum_{n=0}^N J^\circ_n(\widetilde{\varphi}_n(\widetilde{\omega}))_T \in L^p(\Omega,\mathcal{F}_T,\mathbb{P})$ with $\widetilde{\varphi}_0 := g_0 \in \mathbb{R}$. Since $\widetilde{\Omega} \ni \widetilde{\omega} \mapsto \widetilde{\varphi}_n(\widetilde{\omega}) \in L^{np}(X)^{\otimes n}$ is by Lemma~\ref{LemmaRandNNWellDef} $\widetilde{\mathbb{P}}$-strongly measurable and $J^\circ_n(\cdot)_T: L^{np}(X)^{\otimes n} \rightarrow L^p(\Omega,\mathcal{F}_T,\mathbb{P})$ is by Lemma~\ref{LemmaItIntLinearBDG}~\ref{LemmaItIntLinearBDG3} linear and bounded, thus continuous and measurable, it follows from \cite[Corollary~1.1.11]{hytoenen16} that $\widetilde{\Omega} \ni \widetilde{\omega} \mapsto \sum_{n=0}^N J^\circ_n(\widetilde{\varphi}_n(\widetilde{\omega}))_T \in L^p(\Omega,\mathcal{F}_T,\mathbb{P})$ is also $\widetilde{\mathbb{P}}$-strongly measurable. Hence, by using Minkowksi's inequality, Lemma~\ref{LemmaItIntLinearBDG}~\ref{LemmaItIntLinearBDG3}, and the inequalities \eqref{EqThmRandUATProof1}+\eqref{EqThmRandUATProof2}, it follows that
	\begin{equation*}
		\begin{aligned}
			\widetilde{\mathbb{E}}\left[ \left\Vert G - \sum_{n=0}^N J^\circ_n(\widetilde{\varphi}_n(\cdot))_T \right\Vert_{L^p(\mathbb{P})}^r \right]^\frac{1}{r} & \leq \widetilde{\mathbb{E}}\left[ \left( \left\Vert G - \sum_{n=0}^N J^\circ_n(g_n)_T \right\Vert_{L^p(\mathbb{P})} + \sum_{n=0}^N \left\Vert J^\circ_n\left( g_n - \widetilde{\varphi}_n(\cdot) \right)_T \right\Vert_{L^p(\mathbb{P})} \right)^r \right]^\frac{1}{r} \\
			& \leq \left\Vert G - \sum_{n=0}^N J^\circ_n(g_n)_T \right\Vert_{L^p(\mathbb{P})} + \sum_{n=1}^N \widetilde{\mathbb{E}}\left[ \left\Vert J^\circ_n\left( g_n - \widetilde{\varphi}_n(\cdot) \right)_T \right\Vert_{L^p(\mathbb{P})}^r \right]^\frac{1}{r} \\
			& \leq \left\Vert G - \sum_{n=0}^N J^\circ_n(g_n)_T \right\Vert_{L^p(\mathbb{P})} + \sum_{n=1}^N C_{n,p} \widetilde{\mathbb{E}}\left[ \left\Vert g_n - \widetilde{\varphi}_n(\cdot) \right\Vert_{L^{np}(X)^{\otimes n}}^r \right]^\frac{1}{r} \\
			& < \frac{\varepsilon}{2} + \sum_{n=1}^N C_{n,p} \frac{\varepsilon}{2 C_{n,p} N} = \varepsilon,
		\end{aligned}
	\end{equation*}
	which completes the proof.
\end{proof}

\subsubsection{Proof of results in Section~\ref{SecLpHedgingRN}}
\label{SecProofsLpHedgingRN}

\begin{proof}[Proof of Theorem~\ref{ThmLpHedgingNN}]
	Fix $G \in L^p(\Omega,\mathcal{F}_T,\mathbb{P})$ and $\varepsilon > 0$. Then, by Theorem~\ref{ThmLpHedging}, there exist some $N \in \mathbb{N}$ and $g_0 \in \mathbb{R}$ as well as $m_n \in \mathbb{N}$ and $g_{n,j,0},g_{n,j,1} \in L^{np}(X)$, $n = 1,...,N$ and $j = 1,...,m_n$, with
	\begin{equation}
		\label{EqThmLpHedgingNNProof1}
		\left\Vert G - g_0 - \int_0^T \left( \vartheta^{g_{1:N}}_t \right)^\top dX_t \right\Vert_{L^p(\mathbb{P})} \leq \frac{\varepsilon}{2} + \inf_{(c,\theta) \in \mathbb{R} \times \Theta^p(X)} \left\Vert G - c - \int_0^T \theta_t^\top dX_t \right\Vert_{L^p(\mathbb{P})},
	\end{equation}
	where $\vartheta^{g_{1:N}}_t := \sum_{n=1}^N \sum_{j=1}^{m_n} \frac{W(g_{n,j,0})_t^{n-1}}{(n-1)!} g_{n,j,1}(t)$, for $t \in [0,T]$. Moreover, for every $n = 2,...,N$ and $j = 1,...,m_n$, we follow the proof of Proposition~\ref{PropUAT} to obtain some $\varphi_{n,j,0} \in \mathcal{NN}^\rho_{d,1}$ such that
	\begin{equation}
		\label{EqThmLpHedgingNNProof2}
		\left\Vert g_{n,j,0}^{\otimes (n-1)} - \varphi_{n,j,0}^{\otimes (n-1)} \right\Vert_{L^{np}(X)^{\otimes (n-1)}} < \frac{\varepsilon}{4 C_{n,p} N m_n \left( 1 + \Vert g_{n,j,1} \Vert_{L^{np}(X)} \right)},
	\end{equation}
	where $C_{n,p} > 0$ is the constant from Lemma~\ref{LemmaItIntLinearBDG}~\ref{LemmaItIntLinearBDG3}. In addition, for every $n = 1,...,N$ and $j = 1,...,m_n$, we apply Lemma~\ref{LemmaUAT} to obtain some $\varphi_{n,j,1} \in \mathcal{NN}^\rho_{d,1}$ such that
	\begin{equation}
		\label{EqThmLpHedgingNNProof3}
		\left\Vert g_{n,j,1} - \varphi_{n,j,1} \right\Vert_{L^{np}(X)} < \frac{\varepsilon}{4 C_{n,p} N m_n \left( 1 + \big\Vert \varphi_{n,j,0}^{\otimes (n-1)} \big\Vert_{L^{np}(X)^{\otimes (n-1)}} \right)},
	\end{equation}
	where $\varphi_{n,j,0}^{\otimes 0} := 1$. Hence, by following the proof of Lemma~\ref{LemmaItIntLinearBDG}~\ref{LemmaItIntLinearBDG3} (with Minkowski's inequality, the BDG inequality (with constant $C_p > 0$), and H\"older's inequality), using Lemma~\ref{LemmaItIntLinearBDG}~\ref{LemmaItIntLinearBDG3} with $n-1$ and exponent $\frac{np}{n-1}$, and the inequality~\eqref{EqThmLpHedgingNNProof2}, it follows for every $n = 2,...,N$ and $j = 1,...,m_n$ that
	\begin{equation}
		\label{EqThmLpHedgingNNProof4}
		\begin{aligned}
			& \mathbb{E}\left[ \left\vert \int_0^T J^\circ_{n-1}\left( g_{n,j,0}^{\otimes (n-1)} - \varphi_{n,j,0}^{\otimes (n-1)} \right)_t g_{n,j,1}(t)^\top dX_t \right\vert^p \right]^\frac{1}{p} \\
			& \quad\quad \leq \mathbb{E}\left[ \sup_{t \in [0,T]} \left\vert J^\circ_{n-1}\left( g_{n,j,0}^{\otimes (n-1)} - \varphi_{n,j,0}^{\otimes (n-1)} \right)_t \right\vert^\frac{np}{n-1} \right]^\frac{n-1}{np} \mathbb{E}\left[ \left( \int_0^T \left\vert g_{n,j,1}(t)^\top dA_t \right\vert \right)^{np} \right]^\frac{1}{np} \\
			& \quad\quad\quad\quad + C_p \mathbb{E}\left[ \sup_{t \in [0,T]} \left\vert J^\circ_{n-1}\left( g_{n,j,0}^{\otimes (n-1)} - \varphi_{n,j,0}^{\otimes (n-1)} \right)_t \right\vert^\frac{np}{n-1} \right]^\frac{n-1}{np} \mathbb{E}\left[ \left( \int_0^T g_{n,j,1}(t)^\top d\langle M \rangle_t g_{n,j,1}(t) \right)^\frac{np}{2} \right]^\frac{1}{np} \\
			& \quad\quad \leq \max(1,C_p) C_{n-1,np/(n-1)} \left\Vert g_{n,j,0}^{\otimes (n-1)} - \varphi_{n,j,0}^{\otimes (n-1)} \right\Vert_{L^{np}(X)^{\otimes (n-1)}} \Vert g_{n,j,1} \Vert_{L^{np}(X)} \\
			& \quad\quad < C_{n,p} \frac{\varepsilon}{4 C_{n,p} N m_n \left( 1 + \Vert g_{n,j,1} \Vert_{L^{np}(X)} \right)} \Vert g_{n,j,1} \Vert_{L^{np}(X)} \leq \frac{\varepsilon}{4 N m_n}.
		\end{aligned}
	\end{equation}
	Furthermore, by using similar arguments as in \eqref{EqThmLpHedgingNNProof4} and by inserting the inequality~\eqref{EqThmLpHedgingNNProof3}, we conclude for every $n = 2,...,N$ and $j = 1,...,m_n$ that
	\begin{equation}
		\label{EqThmLpHedgingNNProof5}
		\begin{aligned}
			& \mathbb{E}\left[ \left\vert \int_0^T J^\circ_{n-1}\left( \varphi_{n,j,0}^{\otimes (n-1)} \right)_t (g_{n,j,1}(t) - \varphi_{n,j,1}(t))^\top dX_t \right\vert^p \right]^\frac{1}{p} \\
			& \quad\quad \leq \mathbb{E}\left[ \sup_{t \in [0,T]} \left\vert J^\circ_{n-1}\left( \varphi_{n,j,0}^{\otimes (n-1)} \right)_t \right\vert^\frac{np}{n-1} \right]^\frac{n-1}{np} \mathbb{E}\left[ \left\vert \int_0^T (g_{n,j,1}(t) - \varphi_{n,j,1}(t))^\top dA_t \right\vert^{np} \right]^\frac{1}{np} \\
			& \quad\quad\quad\quad + C_p \mathbb{E}\left[ \sup_{t \in [0,T]} \left\vert J^\circ_{n-1}\left( \varphi_{n,j,0}^{\otimes (n-1)} \right)_t \right\vert^\frac{np}{n-1} \right]^\frac{n-1}{np} \mathbb{E}\left[ \left\vert \int_0^T (g_{n,j,1}(t) - \varphi_{n,j,1}(t))^\top dM_t \right\vert^{np} \right]^\frac{1}{np} \\
			& \quad\quad \leq \max(1,C_p) C_{n-1,np/(n-1)} \left\Vert \varphi_{n,j,0}^{\otimes (n-1)} \right\Vert_{L^{np}(X)^{\otimes (n-1)}} \Vert g_{n,j,1} - \varphi_{n,j,1} \Vert_{L^{np}(X)} \\
			& \quad\quad < C_{n,p} \left\Vert \varphi_{n,j,0}^{\otimes (n-1)} \right\Vert_{L^{np}(X)^{\otimes (n-1)}} \frac{\varepsilon}{4 C_{n,p} N m_n \left( 1 + \big\Vert \varphi_{n,j,0}^{\otimes (n-1)} \big\Vert_{L^{np}(X)^{\otimes (n-1)}} \right)} \leq \frac{\varepsilon}{4 N m_n}.
		\end{aligned}
	\end{equation}
	Thus, by using $\varphi_0 := g_0 \in \mathbb{R}$, the inequality~\eqref{EqThmLpHedgingNNProof1}, Minkowski's inequality together with Proposition~\ref{PropMon}, and Lemma~\ref{LemmaBDG} (if $n = 1$) or the inequalities~\eqref{EqThmLpHedgingNNProof4}+\eqref{EqThmLpHedgingNNProof5} (if $n \geq 2$), it follows that
	\begin{equation*}
		\begin{aligned}
			& \left\Vert G - \varphi_0 - \int_0^T \left( \vartheta^{\varphi_{1:N}}_t \right)^\top dX_t \right\Vert_{L^p(\mathbb{P})} - \inf_{(c,\theta) \in \mathbb{R} \times \Theta^p(X)} \left\Vert G - c - \int_0^T \theta_t^\top dX_t \right\Vert_{L^p(\mathbb{P})} \\
			& \leq \left\Vert \int_0^T \left( \vartheta^{g_{1:N}}_t \right)^\top dX_t - \int_0^T \left( \vartheta^{\varphi_{1:N}}_t \right)^\top dX_t \right\Vert_{L^p(\mathbb{P})} \\
			& \quad\quad + \left\Vert G - g_0 - \int_0^T \left( \vartheta^{g_{1:N}}_t \right)^\top dX_t \right\Vert_{L^p(\mathbb{P})} - \inf_{(c,\theta) \in \mathbb{R} \times \Theta^p(X)} \left\Vert G - c - \int_0^T \theta_t^\top dX_t \right\Vert_{L^p(\mathbb{P})} \\
			& \leq \mathbb{E}\left[ \left\vert \sum_{n=1}^N \sum_{j=1}^{m_n} \int_0^T \frac{W(g_{n,j,0})_t^{n-1}}{(n-1)!} g_{n,j,1}(t)^\top dX_t - \sum_{n=1}^N \sum_{j=1}^{m_n} \int_0^T \frac{W(\varphi_{n,j,0})_t^{n-1}}{(n-1)!} \varphi_{n,j,1}(t)^\top dX_t \right\vert^p \right]^\frac{1}{p} + \frac{\varepsilon}{2} \\
			& \leq \sum_{n=1}^N \sum_{j=1}^{m_n} \mathbb{E}\left[ \left\vert \int_0^T J^\circ_{n-1}\left( g_{n,j,0}^{\otimes (n-1)} - \varphi_{n,j,0}^{\otimes (n-1)} \right)_t g_{n,j,1}(t)^\top dX_t \right\vert^p \right]^\frac{1}{p} \\
			& \quad\quad + \sum_{n=1}^N \sum_{j=1}^{m_n} \mathbb{E}\left[ \left\vert \int_0^T J^\circ_{n-1}\left( \varphi_{n,j,0}^{\otimes (n-1)} \right)_t (g_{n,j,1}(t) - \varphi_{n,j,1}(t))^\top dX_t \right\vert^p \right]^\frac{1}{p} + \frac{\varepsilon}{2} \\
			& < \sum_{n=1}^N m_n \frac{\varepsilon}{4 N m_n} + \sum_{n=1}^N m_n \frac{\varepsilon}{4 N m_n} + \frac{\varepsilon}{2} = \varepsilon,
		\end{aligned}
	\end{equation*}
	which completes the proof.
\end{proof}

\begin{proof}[Proof of Theorem~\ref{ThmLpHedgingRN}]
	Fix $G \in L^p(\Omega,\mathcal{F}_T,\mathbb{P})$ and $\varepsilon > 0$. Then, by Theorem~\ref{ThmLpHedging}, there exist some $N \in \mathbb{N}$ and $g_0 \in \mathbb{R}$ as well as $m_n \in \mathbb{N}$ and $g_{n,j,0},g_{n,j,1} \in L^{np}(X)$, $n = 1,...,N$ and $j = 1,...,m_n$, with
	\begin{equation}
		\label{EqThmLpHedgingRNProof1}
		\left\Vert G - g_0 - \int_0^T \left( \vartheta^{g_{1:N}}_t \right)^\top dX_t \right\Vert_{L^p(\mathbb{P})} \leq \frac{\varepsilon}{2} + \inf_{(c,\theta) \in \mathbb{R} \times \Theta^p(X)} \left\Vert G - c - \int_0^T \theta_t^\top dX_t \right\Vert_{L^p(\mathbb{P})},
	\end{equation}
	where $\vartheta^{g_{1:N}}_t := \sum_{n=1}^N \sum_{j=1}^{m_n} \frac{W(g_{n,j,0})_t^{n-1}}{(n-1)!} g_{n,j,1}(t)$, for $t \in [0,T]$. Moreover, for every $n = 2,...,N$ and $j = 1,...,m_n$, we follow the proof of Proposition~\ref{PropUAT} to obtain some $\widetilde{\varphi}_{n,j,0} \in \mathcal{RN}^\rho_{d,1}$ such that
	\begin{equation}
		\label{EqThmLpHedgingRNProof2}
		\widetilde{\mathbb{E}}\left[ \left\Vert g_{n,j,0}^{\otimes (n-1)} - \widetilde{\varphi}_{n,j,0}(\cdot)^{\otimes (n-1)} \right\Vert_{L^{np}(X)^{\otimes (n-1)}}^r \right]^\frac{1}{r} < \frac{\varepsilon}{4 C_{n,p} N m_n \left( 1 + \Vert g_{n,j,1} \Vert_{L^{np}(X)} \right)},
	\end{equation}
	where $C_{n,p} > 0$ is the constant from Lemma~\ref{LemmaItIntLinearBDG}~\ref{LemmaItIntLinearBDG3}. In addition, for every $n = 1,...,N$ and $j = 1,...,m_n$, we apply Lemma~\ref{LemmaUAT} to obtain some $\widetilde{\varphi}_{n,j,1} \in \mathcal{RN}^\rho_{d,1}$ such that
	\begin{equation}
		\label{EqThmLpHedgingRNProof3}
		\widetilde{\mathbb{E}}\left[ \left\Vert g_{n,j,1} - \widetilde{\varphi}_{n,j,1}(\cdot) \right\Vert_{L^{np}(X)}^{nr} \right]^\frac{1}{nr} < \frac{\varepsilon}{4 C_{n,p} N m_n \left( 1 + \mathbb{E}\Big[ \big\Vert \widetilde{\varphi}_{n,j,0}(\cdot)^{\otimes (n-1)} \big\Vert_{L^{np}(X)^{\otimes (n-1)}}^\frac{nr}{n-1} \Big]^\frac{n-1}{nr} \right)},
	\end{equation}
	where $\widetilde{\varphi}_{n,j,0}(\cdot)^{\otimes 0} := 1$. Hence, we can follow inequality~\eqref{EqThmLpHedgingNNProof4} to conclude for every $n = 2,...,N$ and $j = 1,...,m_n$ that
	\begin{equation}
		\label{EqThmLpHedgingRNProof4}
		\begin{aligned}
			& \widetilde{\mathbb{E}}\left[ \left\Vert \int_0^T J^\circ_{n-1}\left( g_{n,j,0}^{\otimes (n-1)} - \widetilde{\varphi}_{n,j,0}(\cdot)^{\otimes (n-1)} \right)_t g_{n,j,1}(t)^\top dX_t \right\Vert_{L^p(\mathbb{P})}^r \right]^\frac{1}{r} \\
			& \quad\quad \leq C_{n,p} \widetilde{\mathbb{E}}\left[ \left\Vert g_{n,j,0}^{\otimes (n-1)} - \widetilde{\varphi}_{n,j,0}(\cdot)^{\otimes (n-1)} \right\Vert_{L^{np}(X)^{\otimes (n-1)}}^r \right]^\frac{1}{r} \Vert g_{n,j,1} \Vert_{L^{np}(X)} < \frac{\varepsilon}{4 N m_n}.
		\end{aligned}
	\end{equation}
	Furthermore, we can follow the inequality~\eqref{EqThmLpHedgingNNProof5} and apply H\"older's inequality to observe for every $n = 2,...,N$ and $j = 1,...,m_n$ that
	\begin{equation}
		\label{EqThmLpHedgingRNProof5}
		\begin{aligned}
			& \widetilde{\mathbb{E}}\left[ \left\Vert \int_0^T J^\circ_{n-1}\left( \widetilde{\varphi}_{n,j,0}(\cdot)^{\otimes (n-1)} \right)_t (g_{n,j,1}(t) - \widetilde{\varphi}_{n,j,1}(\cdot)(t))^\top dX_t \right\Vert_{L^p(\mathbb{P})}^r \right]^\frac{1}{p} \\
			& \quad\quad \leq C_{n,p} \widetilde{\mathbb{E}}\left[ \left\Vert \widetilde{\varphi}_{n,j,0}(\cdot)^{\otimes (n-1)} \right\Vert_{L^{np}(X)^{\otimes (n-1)}}^r \Vert g_{n,j,1} - \widetilde{\varphi}_{n,j,1}(\cdot) \Vert_{L^{np}(X)}^r \right]^\frac{1}{r} \\
			& \quad\quad \leq C_{n,p} \widetilde{\mathbb{E}}\left[ \left\Vert \widetilde{\varphi}_{n,j,0}(\cdot)^{\otimes (n-1)} \right\Vert_{L^{np}(X)^{\otimes (n-1)}}^\frac{nr}{n-1} \right]^\frac{n-1}{nr} \widetilde{\mathbb{E}}\left[ \Vert g_{n,j,1} - \widetilde{\varphi}_{n,j,1}(\cdot) \Vert_{L^{np}(X)}^{nr} \right]^\frac{1}{nr} < \frac{\varepsilon}{4 N m_n}.
		\end{aligned}
	\end{equation}
	Thus, by using $\widetilde{\varphi}_0 := g_0 \in \mathbb{R}$, the inequality~\eqref{EqThmLpHedgingRNProof1}, Minkowski's inequality together with Proposition~\ref{PropMon}, and Lemma~\ref{LemmaBDG} (if $n = 1$) or the inequalities~\eqref{EqThmLpHedgingRNProof4}+\eqref{EqThmLpHedgingRNProof5} (if $n \geq 2$), it follows that
	\begin{equation*}
		\begin{aligned}
			& \widetilde{\mathbb{E}}\left[ \left\Vert G - \varphi_0 - \int_0^T \left( \vartheta^{\widetilde{\varphi}_{1:N}(\cdot)}_t \right)^\top dX_t \right\Vert_{L^p(\mathbb{P})}^r \right]^\frac{1}{r} - \inf_{(c,\theta) \in \mathbb{R} \times \Theta^p(X)} \left\Vert G - c - \int_0^T \theta_t^\top dX_t \right\Vert_{L^p(\mathbb{P})} \\
			& \leq \widetilde{\mathbb{E}}\left[ \left\Vert \int_0^T \left( \vartheta^{g_{1:N}}_t \right)^\top dX_t - \int_0^T \left( \vartheta^{\widetilde{\varphi}_{1:N}(\cdot)}_t \right)^\top dX_t \right\Vert_{L^p(\mathbb{P})}^r \right]^\frac{1}{r} \\
			& \quad\quad + \left\Vert G - g_0 - \int_0^T \left( \vartheta^{g_{1:N}}_t \right)^\top dX_t \right\Vert_{L^p(\mathbb{P})} - \inf_{(c,\theta) \in \mathbb{R} \times \Theta^p(X)} \left\Vert G - c - \int_0^T \theta_t^\top dX_t \right\Vert_{L^p(\mathbb{P})} \\
			& \leq \widetilde{\mathbb{E}}\left[ \left\Vert \sum_{n=1}^N \sum_{j=1}^{m_n} \int_0^T \frac{W(g_{n,j,0})^{n-1}}{(n-1)!} g_{n,j,1}(t)^\top dX_t - \sum_{n=1}^N \sum_{j=1}^{m_n} \int_0^T \frac{W(\widetilde{\varphi}_{n,j,0}(\cdot))^{n-1}}{(n-1)!} \widetilde{\varphi}_{n,j,1}(\cdot)(t)^\top dX_t \right\Vert_{L^p(\mathbb{P})}^r \right]^\frac{1}{r} \\
			& \quad\quad + \frac{\varepsilon}{2} \\
			& \leq \sum_{n=1}^N \sum_{j=1}^{m_n} \widetilde{\mathbb{E}}\left[ \left\Vert \int_0^T J^\circ_{n-1}\left( g_{n,j,0}^{\otimes (n-1)} - \widetilde{\varphi}_{n,j,0}(\cdot)^{\otimes (n-1)} \right)_t g_{n,j,1}(t)^\top dX_t \right\Vert_{L^p(\mathbb{P})}^r \right]^\frac{1}{r} \\
			& \quad\quad + \sum_{n=1}^N \sum_{j=1}^{m_n} \widetilde{\mathbb{E}}\left[ \left\Vert \int_0^T J^\circ_{n-1}\left( \widetilde{\varphi}_{n,j,0}(\cdot)_{n,j}^{\otimes (n-1)} \right)_t (g_{n,j,1}(t) - \widetilde{\varphi}_{n,j,1}(\cdot)(t))^\top dX_t \right\Vert_{L^p(\mathbb{P})}^r \right]^\frac{1}{r} + \frac{\varepsilon}{2} \\
			& < \sum_{n=1}^N m_n \frac{\varepsilon}{4 N m_n} + \sum_{n=1}^N m_n \frac{\varepsilon}{4 N m_n} + \frac{\varepsilon}{2} = \varepsilon,
		\end{aligned}
	\end{equation*}
	which completes the proof.
\end{proof}

\bibliographystyle{plain}
\bibliography{mybib}

\end{document}